\documentclass[12pt]{article}
\usepackage{fullpage}
\usepackage[margin=1.5cm]{geometry}
\usepackage{csquotes}
\usepackage[T1]{fontenc}    
\usepackage{multirow}
\usepackage{tabularx}
\usepackage{enumitem}
\usepackage{braket}
\usepackage{amsmath}
\usepackage{physics}
\usepackage{hyperref}
\usepackage{titling}
\usepackage{graphicx}
\usepackage{amsmath,amsfonts,amssymb}
\usepackage{physics}
\usepackage{amsthm}
\usepackage{tikz}
\usepackage{pgfplots}
\usepackage{mathrsfs}
\usepackage{musicography}
\usepackage{esint}
\usepackage{hieroglf}
\usepackage{bbm}

\numberwithin{equation}{section}

\newtheorem{theorem}{Theorem}[section]
\newtheorem{corollary}[theorem]{Corollary}
\newtheorem{lemma}[theorem]{Lemma}
\newtheorem{proposition}[theorem]{Proposition}

 { \theoremstyle{definition}
\newtheorem{definition}[theorem]{Definition}
\newtheorem{remark}[theorem]{Remark} }

\makeatletter
\def\widebreve{\mathpalette\wide@breve}
\def\wide@breve#1#2{\sbox\z@{$#1#2$}%
     \mathop{\vbox{\m@th\ialign{##\crcr
\kern0.08em\brevefill#1{0.8\wd\z@}\crcr\noalign{\nointerlineskip}%
                    $\hss#1#2\hss$\crcr}}}\limits}
\def\brevefill#1#2{$\m@th\sbox\tw@{$#1($}%
  \hss\resizebox{#2}{\wd\tw@}{\rotatebox[origin=c]{90}{\upshape(}}\hss$}
\makeatletter

\makeatletter
\usepackage{pict2e,picture}
\pgfkeys{/csteps/inner ysep/.initial=4pt,
    /csteps/inner xsep/.initial=4pt,
    /csteps/inner color/.initial=red,
    /csteps/outer color/.initial=blue,
}
\newsavebox\csteps@CBox
\newlength\csteps@XLength \newlength\csteps@YLength \newlength\csteps@YDepth \newlength\csteps@tmplen
\def\csteps@CircledParam#1#2{\sbox\csteps@CBox{#2}%
    \csteps@XLength=\wd\csteps@CBox\advance\csteps@XLength by\pgfkeysvalueof{/csteps/inner xsep}\relax
    \csteps@tmplen=\pgfkeysvalueof{/csteps/inner ysep}\relax
    \csteps@YDepth=\dp\csteps@CBox\advance\csteps@YDepth by 0.5\csteps@tmplen\relax
    \csteps@YLength=\ht\csteps@CBox\advance\csteps@YLength by\dp\csteps@CBox\advance\csteps@YLength by\pgfkeysvalueof{/csteps/inner ysep}\relax
    \typeout{DBG:#2\space X\space\the\csteps@XLength\space Y:\the\csteps@YLength\space D:\the\csteps@YDepth}%
    \raisebox{-#1\csteps@YDepth}{%
    \ifdim\csteps@XLength>\csteps@YLength
    \makebox[\csteps@XLength]{
        \makebox(0,\csteps@YLength){%
            \color{\pgfkeysvalueof{/csteps/outer color}}\put(0,0){\oval(\csteps@XLength,\csteps@YLength)}%
        }%
    \makebox(0,\csteps@YLength){%
        \put(-.5\wd\csteps@CBox,0){\textcolor{\pgfkeysvalueof{/csteps/inner color}}{#2}}%
    }}%
    \else
    \makebox[\csteps@YLength]{%
        \makebox(0,\csteps@YLength){%
            \color{\pgfkeysvalueof{/csteps/outer color}}\put(0,0){\circle{\csteps@YLength}}%
        }%
    \makebox(0,\csteps@YLength){%
        \put(-.5\wd\csteps@CBox,0){\textcolor{\pgfkeysvalueof{/csteps/inner color}}{#2}}%
     }}%
    \fi
    }%
}
\def\Circled#1{\csteps@CircledParam{1}{#1}}
\def\CircledTop#1{\csteps@CircledParam{0}{#1}}
\makeatother

\tikzset{/csteps/inner ysep=10pt}
\tikzset{/csteps/inner xsep=10pt}

\newcommand{\f}[2]{{\ensuremath{%
    \mathchoice%
    {\dfrac{#1}{#2}}
    {\dfrac{#1}{#2}}
    {\frac{#1}{#2}}
    {\frac{#1}{#2}}
}}}

\def\a{\alpha}

\def\be{\beta}
\def\ga{\gamma}
\def\Ga{\Gamma}

\def\de{\delta}

\def\De{\Delta}
\def\eps{\epsilon}

\def\la{\lambda}
\def\La{\Lambda}

\def\sg{\sigma}
\def\vsg{\varsigma}
\def\Sg{\Sigma}

\def\Ups{\Upsilon}

\def\th{\theta}
\def\Th{\Theta}

\def\Om{\Omega}
\def\om{\omega}
\def\vp{\varphi}

\newcommand{\mc}[1]{\ensuremath{\mathcal{#1}}}
\newcommand{\mf}[1]{\ensuremath{\mathfrak{#1}}}
\newcommand{\msc}[1]{\ensuremath{\mathscr{#1}}}

\newcommand{\bs}[1]{\ensuremath{\boldsymbol{#1}}}

\def \i{ \mathrm i}

\newcommand{\ov}[1]{\ensuremath{\overline{#1}}}
\newcommand{\wt}[1]{\ensuremath{\widetilde{#1}}}
\newcommand{\wh}[1]{\ensuremath{\widehat{#1}}}

\newcommand{\Int}[2]{\ensuremath{\int\limits_{#1}^{#2}}}
\newcommand{\Oint}[2]{\ensuremath{\oint\limits_{#1}^{#2}}}
\newcommand{\Fint}[2]{\ensuremath{\fint\limits_{#1}^{#2}}}

\newcommand{\sul}[2]{\ensuremath{\sum\limits_{#1}^{#2}}}
\newcommand{\pl}[2]{\ensuremath{\prod\limits_{#1}^{#2}}}

\newcommand{\R}{\ensuremath{\mathbb{R}}}
\newcommand{\Cx}{\ensuremath{\mathbb{C}}}

\newcommand{\Dp}[1]{\ensuremath{\partial_{#1}}}

\newcommand{\ex}[1]{\ensuremath{\e{e}^{#1}}}

\renewcommand{\op}[1]{ \boldsymbol{ \texttt{#1} } }

\renewcommand{\dd}{\mathrm{d}}
\newcommand{\e}[1]{\ensuremath{\mathrm{#1}}}

\newcommand{\intff}[2]{\ensuremath{ [  #1 \,; #2 ] }}
\newcommand{\intfo}[2]{\ensuremath{ [  #1 \,; #2 [ }}
\newcommand{\intof}[2]{\ensuremath{ ]  #1 \,; #2 ] }}
\newcommand{\intoo}[2]{\ensuremath{ ]  #1 \,; #2 [ }}

\newcommand{\intn}[2]{\ensuremath{[\![ \, #1 \,;\, #2 \,]\!]}}

\newcommand\beq{\begin{equation}}
\newcommand\enq{\end{equation}}
\newcommand\bem{\begin{multline}}
\newcommand\enm{\end{multline}}

\def\beqa{\begin{eqnarray}}
\def\eeqa{\end{eqnarray}}
\def\ba{\begin{array}}
\def\ea{\end{array}}
\def\det{\operatorname{det}}

\begin{document}
\title{Asymptotic expansion of the partition function for $\beta$-ensembles with complex potentials}

\author{Alice Guionnet\thanks{ Unité de Mathématiques Pures et Appliquées, Electronic address: \texttt{alice.guionnet@ens-lyon.fr}}, Karol Kozlowski\thanks{Laboratoire de Physique, Electronic address: \texttt{karol.kozlowski@ens-lyon.fr}}  and Alex Little\thanks{Unité de Mathématiques Pures et Appliquées,Electronic address: \texttt{alexander.little@ens-lyon.fr}} \\
Univ Lyon, ENS de Lyon, Univ Claude Bernard Lyon 1, \\ CNRS, F-69342 Lyon, France}

\maketitle

\abstract{In this work we establish under certain hypotheses the $N \to +\infty$ asymptotic expansion  of integrals of the form
$$\mathcal{Z}_{N,\Gamma}[V] \, = \,  \Int{\Gamma^N}{} \pl{  a < b}{N}(z_a - z_b)^\beta \, \pl{k=1}{N} \ex{ - N \beta V(z_k) } \, \dd\mathbf{z}$$
where $V \in \Cx[X]$,  $\beta \in 2 \mathbb{N}^*$ is an even integer and $\Gamma \subset \mathbb{C}$ is an unbounded contour such that the integral converges.
For even degree, real valued $V$s and when $\Ga = \R$, it is well known that the large-$N$ expansion is characterised by an equilibrium measure
corresponding to the minimiser of an appropriate energy functional. This method bears a structural resemblance with the Laplace method.
By contrast, in the complex valued setting we are considering, the analysis structurally resembles the classical steepest-descent method,  and involves finding a
critical point \textit{and} a steepest descent curve, the latter being a deformation of the original integration contour.  More precisely,  one minimises a curve-dependent energy functional with respect to measures on the curve and then maximises the energy over an appropriate space of curves.  Our analysis deals with the one-cut regime of the associated equilibrium measure.  We establish the existence of an all order asymptotic
expansion for $\ln \mathcal{Z}_{N,\Gamma}[V]$ and explicitly identify the first few terms.}

\newpage

\tableofcontents

\newpage

\section{Introduction}\label{introduction}

In this work we are interested in the large $N$ behaviour of integrals of the form
\begin{equation}\label{Cpartitionfunction}
\mathcal{Z}_{N,\Gamma}[V] \, = \,  \Int{\Gamma^N}{} \pl{  a < b}{N}(z_a - z_b)^\beta \, \pl{k=1}{N} \ex{ - N \beta V(z_k) } \, \dd\mathbf{z}
\end{equation}
where $V\in \Cx[X]$ has degree $\kappa \geq 2$ and $\beta \in 2\mathbb{N}^*$ is a positive  even integer.  The restriction on $\beta$ produces an analytic integrand allowing us to deform the integration curve in \eqref{Cpartitionfunction} by Cauchy's theorem; this plays a crucial role in our analysis.
Note that the case of odd $\beta$ is trivial as it would produce an antisymmetric integrand, making the integral \eqref{Cpartitionfunction} $0$.
By deforming the original contour and rescaling the integration variables we can, without loss of generality,  focus on the setting where
\beq
V(z) = \frac{z^\kappa}{\kappa}+\e{O}(z^{\kappa-1}) \,
\label{definition pot polynomial interet}
\enq
and $\Gamma$ is a simple contour smoothly homeomorphic to $\mathbb{R}$ and going to infinity along the directions $\infty \cdot  \ex{\frac{2\pi \i \alpha}{\kappa}}$ and
$\infty \cdot \ex{\frac{2\pi \i \alpha^\prime}{\kappa}}$ for distinct integers $\alpha, \alpha^\prime \in  \intn{0}{ \kappa-1 }$. These correspond to the incoming and outgoing directions of $\Ga$.

In the course of studying \eqref{Cpartitionfunction} we shall need to consider its real-valued cousin expressed in terms of a smooth homeomorphism  $\ga : \mc{J} \subseteq \R \rightarrow \ga( \mc{J} )$
with $ \ga( \mc{J} )$ a simple curve, as
\beq
 \op{Z}_{N, \ga}[V] \; \overset{\mathrm{def}}{=} \;  \Int{ \mc{J}^N }{} \pl{ a < b }{N} |\gamma(x_a) - \gamma(x_b)|^\beta \pl{a=1}{N}\Big\{ |\gamma^\prime(x_k)|  \cdot \ex{-N \beta \varphi(\gamma(x_k))}  \Big\}
 \, \dd\mathbf{x} \quad \e{with} \quad \vp \, = \, \Re(V) \;.
\label{Rpartitionfunction}
\enq
Note that henceforth the support interval $\mc{J}$ of the homeomorphism is considered as part of the data carried by the label "$\ga$". We observe that when $\mc{J}=\R$ and $\ga$ is the smooth homeomorphism realising a bijection with the curve $\Ga$,
it holds by the triangle inequality that $ \big| \mathcal{Z}_{N,\Gamma}[V] \big| \, \leq \,  \op{Z}_{N, \ga}[V]$.
Let us also remark that the multiple integral \eqref{Rpartitionfunction} has the natural statistical mechanical interpretation as the partition function of a Dyson gas confined to the curve $\ga(\mc{J})$ at inverse temperature
$\beta = 1/T$.  By analogy with the classical steepest descent method \cite{deBruijn}, and this will be made precise in the course of this paper,  one might expect that \eqref{Cpartitionfunction} will be well-approximated by \eqref{Rpartitionfunction}
provided that the contour $\Ga=\ga(\R)$ is such that the oscillations of the complex integral are minimised. In fact, since one can optimise with respect to the original
contour, we always have the trivial bound
\begin{align}\label{trianglebound}
 |\mathcal{Z}_{N,\Gamma}[V]|  \, \leq  \, \inf_{\widetilde{\Gamma}=\wt{\ga}(\R) \in \mathcal{T}}    \op{Z}_{N, \wt{\ga}}[V]
\end{align}
where $\mathcal{T}$ denotes an appropriate set of contours that preserve the value of \eqref{Cpartitionfunction}.
Experience with the steepest descent method leads one to expect that the \textit{rhs} of \eqref{trianglebound} is optimal to leading order.

The two partition functions introduced above are associated with two "models", the "real model" and the "complex model" respectively.
Given $\mathcal{J} \subset \mathbb{R}$ a closed interval and $\gamma : \mathcal{J} \longrightarrow \mathbb{C}$ a smooth simple curve,
the "real model" corresponds to the probability measure on $\mathcal{J}^N$ with density
\beq
\dd\mathbb{P}_{N,\ga}(\mathbf{x})  \, \overset{\mathrm{def}}{=}  \,
\frac{1}{ \op{Z}_{N, \ga}[V] }\pl{ a < b }{N} |\gamma(x_a) - \gamma(x_b)|^\beta \pl{a=1}{N}\Big\{ |\gamma^\prime(x_k)|  \cdot \ex{-N \beta \varphi(\gamma(x_k))}  \Big\}
 \, \dd\mathbf{x}, \qquad  \mathbf{x}\in \mathcal{J}^N \;.
\label{realmodel}
\enq
Further, given  a smooth simple curve $\Sigma \subset \mathbb{C}$, we define the "complex model" as the complex-valued "measure" on $\Sigma^N$ having mass 1 with density
in respect to the natural measure $\dd \mathbf{z}$ on $\Sigma^N$ as
\beq
\dd \Pi_{N, \Sg}(\mathbf{z})   \overset{\mathrm{def}}{=} \frac{1}{ \mathcal{Z}_{N,\Sigma}[V] }\pl{ a < b }{ N} (z_a - z_b)^\beta  \cdot \pl{a=1}{N}\ex{-N \beta  V(z_a)} \, \dd\mathbf{z}, \qquad  \mathbf{z}\in \Sigma^N
\;.  \label{complexmodel}
\enq
The normalisation factor $\mathcal{Z}_{N,\Sigma}[V]$ is just given by \eqref{Cpartitionfunction} but with the caveat that now $\Sigma$ might be compact.

We introduce these more general situations because, at several stages of our analysis, we shall need to focus on "effective" situations where
integration variables, \textit{aka} "particles", are restricted to certain compact subsets of the original curve $\Ga$.
We shall not dwell for the moment upon the precise conditions that must be imposed on $\gamma$ so that these models are well-defined.

The asymptotic analysis of \eqref{Rpartitionfunction} and hence of \eqref{Cpartitionfunction} is closely related to a certain problem in potential theory.
To see this, observe that for $\mathbf{z} \in \Gamma^N$ having pairwise distinct coordinates we have
\begin{align*}
\sul{ i < j }{ N} \ln \frac{1}{|z_i - z_j|} + (N-1) \sum_{j=1}^N \varphi(z_j)  \, =\,
\frac{N^2}{2} \Int{\substack{\Gamma^2 \\ z \neq w}}{} \left( \ln \frac{1}{|z-w|} + \varphi(z) + \varphi(w) \right)\, \dd L_N^{(\mathbf{z})}(z) \otimes \dd L_N^{(\mathbf{z})}(w)  \, .
\end{align*}
Here, we have  introduced the empirical measure
\beq
L_N^{(\mathbf{z})} = \frac{1}{N}\sum_{k=1}^N \delta_{z_k} \in \mathcal{P}(\mathbb{C})
\label{definition empirical measure}
\enq
associated to $\mathbf{z}= (z_1, \dots, z_N) \in \mathbb{C}^N$.

This observation motivates the introduction of the following "energy" functional on the set  $\mathcal{P}(\Gamma)$ of Borel probability measures on $\Gamma$.
\begin{align}
\begin{split}
&\op{I}_{\Ga} : \mathcal{P}(\Gamma) \longrightarrow (-\infty,+\infty] \\
&\op{I}_{\Ga} [\mu] \overset{\mathrm{def}}{=} \Int{\Ga^2}{} \left(    \ln \frac{1}{|z - w|} +  \varphi(z)+ \varphi(w)\right) \, \dd\mu(z) \otimes \dd\mu(w)
\end{split}
\end{align}
We shall see later on that $\op{I}_{\Ga}$ has a unique minimiser $\mu_{\Ga}$.

The real density $\dd \mathbb{P}_{N,\ga}(\mathbf{x}) / \dd \mathbf{x} $ is morally of the form
$$ \exp\bigg\{ - \frac{\beta N^2}{2} \op{I}_{\Ga}[L_N^{(\gamma(\mathbf{x}))}] \, + \, N L_N^{(\mathbf{x})}( \ln |\ga^{\prime}|) \bigg\} \, , $$
in which $\gamma(\mathbf{x}) \, = \, \big( \gamma(x_1) , \dots, \gamma(x_N) \big)$ and
\beq
L_N^{(\mathbf{z})}(f) = \frac{1}{N}\sum_{k=1}^N f(z_k)
\label{definition linear statistics}
\enq
is the \textnormal{linear statistic} associated to the measurable function $f$. Still,  this is only formal since, strictly speaking, $\op{I}_{\Ga}[L_N^{(\mathbf{z})}] = +\infty$
due to the blow up issuing from the diagonal.
Nevertheless, configurations of particles such that $L_N^{(\mathbf{z})}$ is close to this minimiser dominate the asymptotics, so that
$$ \frac{1}{N^2} \ln \op{Z}_{N,\ga }[V]  \to - \frac{\beta}{2} \inf_{\mu \in \mathcal{P}(\Gamma)} \op{I}_{\Ga}[\mu]= - \frac{\beta}{2}  \op{I}_{\Ga}[\mu_{\Ga}] \quad \e{as} \quad  N \to +\infty \, . $$
In physical terms, the partition function is dominated by the minimal energy configuration of the particles.  This is because the long-range forces $\ln\frac{1}{|z-w|}$
and the scaling in $N$ of the potential ensure that the "energy" is order $N^2$ whereas the "entropy" and the contour curvature
contribute only at order $N$, and so the former is the dominant part of the free energy.

\vspace{2mm}

For our complex partition function we expect, after optimising our bound \eqref{trianglebound}, that we have
\begin{align}\label{leadingorder}
&\limsup_{N \to +\infty}\frac{\ln |\mathcal{Z}_{N,\Gamma}[V]|}{N^2}  \leq - \frac{\beta}{2}  \sup_{\widetilde{\Gamma}\in \mathcal{T}} \inf_{\mu \in \mathcal{P}(\widetilde{\Gamma})} \op{I}_{ \wt{\Ga} }[\mu]
\end{align}
We shall show in this work that, in fact, the above upper bound is an equality.
 The \textit{rhs} of  \eqref{leadingorder} leads to the so-called \textit{max-min energy problem} which consists in searching for a curve such that the supremum in \eqref{leadingorder} is attained.
 It turns out that, at least for polynomial potentials, the supremum in \eqref{leadingorder} is indeed attained and this curve satisfies a so-called $S$-property;
a crucial result of Kuijilaars and Silva \cite{lirias1643667}. We should remark that a construction of the optimal curve for our problem
without reference to a variational problem was already given in the off-critical regime by Bertola \cite{Bertola:2011aa}.
We shall denote such a contour solution, under suitable growth at infinity, as $\Ga_{\e{eq}}$.

Before stating our main theorem, it is necessary to  introduce the notion of the complex energy functional.
\begin{definition}[Complex energy functional]\label{complexenergy}
Let $\Gamma \subset \mathbb{C}$ be a smooth, simple, unbounded contour and let $\mu \in \mathcal{P}(\Gamma)$ be a probability measure on $\Gamma$ with continuous density
with respect to arc length on $\Gamma$. Then introduce the "$g$-function" as
\begin{align*}
&g_{\Ga}[\mu] : \mathbb{C}\setminus \Gamma \longrightarrow \mathbb{C}
&g_{\Ga}[\mu](z) = - \Int{\Gamma}{} \ln (z-w) \cdot  \dd\mu(w) \;.
\end{align*}
For $w \in \Ga$, the map $z \mapsto \ln (z-w)$ is defined as having a cut along $\Ga\intff{\infty \ex{ \f{ 2\i\pi \a }{ \kappa }  }}{ w }$, the portion of
$\Ga$ located between its incoming direction and $w$. As a consequence, $g_{\Ga}[\mu]$ has cuts along $\Ga\intff{\infty \ex{ \f{ 2\i\pi \a }{ \kappa }  }}{ \mf{e}_{\mu} }$,
in which $\mf{e}_{\mu}$ is the extremal point of $\e{supp} \mu$ along the direction of $\Ga$.
$g_{\Ga}[\mu]$ has derivative $g_{\Ga}[\mu]^\prime(z) = \int_{\Gamma}^{} \frac{1}{w-z} \, \dd\mu(w)$.
Then the \textnormal{complex energy} associated to the measure $\mu \in \mathcal{P}(\Gamma)$ is
\begin{align}
\mc{I}_{\Gamma}[\mu] \overset{\mathrm{def}}{=} \int_\Gamma \left( \frac{1}{2} \big( g_{\Ga}[\mu] \big)_+(z)+\frac{1}{2} \big( g_{\Ga}[\mu] \big)_-(z)+2V(z) \right) \, \dd\mu(z)
\label{definition fnelle energie complexe}
\end{align}
where $\big( g_{\Ga}[\mu] \big)_{\pm}$ are the left and right boundary values of $g_{\Ga}[\mu] $ on the contour $\Gamma$.
\end{definition}
We remark that $\Re \big( \mc{I}_{\Gamma}[\mu] \big) \, =\,  \op{I}_{\Ga}[\mu]$ where  $\varphi = \Re V$, hence it is appropriate to view $\mc{I}_{\Gamma}$
as the natural complexification of  $\op{I}_{\Ga}[\mu]$. Note also that
$\mc{I}_{\Gamma}[\mu]$ is unaffected by modifying the contour $\Gamma$ outside of the support of $\mu$.

The central hypothesis of this paper is the following.
\begin{enumerate}[label=(V\arabic*)]
\item \label{potentialh2} The potential $V$ is a polynomial with leading term $\frac{z^\kappa}{\kappa}$ and is one-cut regular according to Definition \ref{regular}; that is, $V$ is off-critical and the support of the equilibrium measure $\mu_{\e{\Ga}_{\e{eq}}}$ associated to $V$ on the $S$-curve is connected.
\end{enumerate}
A recent work \cite{Bertola2022} has shown that the off-critical regime for the max-min energy problem is an open set within the space of polynomial potentials  and it is furthermore expected that the set of critical polynomial potentials has zero Lebesgue measure. The latter has been proven for logarithmic potential problems on the real line \cite{Mclaughlin}. Thus off-criticality is believed to characterise the typical case.

Although,  at the outset,  one may start with a partition function $\mc{Z}_{N,\Ga}[V]$ that involves contour $\Gamma$ of limited regularity, \textit{e.g.}
obtained by concatenating a finite number of $\mc{C}^1$ arcs, it is always possible to invoke Cauchy's theorem and deform $\Ga$ to a new curve that is $\mc{C}^\infty$
and behaves in a "nice" way at infinity. More precisely,  "nice" curves satisfy the following properties.

\begin{enumerate}[label=($\Gamma$\arabic*)]
\item  \label{hypothesis2}$\Gamma$ is the homeomorphic image of $\mathbb{R}$ under an infinitely differentiable injective function $\gamma$:
\begin{align*}
&\gamma : \mathbb{R} \longrightarrow \mathbb{C} \\
&\gamma(\mathbb{R}) = \Gamma.
\end{align*}
\item \label{hypothesis3} $|\gamma(x)| \to +\infty$ as $x \to \pm \infty$ and furthermore there exists integers $\alpha ,  \alpha^\prime \in \intn{0}{\kappa - 1}$,
$\alpha \neq \alpha^\prime$,  and $K > 0$ such that
\begin{align*}
\arg \gamma(x) &= \frac{2\pi \alpha}{\kappa} & & \forall x \in \intof{-\infty}{- K},\\
\arg \gamma(x) &= \frac{2\pi \alpha^\prime}{\kappa} & & \forall x \in \intfo{K}{+\infty} \, .
\end{align*}
\end{enumerate}

\begin{remark}

Note that \ref{hypothesis3} implies that
\begin{align}\label{Vinf}
\lim_{\substack{|z| \to \infty  \\  z \in \Gamma}}\frac{ V(z)}{|z|^\kappa}  = \frac{1}{\kappa}.
\end{align}
This means that the "nice" curves coincide at $\infty$ with directions of the fastest growth of the confining potential.

Moreover, \ref{hypothesis2} and \ref{hypothesis3} imply that there exists a constant $0 < C_\gamma < 1$ such that
\begin{align}\label{hypothesis4}
0 < C_\gamma \leq  \frac{|\gamma(x)-\gamma(y)|}{|x-y|} \leq C_\gamma^{-1} & & \forall x,y \in \mathbb{R}.
\end{align}
Note that (\ref{hypothesis4}) holds automatically if $x$ and $y$ are confined to a compact set so long as $\gamma$ is $C^1$ and injective.
\end{remark}

It will be proven in Appendix \ref{curvecontruction}, that the solution of the max-min energy problem
in the off-critical regime, shown to exist by \cite{lirias1643667}, can be taken to satisfy \ref{hypothesis2}-\ref{hypothesis3}.

A central tool of the present paper is the method of \textit{Dyson-Schwinger equations}, also known as loop equations or Ward identities.
These equations relate various moments of linear statistics \eqref{definition linear statistics}, with respect to either the real \eqref{realmodel} or complex measure \eqref{complexmodel}. The Dyson-Schwinger equations
involve the so-called \textit{master operator}, which, if it can be inverted on an appropriate functional space,
allows one to solve these equations asymptotically in the large-$N$ regime and hence compute the asymptotic expansion of the moments of any sufficiently regular linear statistic.
We shall establish that,  under our hypotheses,  the master operators arising in the situations of interest to us are invertible.  From the large-$N$ expansion of linear statistics one may then deduce a similar expansion for the partition function \eqref{Cpartitionfunction}.
Indeed,  if one has a family of potentials $\{ V_t \}_{t \in \intff{0}{1}}$ depending smoothly on a parameter $t \in \intff{0}{1}$ and so that $V_{1}=V$, with $V$ the polynomial potential one starts with,
and $V_0$ a quadratic potential, then one has the identity
\begin{align*}
\frac{\partial}{\partial t} \ln \mathcal{Z}_{N, \Sigma_t}[V_t] \, = \,  - \beta N^2 \mc{E}_{N, \Sg_t} \bigg[ L_N^{(\mathbf{z})}\left( \frac{\partial V_t}{\partial t}\right) \bigg] \, + \,
(\text{boundary terms})
\end{align*}
where $\{ \Sigma_t \}_{t \in [0,1]}$ is a family of compact arcs in $\mathbb{C}$, smoothly varying in $t \in \intff{0}{1}$
and $ \mc{E}_{N, \Sg_t}$ refers to the expectation with respect to the complex model \eqref{complexmodel}. We remind the reader that, in case of quadratic potentials, the partition function is given in closed form by the Selberg-Mehta formula  (see Ch. 17 of \cite{mehta}).

Let us now state our main theorem.
\begin{theorem}[Expansion of the complex partition function]\label{maintheorem}
Let $V$ satisfy hypothesis \textnormal{\ref{potentialh2}} and let $\Gamma \subset \mathbb{C}$ be an "admissible" contour in the sense of Definition \ref{admissible};
that is, a simple (unbounded) contour consisting of a finite number of $C^1$ arcs and "stretching out to infinity" in the directions $\ex{ \frac{ 2\pi \i \alpha }{ \kappa }}$ and
$\ex{\frac{2\pi \i \alpha^\prime}{\kappa}}$. Then there exists an admissible contour $\Gamma_{\e{eq}} \subset \mathbb{C}$,  a deformation of $\Gamma$,
satisfying hypotheses \textnormal{\ref{hypothesis2}-\ref{hypothesis3}} and solving the \textnormal{max-min energy problem} (see (\ref{maxmin})). Let $\mu_{ \Ga_{\e{eq}} } \in \mathcal{P}(\Ga_{\e{eq}})$
be the associated equilibrium measure for $\Ga_{\e{eq}}$ in the external field $\varphi = \Re V$ and let $ \tfrac{ \dd \mu_{ \Ga_{\e{eq}} } }{ \dd z }$ be its density with respect to the natural measure $\dd z$
on $\Ga_{\e{eq}}$, so that $\dd \mu_{ \Ga_{\e{eq}} }(z) =  \tfrac{ \dd \mu_{ \Ga_{\e{eq}} } }{ \dd z } \, \dd z$.

Then, there is a sequence of complex numbers $( F_k(\beta,V) )_{k = -2}^\infty$ such that for every $K \in \mathbb{N}$ we have
\begin{equation}\boxed{\begin{split}
\ln \mathcal{Z}_{N,\Gamma}[V] &= \frac{\beta}{2} N \ln N + \frac{3+\frac{\beta}{2}+\frac{2}{\beta}}{12} \ln N + \sum_{k=-2}^K \frac{F_k(\beta,V)}{N^k} + \e{O}(N^{-K-1}), & & N \to +\infty.
\end{split}}
\end{equation}
In particular
\begin{align*}
F_{-2}(\beta,V) &= - \frac{\beta}{2}  \mc{I}_{  \Ga_{\e{eq}} }[\mu_{ \Ga_{\e{eq}} }], \\
F_{-1}(\beta,V) &= \left(\frac{\beta}{2} - 1\right) \left[ \Int{   \Ga_{\e{eq}} }{} \ln \Big\{ \frac{\dd \mu_{ \Ga_{\e{eq}} } }{ \dd z }(z) \Big\}  \, \cdot  \, \dd \mu_{ \Ga_{\e{eq}} } (z)
    +   \ln \frac{\beta}{2} \right] + \frac{\beta}{2} \ln  \frac{2\pi}{\e{e}}- \ln \Gamma\left(  \frac{\beta}{2}\right).
\end{align*}

\end{theorem}
\begin{remark} The above expansion is a "complexification" of the known expansion for the partition for $\beta$-ensembles on the real line (see Section 7.1.1 of \cite{BG24}),
so that $ \mc{I}_{  \Ga_{\e{eq}} }$ is a complexification of the usual logarithmic energy in an external field and
$$\mathrm{Ent}[\mu_{\Ga_{\e{eq}}} ] = - \Int{   \Ga_{\e{eq}} }{} \ln \Big\{ \frac{\dd \mu_{ \Ga_{\e{eq}} } }{ \dd z }(z) \Big\}  \, \cdot  \, \dd \mu_{ \Ga_{\e{eq}} } (z) $$
is a complexification of the entropy. This is natural since in the case of $V$ having real coefficients and $\Gamma = \mathbb{R}$ we expect to reproduce known results for the real case.

One should note that the density $\tfrac{ \dd\mu_{ \Ga_{\e{eq}} }(z) }{\dd z}$ is complex valued since it is a density with respect to the complex differential $\dd z = \ga_{\e{eq}}^{\prime}(t) \dd t$ on $\Ga_{\e{eq}}$,  so that overall $\dd\mu_{ \Ga_{\e{eq}} }(z) = \tfrac{ \dd\mu_{ \Ga_{\e{eq}} }(z) }{\dd z} \dd z$ is real valued.
\end{remark}

The motivations for studying the complex valued partition functions (\ref{Cpartitionfunction}) are multifarious. In the case $\beta = 2$ and $\Gamma = \mathbb{R}$, such integrals are related to counting surfaces of a given
genus \cite{countingsurfaces}, of great importance because of the relevance to quantum gravity. Here one is interested in  how $F_k(\beta,V)$ depend on the coefficients of $V$, more precisely, how $F_k(\beta,V)$, when Taylor
expanded in the coefficients of $V$, behaves deep into the expansion. This behaviour is governed by the structure of singularities in the complex plane,
and so one is forced to consider a polynomial potential $V$ with complex coefficients. This has motivated recent work studying integrals of the form (\ref{Cpartitionfunction}),
for $\beta = 2$, in the cubic \cite{10.1063/5.0086911,8177967}  and quartic \cite{https://doi.org/10.1002/cpa.22164} cases by means of the Riemann-Hilbert approach to orthogonal polynomials. Riemann-Hilbert techniques are,
however, difficult to extend beyond  $\beta = 2$ and Riemann--Hilbert formulations are only known to exist for $\be\in \{1,2,4\}$. The partition function (\ref{Cpartitionfunction}) also appears in the theory of non-Hermitian orthogonal polynomials on curves in $\mathbb{C}$; namely it appears as the denominator in a Heine-type formula. The asymptotics of such non-Hermitian orthogonal polynomials were considered in the quartic case by \cite{Bertola2011AsymptoticsOO}, again by Riemann-Hilbert techniques.

Finally, correlation functions in truly interacting, \textit{i.e.} non-equivalent to free fermions, quantum integrable models are conjectured to be given by series of multiple integrals, \textit{viz}. series of the form
\beq
\sul{N \geq 0}{} \pl{s=1}{a_N} \Big\{ \Int{ \Ga }{}  \dd \la_s  \Big\} \cdot \mc{F}_N\big( \bs{\la} \mid \bs{x} \big)
\label{ecriture series integrales multiples}
\enq
where $a_N$ is typically an affine  function of $N$,  $\mc{F}_N$ an explicit integrand and $\Ga$ a simple smooth curve in $\Cx$.
$\bs{x}$ is a set of parameters on which the given correlator may depend. The building blocks of \eqref{ecriture series integrales multiples}
are dependent on the model under consideration and on the correlator of interest.
In certain models such as the XXZ spin-1/2 chain or the Sine-Gordon quantum field theory, such series yield
one parameter deformations of Painlevé $\tau$-functions: see \textit{e.g.} \cite{KozKitMailSlaTerXXZsgZsgZAsymptotics}
for the generalisation of the PV transcendent or \cite{BabujianKarowskiExactFFSineGordonBootsstrapII} which generalises the PIII
transcendent of \cite{BernardLeclairDiffEqnForPIII}. The main mathematical issue with these representations is that their convergence theory is not well developed.
The only case where it has been possible to establish convergence
pertains to the space-like separated vacuum two-point functions in the integrable Sinh-Gordon quantum field theory \cite{KozConvergenceFFSeriesSinhGordon2ptFcts}.
There, one has $a_N=N$ and $\Ga=\R$, which allowed \cite{KozConvergenceFFSeriesSinhGordon2ptFcts} to exploit techniques developed for the large-$N$ analysis of real valued $\be$-ensembles and
real-valued generalisations thereof so as to obtain an upper bound in $N$ for the $N$-fold integrals making up the series and so deduce convergence.
While constituting a definite progress,  more tools are needed to treat the case of more general models for which the contour  $\Ga$ does not coincide with $\R$ and $\mc{F}_N$ is genuinely complex valued. The techniques of large-$N$ analysis for $\be$-ensembles in the complex valued setting \eqref{Cpartitionfunction}
developed in the present work may be thus seen as paving the way for developing a convergence theory for series of multiple integrals describing correlation functions in quantum integrable models.

This work draws on many existing techniques in the literature,  namely the large deviation theory of $\beta$-ensembles \cite{Borot:2013aa,BG24} and Coulomb gases with more general mean-field interactions \cite{10.1093/imrn/rnu260}.  In particular the case where $V$ has real coefficients, even degree and positive leading coefficient, and $\Gamma = \mathbb{R}$,  such integrals (\ref{Cpartitionfunction}) have been studied in the literature for off-critical potentials, both in the one-cut \cite{Borot:2013aa} and in the multi-cut regime \cite{BG24}.  Extending to critical potentials is an open problem and only partial results are known; in particular Bekerman, Leblé and Serfaty \cite{10.1214/18-EJP209} have shown that, on a certain subspace for which the master operator is invertible, linear statistics asymptotically obey a CLT.  If the method of the aforementioned papers is regarded as an infinite dimensional version of the Laplace method, then the theory developed in the present paper can be regarded as an infinite dimensional version of the method of steepest descent, where the $S$-curve plays the role of steepest descent curve.

To motivate what follows,  let us outline the technique of proof. The method of the present paper follows closely the method of Borot-Guionnet \cite{Borot:2013aa} with an important caveat, which is that establishing a lower bound for (\ref{Cpartitionfunction}) is straightforward in the real case but is a major technical challenge in the complex case and is the central problem of the present paper. Such a lower bound is needed to establish \textit{a priori} bounds on moments of linear statistics in the complex model, which is needed to apply the method of Dyson-Schwinger equations.
More precisely, we show that there exists a constant $C_\beta > 0$ such that $|\mathcal{Z}_{N,\Gamma}[V]| \geq  C_\beta \op{Z}_{N,\gamma}[V]$ for $N$ sufficiently large
and with $\Ga=\ga(\R)$. This bound cannot be expected to hold for any contour but holds only when $\Gamma$ is an $S$-curve in the external field $\varphi = \Re(V)$. In this sense, the $S$-curve is the analogue of a steepest descent curve, since it "minimises the oscillations." We prove such a bound by separating the integrand of (\ref{Cpartitionfunction}) into a modulus and a phase and showing that the phase may be written as the complex exponential of a smooth quadratic form in the empirical measure of the particles. In turn, this will be shown to converge to a quadratic form in Gaussian variables, and then by a functional-analytic argument we show that the expectation of the exponential of this quadratic form is bounded from below by something uniformly positive, thus establishing the bound.

The structure of the paper is as follows. In Section \ref{equilibriumsection} we review the theory of equilibrium measures on curves in $\mathbb{C}$. This theory is well-known however we review it to explain certain innovations we make with regard to the interpolation between $V$ and $V_0$. For real $\beta$-ensembles, in the one-cut regular case, this interpolation is simple: by choosing $V_0$ appropriately one can take the intermediate potentials to be $V_t = t V + (1-t) V_0$. Then, the equilibrium measure for $V_t$ is the convex combination (see Equation 7.5 of \cite{BG24})
of the equilibrium measures of $V$ and $V_0$. However if one applies a similar idea in the complex case, the corresponding intermediate equilibrium measures
are in principle much more complicated and it is not at all clear that they are one-cut regular.
To circumvent this difficulty a more elaborate construction was needed. This construction is also an innovation of the present paper.

In Section \ref{concentrationsection} we establish \textit{a priori} concentration bounds of linear statistics for the real model supported on a compact curve (\ref{realmodel}). In Section \ref{maxeigsection} we use the theory of large deviations to show that we may replace the unbounded contour $\Gamma$ with a compact curve $\Sigma \subset \Gamma$ up to an event of exponentially small probability. In Section \ref{realDSsection} we prove the desired lower bound on the complex partition function. More precisely, we prove such a bound holds uniformly for an interpolating family of potentials $V_t$ and their associated $S$-curves. This is done by applying the method of Dyson-Schwinger equations to the real model (\ref{realmodel}) on a compact interval. As already stated, we first separate the complex integrand into a modulus and a phase, and hence the ratio of the partition functions may be written as an expectation with respect to the real model of a complex exponential. We will find the exponent is a smooth quadratic form in the empirical measure of the particles. We then use the method of Dyson-Schwinger equations to prove a central limit theorem for linear statistics in the real model. An essential step in the analysis is to invert the master operator of the real model, the invertibility of which is demonstrated in Appendix \ref{singularsection}. This leads to an expectation with respect to a Gaussian process which can be represented in terms of a Fredholm determinant. This Fredholm determinant representation is proven in Appendix \ref{fredholmsection}. With this lower bound established, Section \ref{expansionsection} returns to established techniques which allow one to characterise the expansion of the partition function.

\newpage

\section{Equilibrium measures on contours}\label{equilibriumsection}

In this section, we  review the theory of logarithmic potentials on curves. Let us assume we have a simple (\textit{i.e.} injective) $C^\infty$ curve
$\gamma : \mathcal{J} \longrightarrow \mathbb{C}$, where $\mathcal{J}$ is either a compact interval or $\mathcal{J} = \mathbb{R}$, and we denote $\Gamma = \gamma(\mathcal{J})$.
In the case of $\mathcal{J} = \mathbb{R}$ we assume $\Gamma$ satisfies \ref{hypothesis2} and \ref{hypothesis3}. Let $\mathcal{P}(\Gamma)$ be the set of Borel probability measures supported on $\Gamma$.
We define the energy functional $\op{I}_{\Ga} :  \mathcal{P}(\Gamma) \longrightarrow \intof{-\infty}{+\infty} $ as
\beq
\op{I}_{\Ga} [\mu] \overset{\mathrm{def}}{=} \Int{\Gamma^2}{} \Big\{ \ln\frac{1}{|z-w|}+\varphi(z) + \varphi(w)  \Big\}  \dd \mu(z)  \dd \mu(w)
\nonumber
\enq
where $\varphi \overset{\mathrm{def}}{=} \Re(V)$ and $V(z) = \frac{z^\kappa}{\kappa}+\e{O}(z^{\kappa-1})$ is the polynomial potential.
\begin{definition}[Bounded Lipschitz topology]
Let us equip $\mathcal{P}(\Gamma)$ with the bounded Lipschitz topology.  More precisely, for a bounded Lipschitz function $f : \Gamma \longrightarrow \mathbb{R}$, let
$$\| f\|_{\mathrm{BL}(\Gamma)} = \sup_{x \in \Gamma} |f(x)| + \sup_{\substack{x,y \in \Gamma \\ x \neq y}} \frac{|f(x)-f(y)|}{|x-y|}$$
be the bounded Lipschitz norm. Then the "dual" norm on (signed) Borel measures of finite total variation is
$$\| \mu \| = \sup_{\| f \|_{\mathrm{BL}(\Gamma)} \leq 1}\left|  \int_\Gamma f \, \dd \mu \right| \, . $$
This induces a metric on $\mathcal{P}(\Gamma)$, $d_{\mathrm{BL}}(\mu,\nu) = \| \mu - \nu \|$.
\end{definition}
It is well-known that this norm metrises weak convergence, so that $\| \mu_n \| \to 0$ if and only if $\mu_n \to 0$ weakly and that $(\mathcal{P}(\Gamma), d_{\mathrm{BL}})$ is a Polish space, \textit{i.e.}
it is complete and separable.
\begin{remark} Given a function $f : \Gamma \longrightarrow \mathbb{R}$ we can naturally associate a function $F = f \circ \gamma : \mathcal{J} \longrightarrow \mathbb{R}$.
By \textnormal{\ref{hypothesis2}} and \textnormal{\eqref{hypothesis4}} we have
$$C_\gamma \| f \|_{\mathrm{BL}(\Gamma)} \leq \| F \|_{\mathrm{BL}(\mathcal{J})} \leq C_\gamma^{-1} \| f \|_{\mathrm{BL}(\Gamma)}\, .$$
Hence the space of bounded Lipschitz functions on $\Gamma$ and the space of bounded Lipschitz functions on $\mathcal{J}$ are equivalent as metric spaces, with the pushforward under $\gamma$
providing a homeomorphism between the two. For this reason we shall often speak of the weak topology on $\mathcal{J}$ and on $\Gamma$ interchangeably.
\end{remark}
Moving forward, is well known \cite{AGZ,BAG97,BAZ98} that the energy functional $\op{I}_{\Ga}$ has the following properties.
\begin{lemma}\label{goodrate}
One has that
\begin{enumerate}[label=(\arabic*)]
\item $\op{I}_{\Ga}$ is a "good rate function," i.e. is lower semicontinuous and has compact level sets (with respect to the weak topology).
\item There exists a unique probability measure $\mu_{\Ga} \in \mathcal{P}(\Gamma)$, the so-called equilibrium measure, which minimises the energy, \textit{i.e.}
$\op{I}_{\Ga}[\mu_{\Ga}] = \underset{\nu \in \mathcal{P}(\Gamma)}{\inf}\op{I}_{\Ga}[\nu]$.
\item $\mu_{\Ga}$ is compactly supported.
\end{enumerate}
\end{lemma}
\begin{proof}
The first statement is proven by the same method as the proof of Theorem 2.6.1 in \cite{AGZ} thanks to  \eqref{Vinf}.
The existence of a minimiser follows from the first statement. That the minimiser is unique follows from Lemma \ref{logarithmicdistance}.
The final statement follows from the proof of Lemma 2.6.2 in \cite{AGZ}.
\end{proof}
The above proof differs only from standard proofs in that one must use a slightly different argument to show the positive definiteness of the "logarithmic distance",
which is here defined for measures in the complex plane rather than on the real line. This is contained in the following Lemma, which also implies convexity of the
energy functional $\op{I}_{\Ga}$.
\begin{lemma}\label{logarithmicdistance} Let $\mu$ be a signed, real valued measure of finite total variation in $\mathbb{C}$ and zero net mass, $\mu[\mathbb{C}] = 0$. Then
\begin{align}\label{logarithmicIP}
\int_{\mathbb{C}^2} \ln\frac{1}{|z-w|} \, \dd\mu(z)  \dd\mu(w) = \frac{1}{2\pi} \int_{\mathbb{R}^2} \frac{|\widehat{\mu}(p,q)|^2}{p^2 + q^2} \, \dd p \, \dd q \geq 0
\end{align}
where
$$\widehat{\mu}(p,q) = \int_{\mathbb{C}} \ex{-\i p x - \i q y} \, \dd\mu(x+\i y)$$
is the planar Fourier transform of the measure.
\end{lemma}
\begin{proof} This formula can be found as Lemma 2.4 of \cite{19133cf7-3d88-3b93-9dd8-2d2fe21932d0} but we include a proof for completeness.

\begin{align*}
\Int{ \mathbb{C}^2 }{} \ln{|z-w|}\, \dd \mu(z)  \dd \mu(w) &=  \frac{1}{2}  \Int{\mathbb{C}^2}{} \Bigg( \Int{1}{|z-w|^2} \frac{\dd s}{ s } \Bigg) \, \dd \mu(z)  \dd \mu(w) \\
&=\frac{1}{2}  \Int{ \mathbb{C}^2 }{} \Bigg( \Int{ 1 }{ |z-w|^2 } \bigg( \Int{0}{+\infty} \ex{-s t }\, \dd t \bigg) \dd s \Bigg) \, \dd \mu(z)  \dd \mu(w) \\
&= -\frac{1}{2} \Int{0}{+\infty}  \Bigg( \Int{ \mathbb{C}^2 }{} \ex{ - \frac{|z-w|^2}{t} }  \, \dd \mu(z)  \dd \mu(w) \Bigg) \, \frac{ \dd t }{ t } \;,
\end{align*}
where we have used that $\mu$ has net mass zero. Finally, if we write $z = x+\i y$ and $w = u+\i v$ we may use the Hubbard-Stratonovich transformation
$$\ex{ -\frac{|z-w|^2}{t} }  \, = \,  \frac{t}{4\pi} \Int{ \mathbb{R}^2 }{} \ex{ - \frac{t}{4}(p^2 + q^2) -\i (x-u) p - \i (y-v) q } \, \dd p\, \dd q \; . $$
Substituting this into our formula and performing the integral over $t$ yields the result.
\end{proof}

\begin{remark}
Note that if $\mu$ is supported on the real line then the Fourier transform $\widehat{\mu}(p,q) = \widehat{\mu}(p)$ is independent of $q$.
Then the integral $\int_\mathbb{R} \frac{1}{p^2 + q^2 }\, dq = \frac{\pi}{|p|}$ may be performed and we recover the usual formula which appears for measures on $\mathbb{R}$.
\end{remark}
\begin{definition}\label{logdist}
Let $\mu,\nu \in \mathcal{P}(\Gamma)$, then
\begin{align*}
&\mf{D} :  \mathcal{P}(\Gamma) \times  \mathcal{P}(\Gamma) \longrightarrow \intff{0}{+\infty} \\
&\mf{D}(\mu,\nu) = \sqrt{ \Int{\mathbb{C}^2}{} \ln\frac{1}{|z-w|} \, \dd(\mu-\nu)(z)  \dd(\mu-\nu)(w)}
\end{align*}
defines a distance (which may take the value $+\infty$) on $ \mathcal{P}(\Gamma)$. By using $\gamma$ to pull back the measures onto $\mathcal{J}$, the support of $\ga$,
\begin{align*}
&\mf{D}_{\ga} :  \mathcal{P}(\mathcal{J}) \times  \mathcal{P}(\mathcal{J}) \longrightarrow \intff{0}{+\infty} \\
&\mf{D}_{\ga}(\mu,\nu) = \sqrt{ \Int{ \mathcal{J}^2 }{} \ln\frac{1}{|\gamma(x)-\gamma(y)|} \, \dd(\mu-\nu)(x)  \dd(\mu-\nu)(y) }
\end{align*}
similarly defines a distance.
\end{definition}
From Lemma \ref{logarithmicdistance}, it is easy to deduce the following corollary.
\begin{corollary}[Logarithmic Cauchy-Schwarz]\label{logCS}
Let $\mu$, $\nu$ be two signed Borel measures on $\Gamma$,  each of zero total mass $\mu[\Gamma] = \nu[\Gamma] = 0$ and of finite total variation.  Then
$$\Int{ \Gamma^2 }{} \ln\frac{1}{|x-y|}\, \dd\mu(x) \, \dd\nu(y)  \, \leq \, \sqrt{  \Int{ \Gamma^2 }{ } \ln\frac{1}{|x-y|}\, \dd\mu(x) \, \dd\mu(y) }
\cdot \sqrt{ \Int{ \Gamma^2 }{} \ln\frac{1}{|x-y|}\, \dd\nu(x) \, \dd\nu(y) } \; .$$
\end{corollary}
In the case that $\mathcal{J}$ is the compact interval $\intff{a}{b}$ we can establish a relationship between $\mf{D}_{\ga}$ and the topology of weak convergence.
\begin{proposition}[Relationship of $\mf{D}_{\ga}$ to weak topology]\label{weaktopology}
Let $\mc{J}=\intff{a}{b}$ with $a < b$ and let $f \in \mc{W}^{\infty}_{1}(\mc{J})$, where $\mc{W}^{\infty}_{p}(\mc{J})$ and its related norm are recalled in Definition \ref{Ckspaces}.
Let $\mu, \nu \in \mathcal{P}(\mc{J})$.  Then there is a constant $\mf{C}_{\mc{J}} > 0$ such that
\begin{align}
\left| \Int{\mc{J}}{} f \, \dd (\mu- \nu) \right| \leq \mf{C}_{\mc{J}} \| f \| _{\mc{W}^{\infty}_{1}(\mc{J})} \mf{D}_{\ga}(\mu,\nu) \;.
\end{align}
Above, $\mf{D}_{ \ga }$ is as introduced through Definition \ref{logdist}.
If $\gamma$ depends on some auxiliary parameter $t \in \intff{0}{1}$ in such a way that the constants in \eqref{hypothesis4} are uniform in $t$ and the
$\norm{\ga}_{\mc{W}_{p}^{\infty}(\mc{J})}$ are also uniformly bounded in $t$,
then
$\mf{C}_{\mc{J}} > 0$ may be taken independent of $t$.
\end{proposition}
\begin{proof}
Let us first consider the case when $f \in \mc{C}^\infty(\mc{J})$. By virtue of Proposition \ref{realmasterinverse}, \textit{c.f.} \eqref{ecriture explicite inverse real master operator},
there exists $g \in \mc{C}^{\infty}(\mc{J})$ and an $f$-dependent constant $\mc{K}_{\ga}[f]$  such that
$$f(x) - \mc{K}_{\ga}[f] \,= \, \Int{ \mc{J} }{} \ln \frac{1}{|\gamma(x) - \gamma(y)| }\frac{\dd}{\dd y}\left( \sqrt{(y-a)(b-y)} g(y) \right) \, \dd y =: \Xi_{\gamma}[g](x)$$
where $\Xi_{\gamma}$ is the real master operator defined in \eqref{definition real master operator}, and subordinate to the choice $\mc{J}^{\prime}=\mc{J}$.
Note that the above formula follows from an integration by parts in the definition \eqref{definition real master operator}.
Then
$$
\dd \sigma(x) \, = \, \mathbbm{1}_{\mc{J}}(x) \frac{ \dd }{ \dd x } \left( \sqrt{(x-a)(b-x)} g(x) \right) \dd x$$
defines a signed Borel measure of zero net mass on $\mc{J}$.  Substituting this in and applying Corollary \ref{logCS} we have
\begin{align*}
\left| \Int{\mc{J}}{} f(s) \, \dd (\mu- \nu)(s) \right| & \leq \mf{D}_{\ga }(\mu,\nu) \sqrt{\Int{\mc{J}}{} f(x) \frac{\dd}{\dd x}\left(\sqrt{(x-a)(b-x)}g(x)\right)\, \dd x}
\leq \, C_{ \mc{J} } \| f \|_{\mc{W}^{\infty}_{1}(\mc{J})} \mf{D}_{\ga}(\mu,\nu)
\end{align*}
by Proposition \ref{realmasterinverse}.  Furthermore, this Proposition also ensures that $\mf{C}_{\mc{J}} > 0$ may be taken to be independent of $t \in \intff{0}{1}$ if $\gamma $
depends on $t$ as stated in this proposition's hypotheses. The bound then extends to all $\mc{W}^{\infty}_{1}(\mc{J})$ functions by density.
\end{proof}
Moving forward, from the fact that $\mu_{\Ga}$ minimises $\op{I}_{\Ga}$, $\mu_{\Ga}$ must satisfy a first order condition.  This is the content of "Frostman's lemma" which we shall now state.  Define the \textit{electrostatic potential} associated to a probability measure $\mu$ as
\begin{align}
U[\mu](z) \overset{\mathrm{def}}{=} \Int{\mathbb{C}}{} \ln \frac{1}{|z-w|} \, \dd \mu(w).
\end{align}
$U[\mu]$ takes values in $\mathbb{R} \cup \{ + \infty \}$ and is lower semi-continuous, since one can write
$$U[\mu](z) \, =\,  \sup_{L > 0} \Int{ \mathbb{C} }{ } \min\left\{L, \ln \frac{1}{|z-w|}\right\} \, \dd \mu(w).$$
\begin{definition}[Zero capacity set]
We say a compact set $K \subset \mathbb{C}$ has \textit{zero capacity} if for every $\nu \in \mathcal{P}(K)$,
$$ \Int{K^2}{} \ln\frac{1}{|z-w|}\, \dd\nu(z) \, \dd\nu(w) = +\infty \, . $$
An arbitrary Borel set $A  \subset \mathbb{C}$ is said to to have zero capacity if for every compact set $K \subset A$, $K$ has zero capacity (see Section I.1 of \cite{saff}).
We say a statement holds quasi-everywhere (q.e.) if it fails on at most a set of zero capacity.
\end{definition}
\begin{lemma}[Frostman-Euler-Lagrange]\label{frostmanlemma} There exists a constant $C_{\Ga} \in \mathbb{R}$, termed the "equilibrium constant," such that
\begin{align}
U[\mu_{\Ga}](z) + \varphi(z) & \geq  C_{\Ga} & & \forall  z \in \Gamma  \setminus \mathrm{supp}\, \mu_{\Ga} \, , \label{Frostman1} \vspace{3mm}\\
U[\mu_{\Ga}](z) + \varphi(z) &= C_{\Ga} & & \mathrm{q.e.} \, z \in \mathrm{supp}\, \mu_{\Ga}  \, . \label{Frostman2}
\end{align}
Conversely (\ref{Frostman1}) and (\ref{Frostman2}) are also sufficient conditions for the equilibrium measure, and since the equilibrium measure is unique, they uniquely define the equilibrium measure. More precisely, if one has a probability measure on $\Gamma$ for which there exists a constant $C_{\Ga}$ such that (\ref{Frostman1}) and (\ref{Frostman2}) hold, then such a measure is identical with the equilibrium measure.
\end{lemma}
\begin{proof}
See Theorem I.1.3 in \cite{saff}.
\end{proof}
\begin{remark}
If $\mu_{\Ga}$ has continuous density with respect to arc length measure then $U[\mu_{\Ga}]$ is continuous and the equality (\ref{Frostman2}) holds everywhere in the support.
\end{remark}
\begin{corollary}\label{deformation}
Let $\Gamma$ and $\widetilde{\Gamma}$ be two "admissible" contours (in the sense of Definition \ref{admissible}) and let $\varphi_{\Ga; \mathrm{eff}}$ be its associated \textnormal{effective potential}
\begin{align}\label{effective}
\varphi_{\Ga; \mathrm{eff}}(z) \overset{\mathrm{def}}{=} \varphi(z) + U[\mu_{\Gamma}](z) - C_{\Gamma}, & & z \in \mathbb{C}.
\end{align}
Suppose that $\Gamma$ and $\widetilde{\Gamma}$ only differ outside $\mathrm{supp}\, \mu_{\Gamma}$, i.e. $\mathrm{supp}\, \mu_{\Gamma} \subset \Gamma \cap \widetilde{\Gamma}$,
and $\varphi_{\Ga;\mathrm{eff}}(z) \geq 0$ for quasi-every $z \in \widetilde{\Gamma}$. Then $\mu_{\Gamma}$ is also the equilibrium measure for the contour $\widetilde{\Gamma}$,
i.e. $\mu_{\Gamma}=\mu_{ \wt{\Gamma} }$.
\end{corollary}
This corollary effectively states that we are free to deform the contour away from the support in the region where the effective potential is non-negative without affecting the equilibrium measure.

\begin{lemma}\label{continuousdensity}
The equilibrium measure has a continuous density with respect to Lebesgue measure. More precisely,  there is a function $\lambda \in \mc{C}^\infty(\Gamma)$ such that
$$\dd \mu_{\Ga}(x) = \sqrt{\max\{ \lambda(x),0\} } \, |\dd x|$$
where $|dx|$ is arc-length measure on $\Gamma$.
\end{lemma}
\begin{proof}
See the proof of Lemma 3.2 in \cite{Bekerman:2015aa}.
\end{proof}

The theory developed thus far has pertained to sufficiently regular but otherwise arbitrary contours in $\mathbb{C}$.
However, if we optimise our contour, then we can achieve additional desirable properties. Motivated by our discussion in the introduction,
we will deform $\Gamma$ to a contour $\Gamma_{\e{eq}}$ which solves the so-called \textit{max-min energy problem}. To make this precise we must introduce some definitions, taken from \cite{lirias1643667}.
We first introduce the set of admissible contours. In the aforementioned work, the authors consider contours with more complicated asymptotic conditions
labelled by non-crossing partitions. Our setup amounts to considering only the partition of $\{ 1, \dots, \kappa \}$ containing $\{ \alpha,  \alpha^{\prime} \}$ and all other sets being singleton sets.
\begin{definition}[Stretching out to infinity \cite{lirias1643667}]
We define $\kappa$ sectors
\beq
S_j  \, = \, \left\{ z\in \mathbb{C} \, : \, \left| \arg(z) - \theta_j \right| < \frac{\pi}{2\kappa} \right\}  \quad \e{where} \qquad
\theta_j  \, =  \, \frac{ 2\pi (j-1) }{ \kappa } \quad \e{with} \quad  j \in \intn{1}{\kappa} \;.
\nonumber
\enq
A set $F \subset \mathbb{C}$ is said to "stretch out to infinity in sector $S_j$" if there is an $\epsilon > 0$ and an $r_0 > 0$ such that for all $r > r_0$ there is a $z \in F$ such that
\beq
|z| = r  \qquad \e{and} \qquad   \left| \arg(z) - \theta_j \right| < \frac{\pi}{2\kappa} - \epsilon.
\nonumber
\enq
\end{definition}
\begin{definition}[Admissible contours \cite{lirias1643667}]\label{admissible}
The set of admissible contours $\mathcal{T}(\alpha,\alpha^\prime)$ is defined as follows. $\Gamma \in \mathcal{T}(\alpha,\alpha^\prime)$ if and only if the following conditions hold
\begin{enumerate}[label=(\arabic*)]
\item $\Gamma$ is a finite union of $\mc{C}^1$ Jordan arcs.
\item $\Gamma$ is connected.
\item $\Gamma$ stretches out to infinity in the sectors $S_\alpha$ and $S_{\alpha^\prime}$.
\item For every $\ell \in \{ 1, \dots, \kappa \} \setminus \{ \alpha,\alpha^\prime \}$, there is an $R>0$ sufficiently large so that $\Gamma \cap (S_\ell \setminus \op{D}_{0,R} ) = \emptyset$.
\end{enumerate}
Above, $\op{D}_{z_0,\eta}$ stands for the open disk of radius $\eta$ centred at $z_0$.

\end{definition}
We say that  $\Gamma_{\e{eq}}  \in \mathcal{T}(\alpha,\alpha^\prime)$ solves the \textit{max-min energy problem} if
\begin{align}\label{maxmin}
\sup_{\widetilde{\Gamma} \in \mathcal{T}(\alpha,\alpha^\prime)} \inf_{\mu \in \mathcal{P}(\widetilde{\Gamma})} \op{I}_{\Ga}[\mu]=  \inf_{\mu \in \mathcal{P}(\Ga_{\e{eq}})} \op{I}_{\Ga_{\e{eq}}} [\mu]
\, = \, \op{I}_{\Ga_{\e{eq}}} \big[ \mu_{\Ga_{\e{eq}}}  \big]\, .
\end{align}
Related to this is the so-called $S$-property.
\begin{definition}[$S$-property]
We say that $\Gamma \in \mathcal{T}(\alpha,\alpha^\prime)$ satisfies the $S$-property if there is a set of zero capacity $E$ such that for any $z \in \mathrm{supp}\, \mu_{\Ga} \setminus E$
there is a neighbourhood $D \ni z$ such that $D \cap  \mathrm{supp}\, \mu_{\Ga}$ is an analytic arc and
\begin{align}\label{Scurve}
\frac{\partial}{\partial n_+} \big( U[\mu_{\Ga}]+ \varphi \big)(z) = \frac{\partial}{\partial n_-} \big( U[\mu_{\Ga}] + \varphi \big)(z) & & \forall z \in \mathrm{supp}\, \mu_{\Ga} \setminus E
\end{align}
where $\frac{\partial}{\partial n_\pm}$ are the normal derivatives taken on either side of the contour. We remind the reader that
$\mu_{\Ga} \in \mathcal{P}(\Gamma)$ is the equilibrium measure on $\Gamma$ in the external field $\varphi$. If $\Gamma$ satisfies the $S$-property we say that $\Gamma$ is an $S$-curve.
\end{definition}
The $S$-curve condition should be regarded as a first order condition for the maximisation of the energy of the curve.
We now state a crucial existence theorem established by \cite{lirias1643667}. Let $\mathsf{C}[\mu]$ be the Cauchy transform of a measure $\mu$, namely,
\begin{align*}
\mathsf{C}[\mu](z) \overset{\mathrm{def}}{=} \frac{1}{2\pi \i} \Int{\mathbb{C}}{} \frac{1}{w-z} \, \dd \mu(w), & & z \in \mathbb{C}\setminus \mathrm{supp} \, \mu .
\end{align*}
\begin{theorem}[Existence of curve solving the max-min energy problem \cite{lirias1643667}]\label{existence}
There exists a contour $\Gamma_{\e{eq}} \in \mathcal{T}(\alpha,\alpha^\prime)$ which solves the max-min energy problem (\ref{maxmin}).
Furthermore the contour $\Gamma_{\e{eq}}$ is an $S$-curve in the external field $\varphi = \Re (V)$.

The equilibrium measure $\mu_{\Ga_{\e{eq}}}$ of $\Ga_{\e{eq}}$ in external field $\varphi$ is supported on a finite union of analytic arcs that are critical trajectories of the quadratic differential $-R(z) \, \dd z^2$,
where
\begin{align}\label{Rformula}
R(z) = \left( 2\pi \i \, \mathsf{C}[ \mu_{\Ga_{\e{eq}}} ] (z) + V^\prime(z) \right)^2  & & z \in \mathbb{C}\setminus \mathrm{supp} \, \mu_{\Ga_{\e{eq}}}
\end{align}
is a monic polynomial of degree $2 \kappa-2$. The equilibrium density may be expressed in terms of $R$ by
\begin{align*}
\dd \mu_{\Ga_{\e{eq}}} (z) \, =  \, \frac{1}{\i \pi} \sqrt{R(z)}_+ \, \dd z .
\end{align*}
Above, $\sqrt{R(z)}_{+/-}$ denotes the limiting value of $s \mapsto \sqrt{R(s)}$ as $s \rightarrow z\in \Ga_{\e{eq}}$ non-tangentially from the left/right of $\mathrm{supp} \,  \mu_{\Ga_{\e{eq}}}$
with respect to its orientation. Moreover, one has the closed formula
\begin{align}\label{squareroot}
\sqrt{R(z)} = 2\pi \i \, \mathsf{C}[ \mu_{\Ga_{\e{eq}}} ](z) + V^\prime(z),   & & z \in \mathbb{C}\setminus \mathrm{supp} \, \mu_{\Ga_{\e{eq}}} \, .
\end{align}
Finally, the $S$-curve condition (\ref{Scurve}) holds everywhere except at the endpoints of the analytic arcs.
\end{theorem}

Note that the arc one-form $\dd z$ and $\frac{ 1 }{ \i \pi } \sqrt{R(z)}_+$ are separately not real valued, however their product is real valued,
and to say that $\frac{ 1 }{ \i \pi } \sqrt{R(z)}_+ \, \dd z > 0$ is exactly to say that the support of the equilibrium measure is a critical trajectory of $-R(z) \, \dd z^2$.
We stress that the curve $\Gamma_{\e{eq}}$ is \textit{not} unique. Indeed, by Corollary \ref{deformation} the curve $\Gamma_{\e{eq}}$ can be deformed outside of $\mathrm{supp}\, \mu_{\Ga_{\e{eq}}}$
into another curve $\wt{\Gamma}_{\e{eq}}$ so long as $\varphi_{ \Ga ; \mathrm{eff}}$ defined in  (\ref{effective}) remains non-negative throughout $\wt{\Gamma}_{\e{eq}}$.
We also specify that the $S$-curve property only tests properties on the support of the equilibrium measure.  However, up to such deformations, the optimal curve is unique, so that any two contours that solve the max-min
energy problem for $\mathcal{T}(\alpha,\alpha^\prime)$ will share the same equilibrium measure (see p. 7 of \cite{lirias1643667}).

\vspace{2mm}

It follows from the representation \eqref{squareroot} that
\beq
\sqrt{R(z)} \, = \, S(z) \cdot \big[ (z-\zeta_1)(z-\zeta_2) \big]_{ \sg_{\e{eq}} }^{1/2}
\label{ecriture decomposition canonique pour racine de R}
\enq
with $S$ a monic polynomial of degree $\kappa -2$ and $\sg_{\e{eq}}=\e{supp} \mu_{ \Ga_{\e{eq}} }$. Here,
\beq
\big[ (z-\zeta_1)(z-\zeta_2) \big]_{ \sg_{\e{eq}} }^{1/2} \; = \; (z-\zeta_1) \cdot \exp\bigg\{ \Int{ \sg_{\e{eq}} }{}  \f{ \dd s  }{ 2 (s-z) }  \bigg\}
\label{definition racine avec coupure sur une courbe prescrite}
\enq
stands for the square root of $(z-\zeta_1)(z-\zeta_2)$ which has its cuts along $\sg_{\e{eq}}$, is analytic on $\Cx \setminus \sg_{\e{eq}}$, behaves as
$z \big( 1 + \e{o}(1)\big) $ and whose boundary values on $\sg_{\e{eq}}$ satisfy
$$\big[ (z-\zeta_1)(z-\zeta_2) \big]_{ \sg_{\e{eq}} ;+}^{1/2} \; = \;- \big[ (z-\zeta_1)(z-\zeta_2) \big]_{ \sg_{\e{eq}} ;- }^{1/2}  \quad \e{for} \quad z \in \mathrm{supp} \, \mu_{ \Ga_{\e{eq}} } \;.$$
One establishes \eqref{ecriture decomposition canonique pour racine de R}
by observing that given a contour $\Ga_{\e{eq}}$ solving the max-min problem, it holds
\beq
\sqrt{R(z)}_+ \, = \,  - \sqrt{R(z)}_-    \qquad \e{for} \qquad   z \in \mathrm{supp} \, \mu_{ \Ga_{\e{eq}} } \, .
\label{ecriture conditions de saut pour sqrt de R}
\enq
Hence,
\beq
S(z) \, = \,  \f{ \sqrt{R(z)} }{ \big[ (z-\zeta_1)(z-\zeta_2) \big]_{ \sg_{\e{eq}} }^{1/2}  }
\enq
extends into an analytic function on $\Cx$ since it has no jumps on $\sg_{\e{eq}}\setminus \{\zeta_1, \zeta_2\}$ and has removable
singularities at $\zeta_1, \zeta_2$. Since it grows at most as $z^{\kappa-2}$ at infinity, one infers that $S$ is a monoic polynomial
of degree $\kappa-2$.

\vspace{2mm}

We deduce from \eqref{ecriture conditions de saut pour sqrt de R} the complexified Frostman-Euler-Lagrange conditions
\begin{align}\label{EulerLagrange}
\pi \i \, \mathsf{C}[\mu_{ \Ga_{\e{eq}} }]_+(z) + \pi \i \, \mathsf{C}[\mu_{ \Ga_{\e{eq}} }]_-(z) + V^\prime(z) \, = \,  0,
& & \forall z \in \mathrm{supp} \, \mu_{ \Ga_{\e{eq}} } \setminus \{ z \in \Gamma \, : \, R(z) = 0 \}.
\end{align}
This equation compactly expresses both the Frostman conditions and the $S$-curve conditions,
since multiplying by $\gamma^\prime$ and taking the real and imaginary parts one finds the Euler-Lagrange conditions
\eqref{Frostman1} and \eqref{Frostman2} and the $S$-curve condition \eqref{Scurve} respectively.

If we multiply (\ref{EulerLagrange}) by a continuous function $f$ and integrate with respect to $\dd \mu_{ \Ga_{\e{eq}} }$, we find that
\begin{align}\label{fidentity}
- \frac{1}{2}\Int{ \Ga_{\e{eq}}^2 }{} \frac{f(z) - f(w)}{z-w} \, \dd \mu_{ \Ga_{\e{eq}} }(z) \, \dd \mu_{ \Ga_{\e{eq}} }(w)
\, + \,  \Int{ \Ga_{\e{eq}} }{}  V^\prime(z)f(z) \, \dd \mu_{ \Ga_{\e{eq}} }(z) = 0
\end{align}
for all continuous functions $f : \Ga_{\e{eq}} \longrightarrow \mathbb{C}$.

\begin{definition}[One-cut regular]\label{regular}
Let $V(z) = \frac{z^\kappa}{\kappa}+\e{O}(z^{\kappa-1})$ be a polynomial potential and $\Ga_{\e{eq}} \in \mathcal{T}(\alpha,\alpha^\prime)$ be a solution
of the associated max-min energy problem (\ref{maxmin}). Then we say that $V$ is \textnormal{one-cut regular} if
\begin{enumerate}[label=(\arabic*)]
\item The support of the equilibrium measure $\mu_{ \Ga_{\e{eq}} }$ on $\Ga_{\e{eq}}$ is connected.
\item It is possible to choose the curve $\Ga_{\e{eq}} \in \mathcal{T}(\alpha,\alpha^\prime)$ solving the max-min energy problem such that the polynomial
$R$ (defined by (\ref{EulerLagrange})) has simple zeros at the endpoints of the support of the equilibrium measure and no other zeros on $\Ga_{\e{eq}}$.
\end{enumerate}
\end{definition}

Moving forward, given an equilibrium measure $\mu_{ \Ga_{\e{eq}} }$ associated to the solution of the max-min energy problem there is a natural way of choosing the contour
in a neighbourhood of $\mathrm{supp} \, \mu_{ \Ga_{\e{eq}} }$ that has many desirable properties.  Let $V$ be one-cut regular and $\zeta_1, \zeta_2 \in \Ga_{\e{eq}}$ be respectively
the lower and upper endpoint of the support of the equilibrium measure according to the orientation on the contour given by $\dd \mu_{ \Ga_{\e{eq}} } (z) \, = \, \frac{ 1 }{ \i\pi} \sqrt{R(z)}_+\, \dd z > 0$.
Recall that $\sqrt{R(z)}$, as defined through \eqref{squareroot}, is an analytic function on $\Cx \setminus \e{supp} \mu_{ \Ga_{\e{eq}} } $.
Hence, given the smooth injective curve $\ga_{\e{eq}}$ realising $\Ga_{\e{eq}}=\ga_{\e{eq}}(\R)$, we identify $t_1$ as the unique solution to $\ga_{\e{eq}}(t_1)=\zeta_1$
and define $\Ga_{\e{eq}; \zeta_1 } \, = \, \{ \ga_{\e{eq}}(t) \, : \, t \in [t_1; + \infty[ \}  $ as the portion of the curve $ \Ga_{\e{eq}}$ located \textit{after}
the first endpoint of the support of the equilibrium measure.

Then, one can define the map
\begin{align}
\begin{split}
&\Theta :  \Cx \setminus   \Ga_{\e{eq}; \zeta_1} \longrightarrow \Cx \\
&\Theta(z) = \Int{\zeta_1}{z} \frac{ 1 }{ \i\pi } \sqrt{R(w)} \, \dd w.
\end{split}
\end{align}
in which the integration path from $\zeta_1$ to $z$ does not cross $\Ga_{\e{eq}; \zeta_1}$. $\Th$ is holomorphic on $\Cx \setminus   \Ga_{\e{eq}; \zeta_1}$
and admits cuts along $ \Ga_{\e{eq}; \zeta_1}$ with continuous $\pm$ boundary values. Since the integrand has a square root behaviour at $\zeta_1$ and $\zeta_2$,
$\Th$ exhibit $3/2$ root singularities at these points, see \textit{e.g.} \cite{DeiftKriechMcLaughVenakZhouOrthogonalPlyExponWeights}. More precisely,
there exist biholomorphisms $\varrho, \wt{\varrho}$ defined, resp., on $\op{D}_{0,\eta}$ and $\op{D}_{0,\wt{\eta}}$ with $\eta, \wt{\eta}>0$ such that
\beq
\Th(z) \, = \, - \i \big[ \varrho(z - \zeta_1) \big]^{\f{3}{2}} \;, \quad z \in \op{D}_{\zeta_1,\eta}
\qquad \e{and} \qquad
\Th(z) \, = \, \mf{s}(z) - \i \big[ \wt{\varrho}(z - \zeta_2) \big]^{\f{3}{2}} \;, \quad z \in \op{D}_{\zeta_2,\wt{\eta}}   \;,
\label{ecriture cptmt local Theta}
\enq
in which $\mf{s}(z)=1$ if $z$ located to the left of $\Ga_{\e{eq}}$ while $\mf{s}(z)=-1$ to $z$ located to the right of $\Ga_{\e{eq}}$.
There the $3/2$ roots are computed with respect to the principal branch of the logarithm, and the curves
$\e{supp} \mu_{\Ga_{\e{eq}} } \cap \op{D}_{\zeta_1, \eps_1}$ and $\e{supp} \mu_{\Ga_{\e{eq}}} \cap \op{D}_{\zeta_2, \eps_2}$
with $\eps_a>0$ and  small enough are respectively obtained by solving the equations
\beq
\varrho(z - \zeta_1) \, = \, -t \;, \quad \e{resp}. \quad  \wt{\varrho}(z - \zeta_2) \, = \, -t \;,   \quad \e{with} \quad t \in \intff{0}{ \wt{\eps} \, }
\enq
where $ \wt{\eps} >0$ and small enough.

Of particular interest, is the map obtained from taking
the $+$ boundary value on $ \mathrm{supp} \, \mu_{ \Ga_{\e{eq}} }$:
\begin{align}
\begin{split}
&\Theta_+ : \mathrm{supp} \, \mu_{ \Ga_{\e{eq}} } \longrightarrow \intff{0}{1}\\
&\Theta_+(z) = \Int{\zeta_1}{z} \frac{ 1 }{ \i\pi } \sqrt{R(w)}_+\, \dd w.
\end{split}
\end{align}
$\Theta_+$ is manifestly a bijection between $\mathrm{supp} \, \mu_{ \Ga_{\e{eq}} }$ and $\intff{0}{1}$. However,
$\Theta_+$ is \textit{not}  analytic on an open neighbourhood of $\mathrm{supp} \, \mu_{ \Ga_{\e{eq}} } $ because $\Theta(z)$ and $1-\Theta(z)$ vanish with a $\frac{3}{2}$ power at $z = \zeta_1$ and $z=\zeta_2$ respectively.
However, outside of $\zeta_1, \zeta_2$ it can be analytically continued by observing that starting from \eqref{ecriture decomposition canonique pour racine de R} one may recast its integral representation as
\beq
\Theta_+(z) \, = \, \Int{\zeta_1}{z} \frac{ 1 }{ \i\pi } S(w) \big[ (z-\zeta_1)(z-\zeta_2) \big]_{ \sg_{\e{eq}; \downarrow} }^{1/2} \, \dd w
\enq
where $\sg_{\e{eq}; \downarrow}$ is a slight deformation of the contour $\sg_{\e{eq}}$ away to the right of it and $S$ is as introduced in \eqref{ecriture decomposition canonique pour racine de R}.
 The above is manifestly analytic in an open neighbourhood of $z \in \sg_{\e{eq}}\setminus \{\zeta_1, \zeta_2\}$.

\vspace{2mm}

Now, one can improve the behaviour at the endpoints $\zeta_1, \zeta_2$ by introducing an appropriate reparametrisation of $\intff{0}{1}$ by means of a $+$ boundary values of a function
having $3/2$-singularities at the endpoints. This allows one, all-in-all, to cancel these. Let

\begin{align}\label{psifunction}
\begin{split}
&\Psi : \Cx \setminus \R^+ \longrightarrow \Cx \\
&\Psi(z) = \frac{8}{\i \pi }\Int{0}{z}  w^{\f{1}{2}}(w-1)^{\f{1}{2}}  \, \dd w \, ,
\end{split}
\end{align}
where we use the principal branch of the logarithm to define the power laws. $\Psi$ is analytic on its domain of definition and
there exist biholomorphisms $\tau, \wt{\tau}$ defined, resp., on $\op{D}_{0,\eta^{\prime}}$ and $\op{D}_{0,\wt{\eta}^{\,\prime}}$ with $\eta^{\prime}, \wt{\eta}^{\, \prime}>0$ such that
\beq
\Psi(z) \, = \, - \i \big[ \tau(z) \big]^{\f{3}{2}} \;, \quad z \in \op{D}_{0,\eta^{\prime}}
\qquad \e{and} \qquad
\Psi(z) \, = \, \e{sgn}\big(\Im z \big) - \i \big[ \wt{\tau}(z - 1) \big]^{\f{3}{2}} \;, \quad z \in \op{D}_{0,\wt{\eta}^{\, \prime}}   \;.
\label{ecriture cptmt local Psi}
\enq
The $+$ boundary value of $\Psi$ on $\intff{0}{1}$ is a bijection onto the same interval.
\begin{align}\label{psifunction}
\begin{split}
&\Psi_+ : \intff{0}{1} \longrightarrow \intff{0}{1} \\
&\Psi_+(z) = \frac{8}{\i \pi }\Int{0}{z}  w^{\f{1}{2}}(w-1)^{\f{1}{2}}_+  \, \dd w \, .
\end{split}
\end{align}
Just as for $\Th_+$, $\Psi_+$ admits an analytic continuation onto an open neighbourhood of any $z \in \intoo{0}{1}$.

It is direct to infer from the above discussion that
\begin{align}\label{gammadef}
\gamma = \Theta^{-1}_+\circ \Psi_+ : \intff{0}{1} \longrightarrow \mathrm{supp}\,  \mu_{ \Ga_{\e{eq}} }
\end{align}
is a bijection.

\begin{lemma}
\label{Lemma extension analytique courve gamma}

The map $\gamma$ introduced in \eqref{gammadef} extends into a biholomorphism from an open neighbourhood $U_{\ga}$ of $\intff{0}{1}$ onto $\ga\big( U_{\ga} \big)$.
$U_{\ga}$ contains $\mathrm{supp}\,  \mu_{ \Ga_{\e{eq}} }$ in its strict interior. In particular, $\gamma(x)-\zeta_1$ has a simple zero at $x=0$ and $\gamma(x)-\zeta_2$ has a simple zero at $x=1$.

\end{lemma}

\begin{proof}

It follows from the local behaviour of $\Psi$ \eqref{ecriture cptmt local Psi} around $0$ and $1$ and of $\Th$ \eqref{ecriture cptmt local Theta} around $\zeta_1$ and $\zeta_2$ that
\beq
\ga(u) \, = \, \varrho^{-1}\big(\tau(u)\big) \, + \, \zeta_1 \qquad \e{resp.} \qquad
\ga(u) \, = \, \wt{\varrho}^{-1}\big(\wt{\tau}(u-1)\big) \, + \, \zeta_2
\nonumber
\enq
provided that $|u|$, resp. $|u-1|$, is small enough. This ensures that $\ga$ extends to an analytic map in some open neighbourhood of $0$ and $1$.

Let $t_0 \in \intoo{0}{1}$ and let $z_0, s_0$ be defined through $z_0=\Th_+^{-1}(t_0)$ and $s_0 = \Psi_+^{-1}(t_0)$.
Since $\Xi_+^{\prime}(z_0) \not=0$, resp. $\Psi_{+}^{\prime}(s_0)\not=0$, and the function is analytic in some open neighbourhood
of $z_0$, resp. $s_0$, there exist local biholomorphisms $\Ups, \chi$ mapping $0$ to $0$ such that
\beq
\Xi_+(z) \, = \, \underbrace{ \Xi_+(z_0) }_{=t_0} \, + \, \Ups(z-z_0)  \qquad \e{and} \qquad
\Psi_+(s) \, = \, \underbrace{ \Psi_+(s_0) }_{=t_0} \, + \, \chi(s-s_0)  \;,
\enq
on some neighbourhoods of $z_0$ and $s_0$. Hence, for $|s-s_0|$ small enough, it holds
$$\ga(s) \; = \; z_0 + \Ups^{-1}\circ \chi(s-s_0)\, . $$

\end{proof}

From now on, we take $\wt{\eps}>\eps>0$ small enough so that
$$\intff{-\eps}{ 1 +\eps} \subset \intff{-\wt{\eps} }{ 1 + \wt{\eps} \,  } \subset U_{\ga} \, . $$
Below, we will show that we can take $\gamma( \intff{-\epsilon}{1+\epsilon})$ to be a portion of our desired contour $\Gamma_{\e{eq}}$ and we refer to $\gamma$ as the \textit{analytic parametrisation} on this segment.
To "complete" the contour we will need to connect $\gamma(\intff{-\epsilon}{1+\epsilon})$ to infinity in such a way that the resulting contour has all the desired properties. The result of \cite{lirias1643667}
only ensures that any contour solving the max-min energy is at least a union of finitely many $\mc{C}^1$ arcs; however for our problem, it will be convenient to improve on certain properties of the curve.
\begin{proposition}[Existence of desired contour]\label{contourexistence}
Let $V$ be one-cut regular. There exists an admissible contour $\Gamma_{\e{eq}} \in \mathcal{T}(\alpha,\alpha^\prime)$ with the following properties.

\begin{enumerate}[label=(\arabic*)]
\item $\Gamma_{\e{eq}}$ satisfies \textnormal{\ref{hypothesis2}} and \textnormal{\ref{hypothesis3}}.
\item $\Gamma_{\e{eq}}$ solves the max-min energy problem (\ref{maxmin}) in external field $\varphi = \Re (V)$.
\item $\gamma( \intff{-\epsilon}{1+\epsilon}) \subset \Gamma_{\e{eq}}$, where $\gamma$ is the analytic extension of (\ref{gammadef}).
\item For any $z \in \Gamma_{\e{eq}} \setminus \mathrm{supp}\, \mu_{ \Ga_{\e{eq}} }$, $\varphi_{\Ga_{\e{eq}};\mathrm{eff}}(z) > 0$, where the effective potential is defined by \eqref{effective}.
\end{enumerate}
\end{proposition}
\begin{proof}

See Appendix \ref{curvecontruction}.

\end{proof}

\begin{definition}[Complex effective potential]
Let $\Ga_{\e{eq}}$ the contour provided by Proposition \ref{contourexistence}.
Let $\zeta_1, \zeta_2 \in \Gamma_{\e{eq}}$ be the lower and upper endpoints of $\mathrm{supp}\, \mu_{ \Gamma_{\e{eq}} } $ according to the orientation $\frac{1}{\i\pi}\sqrt{R(z)}\, \dd z > 0$.
The \textnormal{complex effective potential} is given by
\begin{align}
&\Phi_{ \Ga_{\e{eq}}; \mathrm{eff} } : \mathbb{C}\setminus \Gamma_{\e{eq}} \longrightarrow \mathbb{C},  &\Phi_{\Ga_{\e{eq}}; \mathrm{eff}}(z) \overset{\mathrm{def}}{=} \int_{\zeta_1}^z \sqrt{R(w)}\, \dd w \label{defphi}
\end{align}
where the integration path from $\zeta_1 \in \mathbb{C}$ to $z\in \mathbb{C}$ is taken arbitrarily as long as it doesn't intersect with $\Gamma_{\e{eq}}$ away from $\zeta_1$.
\end{definition}
It is natural to think of $\Phi_{\Ga_{\e{eq}};\mathrm{eff}}$ as a complexification of the effective potential \eqref{effective} because
$\Re \big[ \Phi_{\Ga_{\e{eq}};\mathrm{eff}} \big] = \varphi_{\Ga_{\e{eq}};\mathrm{eff}}$. Note that
$$\Phi_{\Ga_{\e{eq}};\mathrm{eff}}(z) = V(z) + g_{\Ga_{\e{eq}}}[ \mu_{ \Ga_{\e{eq}} }](z) - C$$
for some constant $C \in \mathbb{C}$, determined by $\Phi_{\mathrm{eff}}(\zeta_1) = 0$,
and where $ g_{\Ga}[ \mu]$ was introduced in Definition \ref{complexenergy}. This can be seen from the fact that the derivatives of both sides are equal by \eqref{Rformula}.
\begin{proposition}\label{steepestdescent}
It holds that $\forall z \in \gamma( \intff{-\epsilon}{1+\epsilon})$
\begin{align}\label{Imconstant}
\Im \left(  \big[ \Phi_{\Ga_{\e{eq}};\mathrm{eff}}\big]_+(z) \, +\, \big[ \Phi_{\Ga_{\e{eq}};\mathrm{eff}} \big]_- (z) \right) = 0  \, ,
\end{align}
where $\pm$ refers to the left and right boundary values according to the orientation of $ \gamma(\intff{-\epsilon}{1+\epsilon})$.
Furthermore, $\forall z \in \gamma\big( \intff{-\epsilon}{1+\epsilon}\setminus \intff{0}{1} \big)$,
\begin{align}\label{effectivepositive}
\Re \left(\big[ \Phi_{\Ga_{\e{eq}};\mathrm{eff}}\big]_+(z) +\big[ \Phi_{\Ga_{\e{eq}};\mathrm{eff}}\big]_-(z) \right) > 0  \,.
\end{align}
\end{proposition}
\begin{proof}
Recall that $\big[ \Phi_{\Ga_{\e{eq}};\mathrm{eff}}\big]_+(\zeta_1) = \big[ \Phi_{\Ga_{\e{eq}};\mathrm{eff}}\big]_-(\zeta_1) = 0$, hence we need only to show that the left hand side of \eqref{Imconstant}
is constant on $ \gamma\big(\intff{-\epsilon}{1+\epsilon}\big)$.  Observe that on $ \e{supp} \mu_{\Ga_{\e{eq}}}=\gamma(\intff{0}{1})$, the maps $z \mapsto \big[ \Phi_{\Ga_{\e{eq}};\mathrm{eff}}\big]_{\pm}(z)$ are
continuously differentiable and furthermore that
$$ \big[ \Phi_{\Ga_{\e{eq}};\mathrm{eff}}\big]_+^{\prime}(z) \, + \, \big[ \Phi_{\Ga_{\e{eq}};\mathrm{eff}}\big]_-^{\prime}(z) \, = \, \sqrt{R(z)}_+  \, +  \, \sqrt{R(z)}_- \, =\,  0 \;.$$
For $z \in \gamma\big( \intff{-\epsilon}{1+\epsilon}\setminus \intff{0}{1} \big)$, we observe that
$$\big[ \Phi_{\Ga_{\e{eq}};\mathrm{eff}}\big]_+(z) +\big[ \Phi_{\Ga_{\e{eq}};\mathrm{eff}}\big]_-(z)  \, = \,  \i\pi  \big\{ \Theta_+(z) \, + \, \Th_- (z) \big\} \, = \,
\i \pi \big\{ \Psi_+(\gamma^{-1}(z)) \, + \, \Psi_-(\gamma^{-1}(z)) \big\} \;. $$
However, for $x<0$, it holds
$$
\i \pi \big\{ \Psi_+(x) \, + \, \Psi_-(x) \big\}  \; = \; -16 \Int{0}{x} \sqrt{w(w-1)} \dd w \, > \, 0
$$
while, for  $x>1$, it holds
$$
\i \pi \big\{ \Psi_+(x) \, + \, \Psi_-(x) \big\}  \; = \; 16 \Int{1}{x} \sqrt{w(w-1)} \dd w \, > \, 0 \;.
$$
This entails the claim.
\end{proof}

Proposition \ref{steepestdescent} should be regarded as showing that $\gamma\big( \intff{-\epsilon}{1+\epsilon} \big)$ plays a role analogous to a steepest descent contour.

\vspace{2mm}

In our approach it is essential to interpolate the partition function we are interested in with one that we can compute.  In past works on $\beta$-ensembles \cite{Borot:2013aa,BG24}
a rather simple interpolation was possible: by matching endpoints, a convex combination of potentials yielded a convex combination of equilibrium measures.
Thus, the intermediate potential $V_t$ and equilibrium measure were both linear in the interpolation parameter $t \in \intff{0}{1}$.  Unfortunately such a simple solution cannot
be made here because the initial and final equilibrium measures live on different curves and so it would make no sense to take convex combinations of the two.
We instead use a more complicated construction which starts by finding an interpolation between the initial and the final $S$-curves and then works backwards to find the family of potentials that gives rise to it.

\vspace{2mm}

Hereafter in this section, $\gamma$ will be the path defined by (\ref{gammadef}),  or more accurately its analytic extension to $\intff{-\epsilon}{1+\epsilon}$.
Let
\begin{align}\label{contourdef}
\gamma_t(x) \overset{\mathrm{def}}{=} \frac{\gamma(tx)(\zeta_2-\zeta_1)+\zeta_1(\gamma(t)-\zeta_2)}{\gamma(t)-\zeta_1} & & t \in \intff{0}{1}, \, \, x \in \intff{-\epsilon}{1+\epsilon} \, .
\end{align}
Note that for every $t\in \intff{0}{1}$, $\gamma_t(0)=\zeta_1$ and $\gamma_t(1)=\zeta_2$.
When $t = 1$, $\gamma_1 = \gamma$. When $t \downarrow 0$, l'Hôpital's rule shows that $\gamma_0(x) = x (\zeta_2-\zeta_1)+\zeta_1$,
which is simply a straight line path from $\zeta_1$ to $\zeta_2$ (recall that $\gamma$ is analytic and so the relevant derivatives exist).

Let us now reverse the logic and find an equilibrium measure and potential associated to the interpolating curve $\gamma_t\big( \intff{-\epsilon}{1+\epsilon}\big)$.
Let us denote the interpolating equilibrium measure as $\frac{1}{\i\pi}\sqrt{R_t(z)}_{+} \,\dd z$ and the potential as $V_t$. $R_t$ or $V_t$ will turn out to be
certain analytic functions in some open neighbourhood of  $\gamma_t\big( \intff{-\epsilon}{1+\epsilon}\big)$ but, in principle, neither of these
will be a polynomial. Note that because $\gamma$ is invertible on the neighbourhood $U_{\ga}$ of $\intff{-\epsilon}{1+\epsilon}$,
 $\gamma_t$ is invertible on $U_{\ga}/t \supset \intff{-\epsilon}{1+\epsilon}$.
 To start with, we define $R_t$  for all $t\in \intff{0}{1}$  on
\beq
\ga_t\Big( \f{1}{t} U_{\ga} \Big) \, = \, \bigg\{   \frac{\gamma(x)(\zeta_2-\zeta_1)+\zeta_1(\gamma(t)-\zeta_2)}{\gamma(t)-\zeta_1}  \; : \;  x \in U_{\ga}   \bigg\}
\enq
through
\begin{align}
  R_t(z)   \,  = \,  64  \f{ \gamma_t^{-1}(z)\big( \gamma_t^{-1}(z) - 1 \big) }{  \big[ \gamma_t^\prime\big( \gamma_t^{-1}(z)\big)  \big]^2} \;.
\label{expression pour Rt de z}
\end{align}
This equation is true for $t=1$ by \eqref{gammadef}, \textit{i.e.} yields $R_1=R$.
$R_t$ so defined is hence analytic in both $t \in \intff{0}{1}$ and $z \in  \ga_t\big(  U_{\ga}/t \big)$.
$R_t$ admits two zeroes on $\ga_t\big( U_{\ga}/t \big)$, both simple. One is located at $z=\zeta_1$ and the other at $z=\zeta_2$.
$ \big[ \gamma_t^\prime\big( \gamma_t^{-1}(z)\big)  \big]^2 \cdot  R_t(z)\leq  0$ on $\ga_t\big( \intff{0}{1} \big)$ and this is the only subset of $\ga_t\big(  U_{\ga}/t \big)$ where it takes real negative values.
We then define $\sqrt{R_t(z)}$ as having a cut along  $\ga_t\big( \intff{0}{1} \big)$. It is then direct to check that
\beq
\frac{1}{ \i\pi} \sqrt{R_t(\gamma_t(x))}_+ \gamma_t^\prime(x)  = \frac{8}{\pi}  \sqrt{x(1-x)}
\label{Rtdef}
\enq
which shows that as $t \downarrow 0$ the equilibrium measure tends towards the Wigner semicircle law.
\begin{remark}\label{semicircleremark}
Note that $\frac{1}{\i\pi} \sqrt{R_t(\gamma_t(x))}_+ \gamma_t^\prime(x) \, \dd x$ is precisely the pullback of the measure
$\frac{1}{\i\pi} \sqrt{R_t(z)} \, \dd z$ from the curve to the real line under $\gamma_t$. Thus \eqref{Rtdef} states that the pullback of this
measure under $\gamma_t$ always yields the semicircular distribution on $\intff{0}{1}$ independently of $t \in \intff{0}{1}$:
\beq
\dd \nu_{\mathrm{sc}}(x) = \frac{8}{\pi}\mathbbm{1}_{ \intff{0}{1} }(x) \sqrt{x(1-x)} \, \dd x \;.
\label{definition mesure semi cercle sur 0 1}
\enq
\end{remark}
Next, let us reconstruct the potential that gives rise to this equilibrium measure.  To define $V_t$ we enforce that \eqref{squareroot} is true, namely we set
\beq\label{Vtdef}
 V_t^\prime(z)  \, = \,  \sqrt{ R_t(z)} \, - \,  \frac{1}{ \i \pi } \hspace{-2mm} \Int{ \ga_t(\intff{0}{1}) }{}  \hspace{-2mm} \frac{ \sqrt{R_t(w)}_+ }{ w - z }\, \dd w \;.
\enq
The above allows one to provide an alternative expression for $V_t^{\prime}$ where its analyticity in the neighbourhood of $\intff{-\eps}{1+\eps}$ is manifest.
However, for that, we need a more structured expression for $\sqrt{R_t(z)}$.
\begin{lemma}
\label{Lemma representation poiur sqrt de Rt}
There exists $S_t(z)$ jointly analytic in $t \in \intff{0}{1}$ and $z \in \gamma_t\big(\intff{-\epsilon}{1+\epsilon}\big)$ such that
\beq
 \sqrt{R_t(z)} \, = \,  \big[(z-\zeta_1)(z-\zeta_2)\big]^{1/2}_{\sg_{t} } \cdot S_t(z)  \qquad where \qquad   \sg_{t} \; = \; \ga_t\big( \intff{0}{1} \big) \, ,
\enq
the square root is defined as in \eqref{definition racine avec coupure sur une courbe prescrite} while
\beq\label{prefactor}
S_t(z) = S(z_t)  \cdot \bigg( \frac{\gamma(t) - \zeta_1}{t(\zeta_2-\zeta_1)} \bigg)^{\f{3}{2}} \sqrt{\frac{1-\gamma_t^{-1}(z)}{\zeta_2 - z}} \sqrt{\frac{\zeta_2-z_t}{1- \gamma^{-1}(z_t)}} \, ,
\enq
where $S(z)$ has been introduced in \eqref{ecriture decomposition canonique pour racine de R} and  $z_t := \frac{z(\gamma(t) - \zeta_1)-\zeta_1(\gamma(t) - \zeta_2)}{\zeta_2 - \zeta_1}$.

In particular,  $\Dp{t} S_t(z)$ is again jointly analytic in $t \in \intff{0}{1}$ and $z\in\gamma_t( \intff{-\epsilon}{1+\epsilon})$.
\end{lemma}
\begin{proof}
We begin with formula \eqref{expression pour Rt de z} and observe that  $\gamma^{-1}_t(z) = \frac{1}{t}\gamma^{-1}(z_t)$. Since
$$
\ga_t^{\prime}(x) \, = \, t \ga^{\prime}(tx) \cdot \f{ \zeta_2-\zeta_1 }{ \ga(t) - \zeta_1 } \;,
$$
one arrives to
$$
R_t(z) \, = \,  64  \,  \ga^{-1}(z_t) \, \f{  \ga_t^{-1}(z)-1 }{ t^3 \big[ \ga^{\prime}\big( \ga^{-1}(z_t) \big) \big]^2 }    \cdot \bigg( \f{\ga(t)-\zeta_1  }{ \zeta_2-\zeta_1} \bigg)^2 \;.
$$
By taking the expression for $R_1(z)=R(z)$ and using \eqref{ecriture decomposition canonique pour racine de R}, one infers that
$$
\f{ 1 }{ \big[ \ga^{\prime}\big( \ga^{-1}(z_t) \big) \big]^2 } \; = \; \f{ S^2(z_t) (z_t-\zeta_2) (z_t-\zeta_1)   }{  64 \, \ga^{-1}(z_t) \big(  \ga^{-1}(z_t) - 1 \big)  } \;.
$$
All-in-all, this leads to $R_{t}(z) \, = \,  S_t^2(z) \, \big( z-\zeta_1\big) \big( z-\zeta_2\big)$. Upon taking the square root
cut along $\sg_{t}=\ga_t\big( \intff{0}{1} \big)$, one gets the claim upon fixing the sign by taking the $t\rightarrow 1$ limit.
In order to deduce analyticity, one observes that  $(t,z) \mapsto z_t$ is clearly jointly analytic,
and furthermore the functions $\frac{1-\gamma^{-1}_t(z)}{\zeta_2 - z}$ and $\frac{ \zeta_2-z_t}{1-\gamma^{-1}(z_t)}$ are jointly analytic in $z$ and $t$ and have no zeroes or poles
so that their logarithms  are well-defined.
\end{proof}

Starting from the representation obtained in Lemma \ref{Lemma representation poiur sqrt de Rt}, \eqref{Vtdef} reduces to
\begin{align}\label{interpolatingpotential}
V_t^\prime(z) \, =  \,   S_t(z) \left(z - \frac{\zeta_1+\zeta_2}{2}\right) \, -
   \Int{ \sg_t }{} \big[(w-\zeta_1)(w-\zeta_2)\big]^{1/2}_{\sg_{t};+ } \, \frac{S_t(w)-S_t(z)}{w-z} \, \dd w\,,
\end{align}
with $\sg_t= \ga_t(\intff{0}{1})$. It is clear from the above that $V_t^\prime$ is analytic on a neighbourhood of the curve
$\gamma_t\big( \intff{-\epsilon}{1+\epsilon}\big)$. Finally, we set
\begin{align}\label{interpolatingpotential2}
V_t(z) :=V\left(\frac{\zeta_1+\zeta_2}{2}\right)  \, + \,  \Int{ \frac{\zeta_1+\zeta_2}{2} }{ z} V_t^\prime(w) \, \dd w.
\end{align}
It is direct that $V_t(z)$ is jointly analytic in both $t$ and $z$. Moreover, $V_1=V$ while
\beq
V_0(z) = \frac{4}{(\zeta_2-\zeta_1)^2}\left( z- \frac{\zeta_1+\zeta_2}{2} \right)^2 + V\left(\frac{\zeta_1+\zeta_2}{2}\right) \; .
\label{definition potentiel V0}
\enq
\begin{lemma}\label{curvelowerbound}
For any $p \in \mathbb{N}$, there exists $C_p>0$ such that
\beq
\norm{ \ga_t }_{ \mc{W}^{\infty}_p( \intff{-\eps}{1+\eps }) } \, \leq \, C_p \qquad \mathrm{uniformly} \; \mathrm{in} \;\; t \in \intff{0}{1} \;.
\enq
All the derivatives of $\gamma_t$ are bounded with bounds independent of $t \in \intff{0}{1}$. Furthermore, there exists a constant $0 < C_\gamma < 1$ independent of $t \in \intff{0}{1}$ such that
\begin{align}
0 < C_\gamma \leq \frac{ |\gamma_t(x) - \gamma_t(y)| }{ |x-y| } \leq C_\gamma^{-1}, & & \forall x,y \in \intff{-\epsilon}{1+\epsilon}\, .
\end{align}
\end{lemma}
\begin{proof}
The first claim is trivial. For the second claim, observe that
$$\frac{|\gamma_t(x) - \gamma_t(y)|}{|x-y|} = |\zeta_2 - \zeta_1| \frac{t}{|\gamma(t) - \zeta_1|} \frac{|\gamma(tx) - \gamma(ty)|}{|tx-ty|} .$$
But $\frac{|\gamma(tx) - \gamma(ty)|}{|tx-ty|}$ is bounded above and below by constants independent of $t$ by compactness of $\intff{-\epsilon}{1+\epsilon}$ and injectivity of $\gamma$.
Finally, the singularity of $\frac{t}{|\gamma(t) - \zeta_1|}$ at $t=0$ is removable because $\gamma^\prime(0) \neq 0$.
\end{proof}

\begin{corollary}\label{Vprimebound}
There is a constant $C>0$ independent of $t \in \intff{0}{1}$ such that $|V^\prime_t\big( \gamma^\prime_t(x) \big)  \gamma^\prime_t (x)| \leq C$ for all $x \in \intff{-\epsilon}{1+\epsilon}$.
\end{corollary}
\begin{proof}
Combining (\ref{Rtdef}) and (\ref{Vtdef}) we find the formula
\begin{align*}
V^\prime_t(\gamma_t(x)) \gamma^\prime_t(x) = 8\left(x - \frac{1}{2}\right) \,  -  \, \frac{8}{\pi}\int_0^1 \sqrt{y(1-y)} \left( \frac{\gamma^\prime_t(x)}{\gamma_t(y) - \gamma_t(x)} - \frac{1}{y-x} \right)\, \dd y \, .
\end{align*}
The claim follows from the observation that $\frac{\gamma^\prime_t(x)}{\gamma_t(y) - \gamma_t(x)} - \frac{1}{y-x}$ is bounded by a constant independent of $t \in \intff{0}{1}$.

\end{proof}

The next proposition shows that upon choosing $V_t$ as described above and considering the associated external field $\varphi_t = \Re (V_t)$, the natural candidate for the equilibrium measure of $\op{I}_{\Sg_t}$
with $\Sigma_t = \gamma_t\big( \intff{-\epsilon}{1+\epsilon} \big)$
 indeed yields the equilibrium measure.

\begin{proposition} Let
\begin{align}\label{tequilmeas}
\dd \mu_{\Sg_t} (z) \overset{\mathrm{def}}{=} \mathbbm{1}_{  \sg_t }(z)\frac{1}{\i\pi}\sqrt{R_t(z)}_+ \, \dd z  \qquad \mathrm{with} \qquad \sg_t \, = \, \gamma_t\big(\intff{0}{1}\big)
\end{align}
and $\Sigma_t = \gamma_t\big( \intff{-\epsilon}{1+\epsilon} \big)$. $ \mu_{\Sg_t}$ is the equilibrium measure of the energy functional $\op{I}_{\Sg_t}$
in the external field $\varphi_t = \Re (V_t)$.  Furthermore, $\Sigma_t$  is an $S$-curve in the external field $\varphi_t$.
\end{proposition}
\begin{proof}
By construction $\dd \mu_{\Sg_t} $ is a probability measure that solves the complexified Euler-Lagrange equations (\ref{EulerLagrange}).
One thus needs to check the positivity properties of the associated effective potential.

Define $\Phi_{\Sg_t; \mathrm{eff}}(z) \; = \; \Int{ \zeta_1 }{ z } \sqrt{R_t(s)} \cdot \dd s $ with $z$ belonging to a small open neighbourhood of
$\Sg_t$ split in two by $\ga_t\big( \intff{ - \eps^{\prime} }{ 1 + \eps^{\prime} } \big)$ with $\eps^{\prime}>\eps$. The path of integration joining $\zeta_1$ to $z$
is always located in the side of $\ga_t\big( \intff{ - \eps^{\prime} }{ 1 + \eps^{\prime} } \big)$ which contains $z$. Then, one has that
$$\Big[\Phi_{\Sg_t; \mathrm{eff}}\Big]_+(\gamma_t(x)) \,= \, \Int{0}{x} \sqrt{ R_t\big(\ga_t(s)\big) }_+ \ga_t^{\prime}(s) \dd s
\; = \; 8\Int{0}{x} s^{\f{1}{2}}_+ (s-1)^{\f{1}{2}}_+ \dd s $$
is independent of $t \in \intff{0}{1}$. Since, by construction,
$\Re\Big( \Phi_{\Sg_t; \mathrm{eff}}(z) \Big) \, = \, \varphi_{ \Sg_t; \e{eff} }(z)$, one gets for $x\leq 0$ that
$$ \varphi_{\Sg_t;\e{eff}}(\gamma_t(x)) \; = \;     8    \Int{x}{0} \sqrt{y(y-1)} \, \dd y \geq 0 \qquad \e{and} \qquad
\varphi_{\Sg_t;\e{eff}}(\gamma_t(x)) \; = \;     8    \Int{1}{x} \sqrt{y(y-1)} \, \dd y  \geq 0 \, , $$
for $x\geq 1$.
Hence, by Lemma \ref{frostmanlemma}, $\dd \mu_{\Sg_t}$ is the unique equilibrium measure on $\Sigma_t$ in external field $\varphi_t$.

\end{proof}

\begin{definition}[Contours]\label{contours} Let $\Gamma_1$ be the contour constructed in Lemma \ref{contourexistence}, $\Gamma_0$
be the infinite straight line passing through $\zeta_1$ and $\zeta_2$, and let
$$\Sigma_t \overset{\mathrm{def}}{=} \gamma_t\big( \intff{-\epsilon}{1+\epsilon} \big)\, .$$
Thus, in particular, $\Sigma_0 \subset \Gamma_0$ and $\Sigma_1 \subset \Gamma_1$.
\end{definition}

\section{Concentration of measure}\label{concentrationsection}

Throughout this section, we focus on a curve  $\eta_t : \mathcal{J}  \longrightarrow \mathbb{C}$ with $\mc{J}$ a compact interval
as introduced in \eqref{contourdef}. The curve $\eta_t$ is assumed to be simple, \textit{i.e.} non-self-intersecting, and $\mc{C}^\infty$.
The curve depends on an auxiliary parameter $t\in \intff{0}{1}$ in a smooth way and such that there exist constants $C_{\eta}>0$ and $\{ C_p \}_{p\geq 0}$ such that
\beq
\|  \eta_t \|_{ \mc{W}^{\infty}_p(\mc{J}) } \, \leq  \,  C_p  \qquad \e{and} \qquad C_{\eta} \, < \, \Big| \f{ \eta_t(x) - \eta_t(y) }{ x - y } \Big| \, < \, C_{\eta}^{-1} \, ,
\label{ecriture bornitude norme eta t et ctes lipschitz}
\enq
this uniformly in $t \in \intff{0}{1}$. Moreover, we shall assume that the curve admits an extension to a strictly larger interval $\wt{\mc{J}}$
\textit{i.e.} $\mc{J} \subset O \subset \wt{\mc{J}}$ with $O$ open in $\R$, satisfying the same properties.

In later sections, we will be basically interested in two cases. The first one when $\eta_{t}=\ga$ for $t \in \intff{0}{1}$
where $\ga$ is the curve associated with $\Ga_{\e{eq}}$ as provided by Proposition \ref{contourexistence} and $\mc{J}= \intff{-K}{K}$ with $K$ very large.
The second one when $\eta_{t}=\ga_t$ for $t \in \intff{0}{1}$
where $\ga_t$ is the curve defined in \eqref{contourdef} and $\mc{J}=\intff{-\eps}{1+\eps}$ .

We shall consider the real model (\ref{realmodel}) associated to $\eta_t$:
\beq
\dd \mathbb{P}_{N,\eta_t}(\mathbf{x})  \overset{\mathrm{def}}{=}  \frac{1}{\op{Z}_{N,\eta_t}[W_t] }
\pl{ i < j}{N}|\eta_t(x_i) - \eta_t(x_j)|^\beta  \cdot \pl{k=1}{N}\Big\{ |\eta^\prime_t(x_k)| \cdot  \ex{-N \beta \widetilde{\varphi}_t(x_k) }  \Big\} \cdot \dd \mathbf{x}\, ,  \qquad   \mathbf{x} \in \mathcal{J}^N
\label{definition modele reel avec courbe eta t et pot Wt}
\enq
where
\beq
\op{Z}_{N,\eta_t}[W_t] =   \Int{\mathcal{J}^N}{} \pl{  i < j }{N} |\eta_t(x_i) - \eta_t(x_j)|^\beta   \cdot
\pl{k=1}{N}\Big\{ |\eta_t^\prime(x_k)| \cdot  \ex{-N \beta \widetilde{\varphi}_t(x_k) }  \Big\}  \cdot \dd \mathbf{x} \, .
\enq
Above and in the following, let us fix the shorthand notation $\widetilde{\varphi}_t \, = \Re \big( W_t \circ \gamma_t \big)$.
For the purpose of this section, the potential $W_t$ will be an analytic function in the neighbourhood of the curve $\eta_t(\mc{J})$
that is also jointly analytic in $t\in \intff{0}{1}$. In particular, it will be such that
\beq
\norm{ W_t }_{\mc{W}^{\infty}_p( \eta_t(\mc{J})) } < C_p \,
\label{ecriture borne sur norme W infty du potentiel}
\enq
for some $C_p$ that is uniform in $t\in \intff{0}{1}$. Eventually, we will be interested in two cases: $W_t=V$, with $V$ the polynomial
of \eqref{definition pot polynomial interet}, and $W_t=V_t$ where $V_t$ is the interpolating potential introduced in \eqref{interpolatingpotential2}.
We will show in Section \ref{maxeigsection} that the compact model where $(\eta_t, W_t)=(\ga,V)$ is a good approximation for
the original model and that it allows one to access to the relevant part of the large-$N$ behaviour of the
interpolation in $t$ between the partition function subordinated to original potential $V$ and one
subordinate to a quadratic one.  The method of this section follows closely Section 4.2 of \cite{Gu19}.

The problem on the interval $\mc{J}$ naturally leads to the energy functional
\begin{align}
\begin{split}\label{Jfunctional}
&\mf{I}_{\mc{J}} : \mathcal{P}(\mathcal{J}) \longrightarrow (-\infty,+\infty] \\
&\mf{I}_{\mc{J}}[\nu] \overset{\mathrm{def}}{=} \int_{\mathcal{J}^2} \left(\ln \frac{1}{|\eta_t(x) - \eta_t(y)|} +  \widetilde{\varphi}_t(x) +  \widetilde{\varphi}_t(y) \right) \, \dd \nu(x)  \dd \nu(y).
\end{split}
\end{align}
$\mf{I}_{\mc{J}} $ is essentially the same functional as $\op{I}_{\Om_t}$, with $\Om_t\, = \, \eta_t\big( \mc{J}\big)$
since the two are related by pulling back measures from the curve to $\mc{J}$ under $ \eta_t$, hence by Lemmata \ref{goodrate} and \ref{continuousdensity},
$\mf{I}_{\mc{J}} $ has a unique compactly supported minimiser given by the pull back measure $ \wt{\mu}_{ \Om_t } \, = \, \eta_t^{\#}[ \mu_{ \Om_t }]$. In particular, when
$(\eta_t, W_t)\, =\, (\ga_t, V_t)$, the latter is readily seen to be given by $\dd \nu_{\mathrm{sc}}(x) = \frac{8}{\pi}\mathbbm{1}_{ \intff{0}{1} }(x) \sqrt{x(1-x)} \, \dd x$ (recall Remark \ref{semicircleremark}).
Throughout this section we will make the additional assumption that $\mc{J}$ strictly contains $\e{supp} \big[  \wt{\mu}_{ \Om_t }  \big]$, \textit{i.e.} that there exists an open set $O$ such that
$ \e{supp} \big[ \wt{\mu}_{ \Om_t }  \big] \subset  O  \subset \mc{J}$.

\begin{lemma}\label{logquantiles}
Let $\nu \in \mathcal{P}(\mathcal{J})$ be a compactly supported Borel probability measure on $\mathbb{R}$, absolutely continuous with respect to Lebesgue measure
and with bounded density.  Let
\begin{align*}
\xi_k \, =\,  \sup\left\{ s \in \mathbb{R} \, : \, \nu\big[ \intof{-\infty}{s} \big] \leq \frac{k}{N} \right\} & & k = 1, \dots, N-1
\end{align*}
be the $\frac{1}{N}$ quantiles.  Let $\xi_0 = \inf \mathrm{supp}\, \nu$ and $\xi_N = \sup \mathrm{supp}\, \nu$. Then
$$\frac{1}{N^2} \sum_{\substack{i,j=1 \\ i \neq j}}^N \ln |\eta_t(\xi_i) - \eta_t(\xi_j)| \, = \,
                    \Int{\mathcal{J}^2}{} \ln |\eta_t(x) - \eta_t(y)| \, \dd \nu(x) \, \dd \nu(y) + \e{O}\left( \frac{\ln N}{N} \right).$$
The remainder in uniform in $t \in \intff{0}{1}$.
\end{lemma}
\begin{proof}
Let us write $\ln|\eta_t(x) - \eta_t(y)| = \Lambda(x,y) + \ln|x - y|$ where $\Lambda(x,y) = \ln\left|\frac{\eta_t(x) - \eta_t(y)}{x-y} \right|$ for $x \neq y$ and $\Lambda(x,x) = \ln|\eta^\prime_t(x)|$.
Note that $\Lambda$ is a $\mc{C}^\infty$ function and, by Lemma \ref{curvelowerbound}, all of its derivatives may be bounded independently of $t \in \intff{0}{1}$.
Hence the only thing we need to show is that
$$\frac{1}{N^2} \sul{ \substack{i,j=1 \\ i \neq j} }{N} \ln |\xi_i - \xi_j| \, =\,  \Int{ \mathcal{J}^2 }{} \ln |x - y| \, \dd \nu(x) \, \dd \nu(y) + \e{O}\left( \frac{\ln N}{N} \right).$$
Let us write
\begin{align*}
\Int{\mathcal{J}^2}{}  \ln|x - y| \, \dd \nu(x)  \dd\nu(y)
&= \Bigg\{ \sul{ \substack{i,j=1 \\ |i-j| \geq 2} }{N} \, + \,  \underbrace{\sul{\substack{i,j=1 \\ |i-j| \leq 1}}{N} }_{=: (A)} \Bigg\}
\hspace{-4mm}  \Int{ \substack{ [\xi_{i-1},\xi_i]  \\ \quad \times [\xi_{j-1},\xi_j] } }{} \hspace{-5mm}  \ln|x - y|  \, \dd \nu(x) \, \dd \nu(y) \;.
\end{align*}
Contribution $(A)$ corresponds to the tridiagonal terms. By hypothesis, $\mathrm{supp} \, \nu \subset \intff{-M}{M}$ and $\dd \nu(x) \leq K \, \dd x$ for some $K > 0$.
This implies that $\xi_i - \xi_{i-1}\geq \frac{1}{KN}$ and that $\ln\frac{|x - y|}{2M} \leq 0$. Hence decomposing
$\ln|x - y| \, = \,  \ln\frac{|x - y|}{2M} +  \ln (2M)$ leads to the bounds
\begin{align*}
\frac{3N-2}{N^2}\ln(2M) \geq (A) \geq   \frac{3N-2}{N^2}\ln(2M)    \, - \, (A^{\prime}) \, ,
\end{align*}
with
\beqa
 (A^{\prime}) & =  &\sul{ i =1 }{N}  \Int{ \xi_{i-1} }{\xi_i } \hspace{-2mm}  \ln\frac{2M}{|x - y|} \,  \dd \nu(x) \dd \nu (y)
\, + \, 2 \sul{ i =1 }{N-1}  \Int{ \xi_{i-1} }{\xi_i } \hspace{-2mm}   \dd \nu(x)  \hspace{-2mm}  \Int{ \xi_{i} }{\xi_{i+1} }\hspace{-2mm}  \dd \nu (y) \,  \ln\frac{2M}{|x - y|} \nonumber \\
& \leq & 2 K \sul{ i =1 }{N}  \Int{ \xi_{i-1} }{\xi_i } \hspace{-2mm}   \dd \nu(x) \Int{0}{ \xi_i-\xi_{i-1} } \hspace{-1mm} \dd y  \ln\frac{2M}{y}
\, + \, 2K \sul{ i =1 }{N-1}  \Int{ \xi_{i-1} }{\xi_i } \hspace{-2mm} \dd \nu(x)  \Int{ \xi_{i} }{\xi_{i+1} } \hspace{-2mm} \dd y  \ln\frac{2M}{y-\xi_i} \nonumber \\
& \leq & \f{ 4 K }{ N } \big[ \ln(2M) + 1 + \ln K + \ln N \big] 2M \;.
\nonumber
\eeqa
Further,  the bounds $\frac{1}{KN}\leq \xi_{i}-\xi_{i-1} \leq 2M$ lead to the estimates  $\sum_{i,j=1, \,   |i - j| = 1}^N \ln|\xi_i - \xi_j| = \e{O}(N \ln N )$.

Finally, the elementary inequality $\left| \ln|x| - \ln|y| \right| \leq |x-y| \left( \frac{1}{|x|}+\frac{1}{|y|} \right)$ allows one to estimate the remaining sums:
\begin{align*}
&\sul{\substack{i,j=1 \\|i - j | \geq 2}}{N} \; \Int{\xi_{i-1}}{\xi_i} \hspace{-1mm} \dd\nu(x) \Int{\xi_{j-1}}{\xi_j} \hspace{-1mm} \dd\nu(y)  \big| \ln|\xi_i - \xi_j| - \ln|x-y| \big| \, \\
&\quad \leq \frac{K}{N}\sum_{\substack{i,j=1 \\|i - j | \geq 2}}^N (\xi_i - \xi_{i-1}+\xi_j-\xi_{j-1}) \left( \frac{1}{|i-j|-1}+\frac{1}{|i-j|}\right) = \e{O}\left( \frac{\ln N}{N} \right).
\end{align*}
\end{proof}
Next, let us establish a lower bound on the partition function.
\begin{lemma}[Lower bound on partition function]\label{partitionlower}
There exists a constant $C > 0$ independent of $t \in \intff{0}{1}$ such that
$$\ln \op{Z}_{N,\eta_t}[W_t] \geq - \frac{\beta N^2}{2} \op{I}_{\Om_t}\big[\mu_{\Om_{t}}\big] - CN \ln N \;. $$
Above, $\Om_t\, = \, \eta_t\big( \mc{J}\big)$ while the curve $\eta_t$ and the potential $W_t$ satisfy to the hypotheses stated in the beginning of the section.
\end{lemma}
\begin{proof}
As in Lemma \ref{logquantiles}, let $\xi_0, \dots, \xi_N$ be the $\frac{1}{N}$ quantiles of the equilibrium measure $ \wt{\mu}_{ \Om_t } \, = \, \eta_t^{\#}[ \mu_{ \Om_t }]$, the pullback to the real line
of $\mu_{ \Om_t }$. Taken $C_{\eta}$ as given in \eqref{ecriture bornitude norme eta t et ctes lipschitz}, one has
\begin{align*}
\op{Z}_{N,\eta_t}[V_t] 
&\geq \;  C_{\eta}^N \hspace{-5mm} \Int{ \| \mathbf{x} - \mathbf{\xi} \|_\infty \leq N^{-2} }{} \pl{ i < j}{ N} |\eta_t(x_i) - \eta_t(x_j)|^\beta \pl{k=1}{N}\ex{-N \beta \widetilde{\varphi}_t(x_k)}\, \dd\mathbf{x}
\end{align*}
where we recall that $\widetilde{\varphi}_t =\Re \big( W_t \circ \eta_t)$.
Then, by  \eqref{ecriture bornitude norme eta t et ctes lipschitz},
$$|\eta_t(x_i) - \eta_t(x_j)| \, \geq \, |\eta_t(\xi_i) - \eta_t(\xi_j)| \left( 1- \frac{2K C_\eta^{-2}}{N} \right)$$
where $\dd\wt{\mu}_{ \Om_t }(x) \leq K \, \dd x$.
Similarly, one has
$$ N \sul{j=1}{N} \widetilde{\varphi}_t(x_k) \, \leq \, N \sul{j=1}{N} \widetilde{\varphi}_t(\xi_k) + C \, ,$$
where we note that by \eqref{ecriture borne sur norme W infty du potentiel} this $C$ may be taken to be independent of $t \in \intff{0}{1}$.
Finally, the resulting integration region has volume $\left( \frac{2}{N^2} \right)^N$. Then, by Lemma \ref{logquantiles}, we obtain the result.
\end{proof}
Moving forward, let us follow Maïda and Maurel-Segala \cite{Maida:2014aa} and define a regularisation of the empirical measure.
\begin{definition}[Regularised empirical measure]\label{regularmeasure}
Let $L_N^{(\mathbf{x})}$ be the empirical measure, \textit{c.f.} \eqref{definition empirical measure},  associated to some vector $\mathbf{x} \in \mathbb{R}^N$.
Since $L_N^{(\mathbf{x})}$ is invariant under permutations of the coordinates of $\mathbf{x}$, we may just as well assume that
$x_1 \leq  \dots \leq  x_N$. Then, we define
\begin{align*}
\widetilde{x}_1 &:= x_1 \\
\widetilde{x}_k &:= \widetilde{x}_{k-1} + \max\left(  x_k - x_{k-1}, N^{-3} \right)  \quad \e{for}  \quad  k = 2, \dots, N \, .
\end{align*}
Next, we define the regularised empirical measure $\overline{L}_N^{(\mathbf{x})} \in \mathcal{P}(\mathbb{R})$ by
\begin{align*}
\dd\overline{L}_N^{(\mathbf{x})} (x) \overset{\mathrm{def}}{=} \frac{1}{N}\sum_{k=1}^N N^{6} \mathbbm{1}_{ \big[ \widetilde{x}_k - \frac{1}{2N^{6}} ;  \widetilde{x}_k + \frac{1}{2N^{6}} \big]}(x) \, \dd x.
\end{align*}
\end{definition}
Now, observe that one may decompose
$$
\Int{}{}  \dd\overline{L}_N^{(\mathbf{x})} (u) \dd\overline{L}_N^{(\mathbf{x})} (v) \ln |\eta_t(u) - \eta_t(v)|\; = \; W_1 + W_2
$$
where
\begin{align}\label{bound1}
W_1 \, = \, \Int{-\f{1}{2}}{\f{1}{2}} \sum_{\substack{i,j=1  \\ i \neq j}}^N \ln |\eta_t(\widetilde{x}_i + N^{-6} u ) - \eta_t(\widetilde{x}_j + N^{-6} v)| \, \dd u \, \dd v \, =\,
\sum_{\substack{i,j=1  \\ i \neq j}}^N  \ln |\eta_t(\widetilde{x}_i ) - \eta_t(\widetilde{x}_j )| + \e{O}\Big( \f{1}{N} \Big)
\end{align}
with a remainder that, owing to  \eqref{ecriture bornitude norme eta t et ctes lipschitz}, is uniform in $t \in \intff{0}{1}$.
Finally, one has
\begin{align}\label{bound2}
\begin{split}
W_2 = \Int{-\f{1}{2}}{\f{1}{2}}  \sum_{i=1}^N \ln |\eta_t(\widetilde{x}_i + N^{-6} u) - \eta_t(\widetilde{x}_i + N^{-6} v)|\, \dd u \, \dd v \geq - C N  \ln N
\end{split}
\end{align}
for some constant $C > 0$ that, owing to Lemma \ref{curvelowerbound}, is uniform in $t \in \intff{0}{1}$.

\begin{lemma}\label{smoothness}
Let $\mathcal{J}$ and $\widetilde{\mathcal{J}}$ be two compact intervals, where $\widetilde{\mathcal{J}}$ is larger on both sides than $\mathcal{J}$,
meaning that there exists an open set $U \subset \mathbb{R}$ such that $\mathcal{J} \subset U \subset \widetilde{\mathcal{J}}$.
Let $\eta_t$ admit an extension to $\widetilde{\mathcal{J}}$ satisfying the same properties and let $\Lambda : \widetilde{\mathcal{J}}^2 \longrightarrow \mathbb{R}$ be given by
\begin{align*}
\Lambda(x,y) &\overset{\mathrm{def}}{=} \begin{cases} \ln \left| \frac{ \eta_t(x)- \eta_t(y)}{x-y}\right| & x,y \in \widetilde{\mathcal{J}}, \, x \neq y \vspace{2mm} \\
\ln|\eta_t^\prime(x)| & x\in \widetilde{\mathcal{J}}, \, x = y.
\end{cases}
\end{align*}
If $\mathbf{x} \in \mathcal{J}^N$ then for $N$ sufficiently large $\widetilde{\mathbf{x}} \in \widetilde{\mathcal{J}}^N$ (recall Definition \ref{regularmeasure})
and furthermore there is a constant $C^\prime > 0$ and $N_0 \in \mathbb{N}$ such that
\begin{align*}
\left| \Lambda(x_i, x_j) - \Lambda(\widetilde{x}_i, \widetilde{x}_j) \right| \leq C^\prime N^{-2} & & \forall i,j \in \intn{1}{N}\, , \quad and \quad N \geq N_0 ,
\end{align*}
with $N_0$ independent of $t \in \intff{0}{1}$.

\end{lemma}
\begin{proof}
We compute the formula
\begin{align*}
\frac{\partial}{\partial x} \Lambda(x,y) = \Re \bigg(\frac{x-y}{\eta_t(x) - \eta_t(y)} \Int{0}{1} u \eta_t^{\prime \prime}(y+u(x-y)) \, \dd u \bigg)
\end{align*}
and a similar formula for $\frac{\partial}{\partial y} \Lambda(x,y)$ by interchanging $x$ and $y$. Thus
$$\| \grad \Lambda \| \leq \sqrt{2} C_{\eta}^{-1} \sup_{x \in \widetilde{\mathcal{J}}} |\eta_t^{\prime \prime}(x)| < +\infty \, ,$$
with a bound that is uniform in $t \in \intff{0}{1}$.

\end{proof}
Moving forward, we may make the upper bound
\begin{align}\label{bound3}
\begin{split}
\sul{ i < j }{  N }\ln|\eta_t(x_i) - \eta_t(x_j)| &= \sul{ i < j }{ N }\ln|x_i -x_j| \, + \,  \sul{ i < j }{ N }\Lambda(x_i, x_j) \\
&\leq \sul{ i < j }{ N } \ln|\widetilde{x}_i - \widetilde{x}_j|  \, + \,  \sul{ i < j }{ N }\Lambda(\widetilde{x}_i, \widetilde{x}_j)+ \frac{1}{2} C^\prime \\
&= \sul{ i < j }{ N }\ln|\eta_t(\widetilde{x}_i) - \eta_t(\widetilde{x}_j)| + \frac{1}{2} C^\prime
\end{split}
\end{align}
Putting (\ref{bound1}), (\ref{bound2}), (\ref{bound3}) and Lemma \ref{partitionlower} together,  we infer there exists a constant $C > 0$ independent of $t \in \intff{0}{1}$
such that, for all $N \geq 2$,
\begin{align}\label{regularisedbound}
\dd \mathbb{P}_{N,\eta_t}(\mathbf{x})   \, \leq \,
\ex{ - \frac{\beta}{2} N^2 \big( \mf{I}_{\mc{J}}[ \overline{L}_N^{(\mathbf{x})}] - \mf{I}_{\mc{J}}[ \wt{\mu}_{\Om_t}] \big) + C N \ln N} \cdot  \dd \mathbf{x}
\end{align}
where we recall that $ \wt{\mu}_{ \Om_t } \, = \, \eta_t^{\#}[ \mu_{ \Om_t }]$ and $\Om_t = \eta_t(\mc{J})$.
\begin{lemma} Let $ \wt{\mu}_{\Om_t}$ be the unique minimiser of $\mf{J}_{\mc{J}}$. Then
\begin{align*}
 \mf{I}_{\mc{J}}\big[ \, \overline{L}_N^{(\mathbf{x})}\big] - \mf{I}_{\mc{J}}[ \wt{\mu}_{\Om_t}]  \geq \Big\{ \mf{D}_{\eta_t}\Big(\overline{L}_N^{(\mathbf{x})}, \wt{\mu}_{\Om_t}\Big) \Big\}^2
\end{align*}
\end{lemma}
\begin{proof}
This follows from the Frostman-Euler-Lagrange conditions (Lemma \ref{frostmanlemma}); namely, there exists a constant $C_{\Om_t}$ such that
\bem
 \mf{I}_{\mc{J}} \big[ \,  \overline{L}_N^{(\mathbf{x})} \big]  \, =  \,  \mf{I}_{\mc{J}}[ \wt{\mu}_{\Om_t}]  \, + \,
 2 \Int{\mathbb{R}}{} \underbrace{\left[ \Int{\mathbb{R}}{} \ln \frac{1}{|\eta_t(x) - \eta_t(y)|} \,
\dd \wt{\mu}_{\Om_t}(y) + \widetilde{\varphi}_t(x) - C_{\Om_t}\right]}_{\geq 0} \, \dd \overline{L}_N^{(\mathbf{x})}(x)  \\
\, + \, \Big\{ \mf{D}_{\eta_t}\Big(\overline{L}_N^{(\mathbf{x})}, \wt{\mu}_{\Om_t}\Big) \Big\}^2 \, .
\end{multline}
\end{proof}
\begin{proposition}\label{logdistbound} There exists a constant $C > 0$  independent of $t \in \intff{0}{1}$ such that, for all $\delta > 0$,
\begin{align*}
\mathbb{P}_{N,\eta_t} \bigg[ \Big\{ \mf{D}_{\eta_t} \big( \,  \overline{L}_N^{(\mathbf{x})}, \wt{\mu}_{\Om_t} \big) \geq \delta \Big\} \bigg] \leq \ex{- \frac{\beta}{2} N^2 \delta^2 + C N \ln N} \;.
\end{align*}
\end{proposition}
\begin{proof}
This follows from integrating both sides of \eqref{regularisedbound} on the event $\mf{D}_{\eta_t} \big( \,  \overline{L}_N^{(\mathbf{x})}, \wt{\mu}_{\Om_t} \big) \geq \de$.
\end{proof}
\begin{corollary}\label{concentration}
Let $\wt{\mc{J}}$ strictly contain $\mc{J}$.
Then, there exist constants $C , C^\prime > 0$ independent of $t \in \intff{0}{1}$ and $N_0 \in \mathbb{N}$ such that for all $f \in \mc{W}^{\infty}_{1}(\mc{J})$,  and $N \geq N_0$,
one has $\wt{\mathbf{x}} \in \wt{\mc{J}}^N$. Moreover, it holds
\begin{align*}
\mathbb{P}_{N,\eta_t}\Bigg[ \bigg\{ \Big| \int_{ \wt{\mc{J}} }^{} f(s)
\, \dd \big(\overline{L}_N^{(\mathbf{x})} -  \wt{\mu}_{\Om_t} \big)(s) \Big| \geq \delta  \bigg\} \Bigg] \, \leq \,
\exp\bigg\{- \frac{ C^\prime  N^2 \delta^2 }{ \| f \|_{ \mc{W}^{\infty}_{1}( \wt{\mc{J}}) }^2} + C N \ln N \bigg\} \;.
\end{align*}
\end{corollary}
\begin{proof}
By Proposition \ref{weaktopology}
\begin{align*}
\mathbb{P}_{N,\eta_t}\Bigg[ \bigg\{ \Big| \int_{ \wt{\mc{J}} }^{} f(s)
\, \dd \big(\overline{L}_N^{(\mathbf{x})} -  \wt{\mu}_{\Om_t} \big)(s) \Big| \geq \delta  \bigg\} \Bigg] & \leq
\mathbb{P}_{N,\eta_t}\Bigg[ \bigg\{  \mf{D}_{\eta_t} \big( \,  \overline{L}_N^{(\mathbf{x})}, \wt{\mu}_{\Om_t} \big)  \geq
                            \frac{\delta}{  \mf{C}_{ \mc{J}}  \| f \|_{ \mc{W}^{\infty}_{1}( \wt{\mc{J}}) } }  \bigg\} \Bigg]
\end{align*}
The result then follows from Proposition \ref{logdistbound}.
\end{proof}
Thus $\Big| \int_{ \wt{\mc{J}} }^{} f(s)\, \dd \big(\overline{L}_N^{(\mathbf{x})} -  \wt{\mu}_{\Om_t} \big)(s) \Big|$ is typically at most of order $\e{O}\Big( \sqrt{\tfrac{\ln N}{N}}  \Big)$.
In fact we shall later see that if $f$ is $\mc{C}^\infty$, the optimal bound is $\e{O}(\tfrac{1}{N})$.  From Corollary \ref{concentration} we may bound the moments.
\begin{corollary}\label{moments} Let $\wt{\mc{J}}$ strictly contain $\mc{J}$.
Then there exists a constant $C_m > 0$ and $N_0 \in \mathbb{N}$, both independent of $t \in \intff{0}{1}$,  such that for all $f \in \mc{W}^{\infty}_1\big( \wt{\mc{J}} \big)$, $m \in \mathbb{N}$ and $N \geq N_0$
\begin{align*}
\mathbb{E}_{N,\eta_t} \Bigg[ \bigg| \int_{\wt{\mc{J}}}^{}f(s)\, \dd \big(\overline{L}_N^{(\mathbf{x})} -  \wt{\mu}_{\Om_t} \big)(s)\bigg|^m \Bigg] \; \leq \;
C_m  \| f \|_{ \mc{W}^{\infty}_1( \wt{\mc{J}} ) }^m \left( \frac{\ln N}{N} \right)^\frac{m}{2}  \;.
\end{align*}
\end{corollary}
\begin{corollary}\label{generatingfunction}
Let $\wt{\mc{J}}$ strictly contain $\mc{J}$.
Let $C, C^\prime > 0$ be the constants appearing in Lemma \ref{concentration} and $N_0 \in \mathbb{N}$.
Then,  for all  $N \geq N_0$ and $f \in \mc{W}^{\infty}_1\big( \wt{\mc{J}} \big)$
such that $  \| f \|_{ \mc{W}^{\infty}_1( \wt{\mc{J}} ) }  \leq \sqrt{C C^\prime N \ln N}$ one has
\begin{align*}
\mathbb{E}_{N,\eta_t} \Bigg[  \ex{ N \big| \int_{ \wt{\mc{J}} }^{ } f(s)\, \dd (\overline{L}_N^{(\mathbf{x})} -  \wt{\mu}_{\Om_t} )(s) \big|} \Bigg]
\, \leq \,  2 \ex{ \sqrt{\frac{C N \ln N }{C^\prime}} \cdot  \| f \|_{ \mc{W}^{\infty}_1( \wt{\mc{J}} ) } } \;.
\end{align*}
\end{corollary}
\begin{proof}
\begin{align*}
\mathbb{E}_{N,\eta_t} \Bigg[  \ex{ N \big| \int_{ \wt{\mc{J}} }^{ } f(s)\, \dd (\overline{L}_N^{(\mathbf{x})} -  \wt{\mu}_{\Om_t} )(s) \big|} \Bigg] \, \leq \,
1 +  \Int{0}{\infty} \ex{x} \,  \min \bigg\{ 1,   \exp\Big\{- \frac{ C^\prime  x^2 }{ \| f \|_{ \mc{W}^{\infty}_1( \wt{\mc{J}} ) }^2} + C N \ln N \Big\} \bigg\} \, \dd x \;.
\end{align*}
Then, the upper bound
$$- \frac{ C^\prime  x^2 }{ \| f \|_{ \mc{W}^{\infty}_1( \wt{\mc{J}} ) }^2} + C N \ln N  \, \leq \,
-\frac{2 \sqrt{C C^\prime N \ln N}}{ \| f \|_{ \mc{W}^{\infty}_1( \wt{\mc{J}} ) }  } \left(x -\| f \|_{ \mc{W}^{\infty}_1( \wt{\mc{J}} ) } \sqrt{C (C^\prime)^{-1} N \ln N}\right)\,   $$
entails the claim.
\end{proof}

\section{Large deviations for the maximal particle}\label{maxeigsection}

In this section, we mainly focus on the real model adapted to the contour $\Ga_{\e{eq}}$ provided by Proposition  \ref{contourexistence}
and realised by the curve $\ga: \R \rightarrow \Ga_{\e{eq}}$. We recall that this contour enjoys the properties \ref{hypothesis2}-\ref{hypothesis3}.
We shall prove that with probability exponentially close to $1$,
the particles are confined to an arbitrary small neighbourhood of the support of the equilibrium measure.
One corollary of this is that the partition function $\op{Z}_{N,\ga^{\eps} }[V]$ associated with the restriction $\ga^{\eps} \, = \, \ga_{\mid \intff{-\eps}{1+\eps} }$ of the curve $\ga$ to the interval
$\intff{-\eps}{1+\eps}$ is a good approximation of the partition function on the unbounded contour $\op{Z}_{N,\ga }[V]$. In particular, one may focus on the
former so as to obtain the asymptotic expansion of the latter.
 The proofs presented in this section follow closely the proofs given in \cite{Borot:2013aa} or \cite{Gu19}.

Let us begin by introducing the pullback to $\mathbb{R}$ under $\gamma$ of the relevant effective potential:
\beq
\widetilde{\varphi}_{\Ga_{\e{eq}};\mathrm{eff}}(x) = \widetilde{\varphi}(x) \, -\,  \Int{0}{1} \ln|\gamma(x) - \gamma(y)| \, \dd \nu_{\e{sc}}(y) \, - \,  C_{ \Ga_{\e{eq}}}\\
\quad \e{with} \quad\widetilde{\varphi} \, = \,  \Re \big( V \circ \gamma \big) \;.
\nonumber
\enq
We recall that, by construction, the pullback of $\mu_{\Ga_{\e{eq}}}$ under $\ga$ is given by  $\nu_{\e{sc}}$ see \eqref{definition mesure semi cercle sur 0 1} and Remark \ref{semicircleremark}.
 $\nu_{\e{sc}}$ minimizes the functional $\mf{I}_{\R}$ introduced in \eqref{Jfunctional} with the choice $\eta_t=\ga$ and $\wt{\vp}_t \, = \,  \Re \big( V \circ \gamma \big)$
 and  $ C_{ \Ga_{\e{eq}}}$ is the infimum of
$x\mapsto \widetilde{\varphi}(x) \, -\,  \int_{0}^{1} \ln|\gamma(x) - \gamma(y)| \, \dd \nu_{\e{sc}}(y) $ which is achieved throughout $\intff{0}{1}$,
the support of $\nu_{\e{sc}}$.
We shall now build up to proof of the Large Deviation Principle (LDP) for the maximal particle at scale $N$ with rate function $\widetilde{\varphi}_{\Ga_{\e{eq}};\mathrm{eff}}$
for the real model $\mathbb{P}_{N,\ga}$, with $\ga(\R)= \Ga_{\e{eq}}$ as mentioned earlier.
\begin{proposition}(LDP for the maximal particle)\label{MaxEigLDP}
For any closed set $F \subset \mathbb{R}$
$$ \limsup_{N\to +\infty} \frac{1}{N}  \ln \mathbb{P}_{N,\ga} \Big[ \big\{ \exists i \in \{1 , \dots, N  \} \, \text{ s.t. } x_i \in F  \big\} \Big] \leq
- \beta \inf_{x \in F}  \widetilde{\varphi}_{\Ga_{\e{eq}}; \mathrm{eff}}(x) $$
and for any open set $U \subset \mathbb{R}$
$$ \liminf_{N \to +\infty} \frac{1}{N}  \ln \mathbb{P}_{N,\ga}\Big[ \big\{ \exists i \in \{1 , \dots, N  \} \, \text{ s.t. } x_i \in U   \big\} \Big]  \geq
- \beta  \inf_{x \in U}  \widetilde{\varphi}_{\Ga_{\e{eq}}; \mathrm{eff}}(x) \; .  $$
\end{proposition}

\begin{corollary}
\label{corollaire approx partition fct}
Let $\ga^{\eps}$ denote the restriction of $\ga$ to $\intff{-\eps}{1+\eps}$, \textit{i.e.} $\ga^{\eps}=\ga_{\mid \intff{-\eps}{1+\eps}}$. Then, there exists $C>0$ such that
\beq
\op{Z}_{N,\ga^{\eps} }[V] \, = \, \op{Z}_{N,\ga }[V] \Big( 1+ \e{O}\big(\ex{-CN } \big)  \Big) \;.
\enq

\end{corollary}

\begin{proof}

One has the identity
\beq
\op{Z}_{N,\ga }[V]  \, = \, \op{Z}_{N,\ga^{\eps} }[V]  \, + \, \op{Z}_{N,\ga }[V] \cdot  \mathbb{P}_{N,\ga} \Big[ \big\{ \exists i \in \{1 , \dots, N  \} \, \text{ s.t. } x_i \in \intff{-\eps}{1+\eps}^{\e{c}}  \big\} \Big]
\enq
so that one concludes owing to Proposition \ref{MaxEigLDP}.

\end{proof}

\begin{lemma}
$\widetilde{\varphi}_{\Ga_{\e{eq}};\mathrm{eff}} : \mathbb{R}\longrightarrow \intfo{0}{+\infty}$ is a good rate function, \textit{i.e.} is lower semi-continuous (lsc)
and has compact level sets.
\end{lemma}
\begin{proof}
$\widetilde{\varphi}_{ \Ga_{\e{eq}}; \mathrm{eff} }$ is continuous on
$\e{supp}[\nu_{\e{sc}}]^c$, vanishes on the compact set $\e{supp}[ \nu_{\e{sc}} ]$,
and so, in particular, is lsc. For any $C \in \mathbb{R}$, $\widetilde{\varphi}_{\Ga_{\e{eq}};\mathrm{eff}}^{-1}((-\infty,C])$ is closed and so we need only to show boundedness.
$\widetilde{\varphi}_{\Ga_{\e{eq}};\mathrm{eff}}(x) \leq C$ implies by
\begin{align}\label{splitting}
|x-y|^2 \leq (1+|x|^2)(1+|y|^2) & & \forall x,y\in \mathbb{C}
\end{align}
that
\begin{align*}
\widetilde{\varphi}(x) - \frac{1}{2} \ln(1+|\gamma(x)|^2) - C_{ \Ga_{\e{eq}}} \leq C - C_{ \Ga_{\e{eq}}}  + \frac{1}{2} \int_\mathbb{R}  \ln(1+|\gamma(y)|^2) \, \dd \nu_{\mathrm{sc}}(y) \;.
\end{align*}
However $\liminf_{x \to \pm \infty} \left( \widetilde{\varphi}(x) -  \frac{1}{2}\ln(1+|\gamma(x)|^2) \right) = +\infty$, and so $x$ must be bounded.
\end{proof}
The method of proving Proposition \ref{MaxEigLDP} involves first proving that the particles are confined with exponentially high probability to \textit{some} compact set.
Once established, we may then use concentration results on compact sets to show that that particles are close to the support of the equilibrium measure.

We begin with the elementary inequality
$$\mathbb{P}_{N,\ga}\big[ \big\{ x_N \in X \big\} \big]  \, \leq  \,  \mathbb{P}_{N,\ga} \big[ \big\{ \exists i \in \{1 , \dots, N  \} \, \text{ s.t. } x_i \in X  \big\} \big]
                    \leq N \mathbb{P}_{N,\ga}\big[ \big\{ x_N \in X \big\} \big] \; . $$
Thus, to prove Proposition \ref{MaxEigLDP}, we need only prove an LDP for $\mathbb{P}_{N,\ga}\big[ \big\{ x_N \in X \big\} \big]$.
For any Borel set $X \subset \mathbb{R}$, let
$$ \Upsilon_{N-1}(X) \overset{\mathrm{def}}{=} \Int{X}{}
\mathbb{E}_{N-1,\ga} \left[ \ex{ (N-1)\beta \int_\mathbb{R} \big[ \ln|\gamma(x) - \gamma(y)| - \widetilde{\varphi}(y) \big] \, \dd L_{N-1}^{(\mathbf{x})}(y)}  \right]
\ex{ - N \beta\widetilde{\varphi}(x) } |\gamma^\prime(x)| \, \dd x \, .$$
Then
$$\mathbb{P}_{N,\ga}\big[ \big\{ x_N \in X \big\} \big] = \frac{\Upsilon_{N-1}(X)}{\Upsilon_{N-1}(\mathbb{R})} \; . $$
Thus we need only prove an LDP for $\Upsilon_{N-1}$. Let us first prove that all the particles are confined to some sufficiently large compact set.
\begin{corollary}\label{Mlimsup}
$$\limsup_{M \to +\infty} \limsup_{N \to +\infty}\frac{1}{N^2} \ln \mathbb{P}_{N,\ga} \big[ A_M \big]  = -\infty$$
where
\beq
A_M \;= \; \Big\{ \mathbf{x} \in \mathbb{R}^N \, :
                        \,  \Int{\mathbb{R}}{} \Big( \widetilde{\varphi}(x) - \frac{1}{2} \ln (1+|\gamma(x)|^2) \Big) \, \dd L_N^{(\mathbf{x})}(x) \geq M \Big\} \;.
\label{ecritude evenement AM}
\enq

\end{corollary}
\begin{proof}
Starting from
\begin{align*}
\mathbb{P}_{N,\ga} \big[ A_M \big]= \frac{1}{ \op{Z}_{N, \ga}[V] } \Int{ \mathbb{R}^N }{} \mathbbm{1}_{A_M}(\mathbf{x})
\pl{\substack{ i, j=1 \\ i \neq j}}{N} |\gamma(x_i) - \gamma(x_j)|^{\frac{\beta}{2}}  \cdot \pl{k=1}{N}\Big\{ |\gamma^\prime(x_k)| \ex{-\beta N  \widetilde{\varphi}(x_k)} \Big\} \ \, \dd \mathbf{x} \;,
\end{align*}
the upper bound \eqref{splitting} yields $ \prod_{\substack{ i, j=1 \\ i \neq j}}^N |\gamma(x_i) - \gamma(x_j)| \leq  \prod_{i=1}^N (1+|\gamma(x_i)|^2)^{N-1}$
which, in turn, gives
$$\mathbb{P}_{N,\ga} \big[ A_M \big] \leq    \exp\Big\{-\beta N^2 M  + N^2 \be \op{I}_{\Ga_{\e{eq}}}[\mu_{\Ga_{\e{eq}}} ] + C N \ln N \Big\}\cdot
\bigg( \int_{\mathbb{R}} \frac{1}{ C_{\gamma} \cdot (1+|\gamma(x)|^2)^{\frac{\beta}{2}}} \, \dd x \bigg)^N \,. $$
Here we have bounded the partition function from below with the help of Lemma \ref{partitionlower}.
By \ref{hypothesis3}, one has that $\Int{\mathbb{R}}{} (1+|\gamma(x)|^2)^{-\frac{\beta}{2}} \, \dd x < +\infty$ as soon as $\beta \geq 2$, which entails the claim.
\end{proof}
\begin{lemma}\label{expconfinement}
For any $C > 0$, there is a $K > 0$ and $N_0 \in \mathbb{N}$ sufficiently large such that
\begin{align*}
\mathbb{P}_{N,\ga} \big[ \big\{  \exists i \in \{1 , \dots, N  \} \, \text{ s.t. } x_i \in   \intff{-K}{K}^{\e{c}} \big\} \big] \leq \ex{ - C N }  & & \forall N \geq N_0 \;.
\end{align*}
\end{lemma}
\begin{proof}
Let us begin by first establishing a rough lower bound on $\Upsilon_{N-1}(\mathbb{R})$.  By positivity and then by Jensen's inequality
\begin{align*}
\Upsilon_{N-1}(\mathbb{R}) &\geq \Upsilon_{N-1}([-1,1]) \\
&\geq 2 C_\gamma \ex{ (N-1) \beta \ln C_\gamma - \frac{N\beta}{2} \int_{-1}^1 \widetilde{\varphi}(x) \, \dd x }
\mathbb{E}_{N-1, \ga} \Big[ \ex{ (N-1) \beta \int_\mathbb{R} \left( \frac{1}{2}\int_{-1}^1 \ln|x-y| \,\dd x - \widetilde{\varphi}(y)\right) \, \dd L_{N-1}^{(\mathbf{x})}(y) } \Big] \;.
\end{align*}
Then we observe that $\frac{1}{2}\int_{-1}^1 \ln|x-y| \, \dd x \geq \frac{1}{2}\int_{-1}^1 \ln|x| \, \dd x = -1$. To bound the remaining expectation,
introduce  $\op{U}_L \, = \, \Int{\mathbb{R}}{} \widetilde{\varphi}(y) \,\dd L_{N-1}^{(\mathbf{x})}(y)  $ and observe the trivial bound
$$\mathbb{E}_{N-1, \ga} \Big[ \ex{- (N-1) \beta \op{U}_L } \Big] \geq \ex{-M \be (N-1) }  \mathbb{E}_{N-1, \ga} \big[ \mathbbm{1}_{\op{U}_L \leq M} \big] \geq
\ex{-M \be N } \Big\{1 \, - \,   \mathbb{E}_{N-1, \ga} \big[ \mathbbm{1}_{\op{U}_L > M} \big]  \Big\} \,. $$
The last expectation may be estimated by means of Corollary \ref{Mlimsup} which ensures that for any $T$, there exists $M_0, N_0 $ such that
$$
\mathbb{P}_{N,\ga} \big[ A_M \big] \, < \, \ex{-N^2 T} \qquad \e{with} \qquad M \geq M_0 \quad \e{and} \quad  N \geq N_0 \; ,
$$
where $A_M$ as given by \eqref{ecritude evenement AM}.

Furthermore, the growth conditions at $\pm \infty$ ensure that $\wt{\vp}(x) \, - \, \tfrac{1}{2} \ln\big( 1+|\ga(x)|^2 \big) \geq  \tfrac{1}{2} \wt{\vp}(x) + C$. Hence,
\beq
\Big\{ \op{U}_L \geq M  \Big\}  \subset A_{ \f{M}{2} + C} \;.
\enq
So that, for $M$ large enough,
$$
\mathbb{E}_{N-1, \ga} \big[ \mathbbm{1}_{\op{U}_L > M} \big]  \leq \ex{-N^2 T} \;.
$$
All together, the above leads to  $\Upsilon_{N-1}(\mathbb{R}) \geq \ex{-C N}$ for some $C>0$.

Moving forward, let us now provide an upper bound for $\Upsilon_{N-1}(\intff{-K}{K}^{\e{c}} )$.
By using \eqref{splitting} and the bound $\frac{1}{2}\ln(1+|\gamma(x)|^2) - \widetilde{\varphi}(x) \leq C$ for all $x \in \mathbb{R}$ and some $C>0$, we find
\begin{align*}
\Upsilon_{N-1}( \intff{-K}{K}^{\e{c}} ) \leq
\exp\bigg\{ (N-1)\beta \Big\{ C + \hspace{-3mm}\underset{x \in  \intff{-K}{K}^{\e{c}} }{\sup}
\Big( \tfrac{1}{2}\ln(1+|\gamma(x)|^2) - \widetilde{\varphi}(x) \Big) \Big\} \bigg\}
 \int_\mathbb{R}\ex{-\beta \widetilde{\varphi}(x)} \, \f{ \dd x   }{  C_\gamma }
\end{align*}
in which the last integral is finite due to the growth conditions of $\widetilde{\varphi}$. Finally,
observing that
$$\underset{x \in  \intff{-K}{K}^{\e{c}} }{\sup} \Big(\frac{1}{2}\ln(1+|\gamma(x)|^2) - \widetilde{\varphi}(x)\Big) \to -\infty \qquad \e{as} \qquad  K \to +\infty$$
completes the proof.
\end{proof}
We remark that in the course of proving Lemma \ref{expconfinement} we have proven exponential tightness of $\Upsilon_{N-1}$.
\begin{lemma}[Exponential tightness]\label{confinement}
$$\limsup_{K \to +\infty}\limsup_{N \to +\infty} \frac{1}{N} \ln \Upsilon_{N-1}( \intff{-K}{K}^{\e{c}} ) = - \infty $$
\end{lemma}
To prove Proposition \ref{MaxEigLDP} we will prove a weak LDP, which by the aforementioned exponential tightness implies a full LDP.
\begin{definition}[Regularised logarithm]
Let $\theta : \mathbb{R} \longrightarrow \intff{0}{1}$ be a $\mc{C}^\infty$ function such that $\theta$ is even,
$\theta(0) = 0$, $\theta$ is nondecreasing for $x \geq 0$,  and $\theta(x) = 1$ for $|x| \geq 1$.
Then, define the $\delta$-regularised logarithm as
$$\ln_\delta |x| \overset{\mathrm{def}}{=} - \Int{x}{1} \theta\left(\frac{y}{\delta}\right) \frac{1}{y} \, \dd y$$
for $\delta > 0$ a constant. Note that $\ln_\delta$ is $\mc{C}^\infty$, even, and $\ln_\delta |x| = \ln|x|$ for $|x| \geq \delta$.
Furthermore $\ln_\delta$ is decreasing in $\delta > 0$ and $\ln_\delta \longrightarrow \ln$ pointwise as $\delta \downarrow 0$ by monotone convergence.
\end{definition}
\begin{lemma}[Upper bound on balls]\label{ballupper}
$$\limsup_{\vsg\downarrow 0} \limsup_{N \to +\infty }
\frac{1}{N} \ln \Upsilon_{N-1}\big( \intoo{\xi-\vsg}{\xi+\vsg} \big) \, \leq  \,
\beta \Int{0}{1} \big[  \ln|\gamma(\xi) - \gamma(y)| - \widetilde{\varphi}(\xi) - \widetilde{\varphi}(y)  \big] \, \dd \nu_{\e{sc}}(y) \, .$$
\end{lemma}
\begin{proof}
Let $B_K$ be the event that all particles lie in the interval $\intff{-K}{K}$. Then
\begin{align*}
&\Upsilon_{N-1}\big( \intoo{\xi-\vsg}{\xi+\vsg} \big) = \Circled{1} + \Circled{2} \\
&\Circled{1} =
\Int{\xi-\vsg}{\xi+\vsg} \mathbb{E}_{N-1,\ga} \bigg[ \mathbbm{1}_{B_K}
\ex{ (N-1)\beta \int_{\mathbb{R}}^{} \big(  \ln|\gamma(x) - \gamma(y)| - \widetilde{\varphi}(y) \big) \, \dd L_{N-1}^{(\mathbf{x})}(y)}  \bigg] \ex{ - N \beta\widetilde{\varphi}(x)} |\gamma^\prime(x)| \, \dd x \, , \\
&\Circled{2} =
\Int{\xi-\vsg}{\xi+\vsg} \mathbb{E}_{N-1,\ga} \bigg[ \mathbbm{1}_{B_K^{\e{c}}}
\ex{ (N-1)\beta \int_{\mathbb{R}}^{} \big(  \ln|\gamma(x) - \gamma(y)| - \widetilde{\varphi}(y) \big) \, \dd L_{N-1}^{(\mathbf{x})}(y)}  \bigg] \ex{ - N \beta\widetilde{\varphi}(x)} |\gamma^\prime(x)| \, \dd x \; .
\end{align*}
Then applying (\ref{splitting}), we may bound
$$\left| \, \Circled{2} \, \right| \leq C_\gamma^{-1} \ex{2(N-1)\beta C^\prime}
\mathbb{P}_{N-1,\ga}\Big[ \big\{  \exists i \in \{1 , \dots, N  \} \, \text{ s.t. } x_i \in  \intff{-K}{K}^{\e{c}} \big\} \Big]
\cdot  \int_\mathbb{R} \ex{-\beta \widetilde{\varphi}(x)} \, \dd x $$
where $C^\prime = \sup_{x \in \mathbb{R}}(\frac{1}{2}\ln(1+|\gamma(x)|^2) - \widetilde{\varphi}(x))$. Then, by Lemma \ref{expconfinement}, for any $C>0$
there is a $K > 0$ and $N_0 \in \mathbb{N}$ such that for all $N \geq N_0$
\begin{align*}
& \mathbb{P}_{N-1,\ga}\Big[ \big\{  \exists i \in \{1 , \dots, N  \} \, \text{ s.t. } x_i \in  \intff{-K}{K}^{\e{c}} \big\} \Big]\leq \ex{-CN} \;.
\end{align*}
We choose $C>0$ as large as needs be to make the upper bound on $\Circled{2}$ exponentially subleading with respect to the behaviour to be determined of $\Circled{1}$.
 Of course, without loss of generality, we may assume $\intoo{\xi-\vsg}{\xi+\vsg} \subset \intff{-K}{K}$.

Moving forward, for the first term, let us make the bound
$\ln|\gamma(x) - \gamma(y)| \leq \ln_{\delta}|x-y| + \Lambda(x,y)$ where $\delta > 0$ and $\Lambda(x,y) = \ln \left| \frac{\gamma(x) - \gamma(y)}{x-y}\right|$.  Let
\begin{align*}
F_{x,\delta}(y) := \ln_{\delta}|x-y| + \Lambda(x,y) - \widetilde{\varphi}(y) \;.
\end{align*}
Observe that $F_{x,\delta}(y)\ge F_{x,0}(y)$ for every $x$ and $y$. Moreover,
by compactness, there exists a constant $C_\delta > 0$ independent of $x \in [-K,K]$ such that $\| F_{x,\delta} \|_{\mc{W}^{\infty}_1 (\intff{-K}{K})} \leq C_\delta$.
Furthermore, one has $\op{Z}_{N,\ga_{\mid \intff{-K}{K}} }[V] \leq \op{Z}_{N,\ga }[V] $, which leads to
\begin{align*}
&\Circled{1} \leq C_\gamma^{-1} \cdot  \ex{ \frac{2\beta C_\delta}{N-1}  } \\
& \Int{ \xi - \vsg }{ \xi + \vsg } \ex{ (N-1)\beta \Int{0}{1} F_{x,\delta}(s) \, \dd\nu_{\e{sc}}(s)  }  \cdot  \mathbb{E}_ {N-1,\ga_{\mid \intff{-K}{K}} }
\bigg[ \ex{  (N-1)\beta \Int{\mathbb{R}}{} F_{x,\delta}(s) \, \dd( \overline{L}_{N-1}^{(\mathbf{x})} -\nu_{\e{sc}})(s) }  \bigg]
\cdot \ex{-N\beta\widetilde{\varphi}(x) } \, \dd x
\end{align*}
where the $\frac{2\beta C_\delta}{N-1}$ term comes from replacing $L_{N-1}^{(\mathbf{x})}$ with $\overline{L}_{N-1}^{(\mathbf{x})}$.
$(x,y)\mapsto F_{x,\delta}(y)$ is smooth on  $\intff{-K}{K}^2$ so that, by compactness, there exists $C^\prime_\delta > 0$
such that
$$ \left|\frac{\partial}{\partial x}F_{x,\delta}(y)\right| \leq C_\delta^\prime  \qquad \e{for} \; \e{all}  \qquad (x,y) \in \intff{-K}{K}^2 \, . $$
This ensures that
$F_{x,\delta} \leq F_{\xi,\delta} + \vsg C_\delta^\prime $.  A similar bound holds for $\widetilde{\varphi}$.
Finally, by Corollary \ref{generatingfunction}, there is a constant $C_\delta^{\prime\prime} > 0$ such that
$$\mathbb{E}_ {N-1,\ga_{\mid \intff{-K}{K}} }
\bigg[ \ex{  (N-1)\beta \Int{\mathbb{R}}{} F_{x,\delta}(s) \, \dd( \overline{L}_{N-1}^{(\mathbf{x})} -\nu_{\e{sc}})(s) }  \bigg]
\leq \,  \ex{ C^{\prime\prime}_\delta \sqrt{N \ln N} } \, .$$
Putting all this together we find that for all $\delta > 0$
\begin{align*}
\limsup_{\vsg \downarrow 0} \limsup_{N \to +\infty } \frac{1}{N} \ln \Upsilon_{N-1}( \intoo{\xi-\vsg}{\xi+\vsg}) \leq \beta \Int{0}{1} F_{\xi,\delta}(s) \,
\dd \nu_{\e{sc}}(s) - \beta  \widetilde{\varphi}(\xi) \,.
\end{align*}
Taking $\delta \downarrow 0$ we obtain the result by monotone convergence.
\end{proof}

\begin{lemma}[Lower bound on balls]\label{balllower}
$$\liminf_{\vsg \downarrow 0} \liminf_{N \to +\infty } \frac{1}{N} \ln \Upsilon_{N-1}( \intoo{\xi-\vsg}{\xi+\vsg})
\geq \beta \Int{\mathbb{R}}{} \big\{ \ln|\gamma(\xi) - \gamma(y)| - \widetilde{\varphi}(\xi) - \widetilde{\varphi}(y)  \Big\} \, \dd \nu_{\e{sc}}(y) \, . $$
\end{lemma}
\begin{proof} 
Let $B_K$ correspond to the event that all particles lie in the interval $\intff{-K}{K}$. Then
\begin{align*}
&\Upsilon_{N-1}( \intoo{\xi-\vsg}{\xi+\vsg} ) \\
&\geq C_\gamma \mathbb{P}_{N-1, \ga}[B_K] \Int{\xi-\vsg}{\xi+\vsg}
\mathbb{E}_{N-1, \ga_{\mid \intff{-K}{K} } }
\left[ \ex{ (N-1)\beta \int_\mathbb{R} \{  \ln|\gamma(x) - \gamma(y)| - \widetilde{\varphi}(y) \} \, \dd L_{N-1}^{(\mathbf{x})}(y) }  \right] \ex{-N\beta\widetilde{\varphi}(x)} \, \dd x \,.
\end{align*}
We first bound $\widetilde{\varphi}(x) \leq \widetilde{\varphi}(\xi) + \vsg \|\widetilde{\varphi}\|_{\mc{W}^{\infty}_1(\intff{-K}{K})}$.
Let $\widetilde{\omega}$ be a $\mc{C}^\infty(\mathbb{R})$ function supported on $\intff{-1}{1}$,  strictly positive on $\intoo{-1}{1}$, and such that $0 \leq \widetilde{\omega} \leq 1$.
Then $C = \int_{-1}^1 \widetilde{\omega}(x) \, \dd x > 0$ and define $\omega_{\vsg}(x) = \vsg^{-1} C^{-1} \widetilde{\omega}\left(\frac{x}{\vsg}\right)$.  Then let us apply Jensen's inequality to find
\bem
\Upsilon_{N-1}(\intoo{\xi-\vsg}{\xi+\vsg}) \, \geq \, \vsg C
C_\gamma \mathbb{P}_{N-1,\ga}[B_K] \ex{-N\beta \widetilde{\varphi}(\xi) - N \beta \vsg \|\widetilde{\varphi}\|_{\mc{W}^{\infty}_{1}( \intff{-K}{K})}} \\
\times \mathbb{E}_{N-1, \ga_{\mid \intff{-K}{K} } } \Big[ \ex{ (N-1)\beta \int_{\mathbb{R}} G_{\xi, \vsg}(s) \, \dd L_{N-1}^{(\mathbf{x})}(s) }   \Big] \;,
\end{multline}
where
$$
G_{\xi, \vsg}(y) \, :=  \, \Int{\mathbb{R}}{}\omega_{\vsg}( x-\xi)  \ln|\gamma(x) - \gamma(y)|  \, \dd x- \widetilde{\varphi}(y) \,.
$$
$G_{\xi, \vsg}$ is $\mc{C}^\infty$; indeed, one has  $\ln|\gamma(x) - \gamma(y)|  = \ln|x-y| + \Lambda(x,y)$ where $\Lambda(x,y) = \ln \left| \frac{\gamma(x) - \gamma(y)}{x-y} \right|$
is $\mc{C}^\infty$ by Lemma \ref{smoothness} whereas the $\ln|x-y|$ piece yields a convolution with $\omega_{\vsg}$ and is hence $\mc{C}^\infty$. Thus
the smoothness of $G_{\xi, \vsg}$ follows from  dominated convergence given that $\omega_{\vsg}$ has compact support.  Applying Jensen's inequality again we find
\bem
\mathbb{E}_{N-1, \ga_{ \mid \intff{-K}{K}} } \Big[ \ex{ (N-1)\beta \int_{\mathbb{R}} G_{\xi, \vsg}(s) \, \dd L_{N-1}^{(\mathbf{x})}(s) }   \Big]  \, \geq \,
\exp\bigg\{  (N-1)\beta \int_{\mathbb{R}} G_{\xi, \vsg}(s) \, \dd \nu_{\e{sc}}(s)    \\
-  (N-1)\beta \, \mathbb{E}_{N-1, \ga_{ \mid \intff{-K}{K}} } \Big[  \big| \int_{\mathbb{R}}{} G_{\xi, \vsg}(s) \, \dd(L_{N-1}^{(\mathbf{x})} - \nu_{\e{sc}})(s) \big| \Big]   \bigg\} \, .
\end{multline}
Then, applying Corollary \ref{moments}, we find
$$
\liminf_{N \to +\infty } \frac{1}{N}\ln \Upsilon_{N-1}(\intoo{\xi-\vsg}{ \xi+\vsg})  \, \geq \, - \beta \widetilde{\varphi}(\xi) - \beta \vsg \|\widetilde{\varphi}\|_{ \mc{W}^{\infty}_{1}(\intff{-K}{K}) }
+ \beta \int_\mathbb{R} G_{\xi,\vsg}(s) \, \dd \nu_{\e{sc}}(s) \,.
$$
By Fubini's theorem
\begin{align*}
\int_\mathbb{R} G_{\xi,\vsg}(s) \, \dd \nu_{\e{sc}}(s)\,  = \,  \Int{\mathbb{R}}{} \omega_{\vsg}(x-\xi) \left[ \int_\mathbb{R}  \left( \ln|\gamma(x) - \gamma(y)| - \widetilde{\varphi}(y) \right) \,
\dd\nu_{\e{sc}}(y)  \right] \, \dd x \,.
\end{align*}
The function  $x \mapsto \int_\mathbb{R}  \left[ \ln|\gamma(x) - \gamma(y)| - \widetilde{\varphi}(y) \right] \, \dd \nu_{\e{sc}}(y) $ is continuous and $\omega_{\vsg}$
is an "approximation of the identity", hence
$$\liminf_{\vsg \downarrow 0 } \liminf_{N \to +\infty } \frac{1}{N}\ln \Upsilon_{N-1}( \intoo{\xi-\vsg}{ \xi + \vsg} ) \, \geq \,
- \beta \widetilde{\varphi}(\xi) + \beta \int_\mathbb{R}  \left( \ln|\gamma(x) - \gamma(y)| - \widetilde{\varphi}(y)  \right) \, \dd \nu_{\e{sc}}(y) \;. $$
\end{proof}
Lemmata \ref{ballupper} and \ref{balllower} imply a weak LDP and, when combined with Lemma \ref{confinement}, imply Proposition \ref{MaxEigLDP}.
In the remainder of this section, we will obtain estimates on the $1$-point density
outside of $\intff{0}{1}$, the support of $\nu_{\e{sc}}$, in the case of the real model
associated with the family of curves
\beq
\ga_t^{\eps} =  {\ga_{t}}_{ \mid \intff{-\eps }{ 1 + \eps } } \, ,
\label{definition courbe gamma t eps}
\enq
 where $\ga_t$ has been introduced in \eqref{contourdef}, and $t \in \intff{0}{1}$.

\begin{definition}[1-point density]\label{1pointdensity} Let $\gamma_t : [-\epsilon,1+\epsilon] \longrightarrow \mathbb{C}$ be defined by \eqref{contourdef} for $t \in \intff{0}{1}$,
and let $V_t$ be the associated potential introduced in \eqref{interpolatingpotential}-\eqref{interpolatingpotential2}. The 1-point density is defined as
\bem
\varrho^t_{N}(x)  \, \overset{\mathrm{def}}{=} \, \frac{ \ex{-N \beta \widetilde{\varphi}_t(x) } |\ga_t^{\prime}(x)| }{ \op{Z}_{N,\ga_t^{\eps} }[V_t] }
\Int{-\eps }{1+\epsilon} \pl{ i < j }{ N-1} |\gamma_t(x_i) - \gamma_t(x_j)|^\beta \\
\times   \pl{k=1}{N-1} \Big\{ |\gamma_t(x) - \gamma_t(x_k)|^\beta \cdot  \ex{  - N \beta \widetilde{\varphi}_t(x_k)} \cdot |\ga_t^{\prime}(x_k)| \Big\} \, \dd \mathbf{x}
\label{definition 1pt Density}
\end{multline}
where we recall that $\widetilde{\varphi}_t = \Re \big( V_t \circ \gamma_t\big)$.
\end{definition}
\begin{proposition}\label{edgebound} There exist $t$-independent constants $C, C^\prime > 0$  such that
$$\max\left\{ \varrho^t_{N}(-\epsilon),  \varrho^t_{N}(1+\epsilon)\right\} \leq C \ex{-C^\prime N}$$
uniformly in $t \in \intff{0}{1}$.
\end{proposition}
\begin{proof}
We first write
\begin{align*}
\varrho^t_{N}(x) \, = \,
 |\gamma^\prime_t(x)| \ex{ - \beta N \widetilde{\varphi}_t(x) }
\cdot  \frac{ \mathbb{E}_{N-1, \ga_t^{\eps}} \Big[ \ex{ (N-1)\beta \int_\mathbb{R} \{ \ln|\gamma_t(x) - \gamma_t(y)| - \widetilde{\varphi}_t(y) \} \, \dd L_{N-1}^{(\mathbf{x})}(y) } \Big] }
{ \Upsilon_{N-1; \eps } (\intff{-\epsilon}{1+\epsilon}) }
\end{align*}
where
\begin{align*}
\Upsilon_{N-1; \eps } (\intff{-\epsilon}{1+\epsilon})  \,  =  \, \Int{-\epsilon}{1+\epsilon}  |\gamma^\prime_t(x)| \ex{-\beta N \widetilde{\varphi}_t(x)}
\mathbb{E}_{N-1, \ga_t^{\eps}} \bigg[ \ex{ (N-1)\beta \int_\mathbb{R} \{ \ln|\gamma_t(x) - \gamma_t(y)|-\widetilde{\varphi}_t(y) \} \, \dd L_{N-1}^{(\mathbf{x})}(y)} \bigg] \, \dd x \;.
\end{align*}
We begin by bounding $\Upsilon_{N-1; \eps }(\intff{-\epsilon}{1+\epsilon})$ from below.
First, observe the bound  $\dd x \geq \frac{\pi}{4} \dd\nu_{\mathrm{sc}}(x)$ and then apply Jensen's inequality to find
\bem
\Upsilon_{N-1;\eps}( \intff{-\epsilon}{1+\epsilon}) \, \geq \, \frac{\pi}{4} C_\gamma \ex{ - \beta \| \widetilde{\varphi_t} \|_{\mc{W}^{\infty}_{0}(\intff{-\eps}{1+\eps})} - (N-1)\beta C_\varphi }\\
\times  \exp\bigg\{ - (N-1)\beta \int_\mathbb{R} \widetilde{\varphi_t}(s) \, \dd \nu_{\mathrm{sc}}(s) -(N-1)\beta \,
\mathbb{E}_{N-1, \ga_t^{\eps}} \Big[ \int_\mathbb{R} \widetilde{\varphi}_{\Sg_t;\mathrm{eff}}(s) \, \dd L_{N-1}^{(\mathbf{x})}(s) \Big]  \bigg\} \, .
\end{multline}
Here, we remind that $\Sg_t$ has been introduced in Definition \ref{contours}. $\widetilde{\varphi}_{\Sg_t;\mathrm{eff}}$ refers to the effective potential
associated to the equilibrium measure $\nu_{\e{sc}}$ on $\intff{0}{1}$ and in presence of the external field $\wt{\vp}_{t}$. Note that $\widetilde{\varphi}_{\Sg_t;\mathrm{eff}}$,
by its very construction, is $t$-independent.

For any $\delta > 0$, consider a $\mc{C}^\infty(\intff{-\epsilon}{1+\epsilon})$ function $\psi_\delta$ which has the following properties:
\begin{enumerate}
 \item $ 0 \leq \psi_\delta \leq \widetilde{\varphi}_{\Sg_t;\mathrm{eff}}(-\epsilon) \,  = \,  \widetilde{\varphi}_{\Sg_t;\mathrm{eff}}(1+\epsilon)$;
 \item $\psi_\delta(x) = 0$ for $x \in \intff{\delta}{1-\delta}$;
 \item  $\psi_\delta(x) \, = \, \widetilde{\varphi}_{\Sg_t;\mathrm{eff}}(-\epsilon)  \, = \,   \widetilde{\varphi}_{\Sg_t;\mathrm{eff}}(1+\epsilon)$ for $x \in \intff{-\epsilon}{1+\epsilon}\setminus \intff{0}{1}$.
\end{enumerate}
Importantly, one has that  $\psi_\delta \geq \widetilde{\varphi}_{\Sg_t;\mathrm{eff}}$. Hence, by Corollary \ref{moments}, there exists a constant $C > 0$ such that
\begin{align*}
\mathbb{E}_{N-1, \ga_t^{\eps}} \left[ \int_\mathbb{R} \widetilde{\varphi}_{\Sg_t;\mathrm{eff}}(s) \, \dd L_{N-1}^{(\mathbf{x})}(s) \right]
\leq \frac{8}{\pi} \delta \widetilde{\varphi}_{\Sg_t;\mathrm{eff}}(-\epsilon) + C \| \psi_\delta \|_{ \mc{W}^{\infty}_1(\intff{-\epsilon}{1+\epsilon})} \sqrt{\frac{\ln N}{N}} \\
    + \frac{2 \| \widetilde{\varphi}_{\Sg_t;\mathrm{eff}} \|_{ \mc{W}^{\infty}_1( \intff{-\epsilon}{1+\epsilon})} }{ (N-1)^2 } \;.
\end{align*}
Thus,  for any $\Delta > 0$, there is an $N_0 \in \mathbb{N}$ independent of $t \in \intff{0}{1}$ such that
\begin{align*}
\Upsilon_{N-1;\eps }(\intff{-\epsilon}{1+\epsilon}) \, \geq \,
\ex{ -(N-1)\beta C_{\Sg_t}  -(N-1)\beta \int_\mathbb{R} \widetilde{\varphi}_t(s) \,\dd \nu_{\mathrm{sc}}(s) - (N-1) \beta \widetilde{\varphi}_{\Sg_t;\mathrm{eff}}(-\epsilon) \Delta} & & \forall N \geq N_0 \; .
\end{align*}
Next, turning to the numerator, for $\delta > 0$, let us bound
$$ \ln|\gamma_t(x) - \gamma_t(y)| \, \leq \, \ln_{\delta}|x-y| + \La(x,y) \qquad \e{where} \qquad  \La(x,y) = \ln \left| \frac{\gamma_t(x) - \gamma_t(y)}{x-y}\right| \, .$$
Then, there is a constant $C > 0$ (independent of $t \in \intff{0}{1}$ and $x \in \intff{-\epsilon}{1+\epsilon}$) such that
$\| \La(x,\cdot) \|_{ \mc{W}^{\infty}_1(\intff{-\epsilon}{1+\epsilon})} \leq C$.  From this it follows that there is a $C_\delta > 0$ such that
\bem
\mathbb{E}_{N -1, \ga_t^{\eps}} \bigg[ \ex{ (N-1)\beta \int_\mathbb{R} \{ \ln|\gamma_t(x) - \gamma_t(y)| - \widetilde{\varphi}_t(y)\} \, \dd L_{N-1}^{(\mathbf{x})}(y)} \bigg] \\
\leq \ex{ (N-1)\beta \int_\mathbb{R}\{ \ln_\delta|x-y| + \La(x,y) - \widetilde{\varphi}_t(y) \} \, \dd\nu_{\mathrm{sc}}(y) + C_\delta \sqrt{N \ln N}} \;.
\nonumber
\end{multline}
$\int_\mathbb{R} \ln_\delta|x-y| \,  \dd \nu_{\mathrm{sc}}(y)  \longrightarrow \int_\mathbb{R} \ln|x-y| \,  \dd\nu_{\mathrm{sc}}(y)$ as $\delta \downarrow 0$,
hence for any $\Delta > 0$ by taking $\delta > 0$ sufficiently small $\int_\mathbb{R} \ln_\delta|x-y| \,  \dd \nu_{\mathrm{sc}}(y) \leq \int_\mathbb{R} \ln|x-y| \,  \dd \nu_{\mathrm{sc}}(y) + \Delta$.
Then, for any $\Delta > 0$, there is an $N_0 \in \mathbb{N}$
\begin{align*}
\mathbb{E}_{N-1, \ga_t^{\eps} } \left[ \ex{(N-1)\beta \int_\mathbb{R} \{ \ln|\gamma_t(x) - \gamma_t(y)| - \widetilde{\varphi}_t(y) \} \, \dd L_{N-1}^{(\mathbf{x})}(y)} \right]
\, \leq  \, \ex{ (N-1)\beta \int_\mathbb{R}\{ \ln|\gamma_t(x)-\gamma_t(y)| - \widetilde{\varphi}_t(y) \} \, \dd \nu_{\mathrm{sc}}(y) + N \Delta }
\end{align*}
for all $N \geq N_0$.  Then, there exists $C > 0$ and $N_0 \in \mathbb{N}$, such that
\begin{align*}
\varrho^t_{N}(x) \leq C \ex{-(N-1)\beta \widetilde{\varphi}_{\Sg_t;\mathrm{eff}}(x) + N \Delta + (N-1)\beta \widetilde{\varphi}_{\Sg_t;\mathrm{eff}}(-\epsilon) \Delta } & & \forall N \geq N_0 \;.
\end{align*}
Setting $x = -\epsilon$ or $x = 1+\epsilon$ and taking $\Delta > 0$ sufficiently small completes the proof.
\end{proof}

\section{Lower bound on the complex partition function}\label{realDSsection}

Throughout this section we focus on the curve $\ga_t^{\eps}$, the restriction of $\gamma_t$ to
\beq
I_{\eps}=\intff{-\eps}{1+\eps}
\label{definition intervale eps voisinage 0 et 1}
\enq
\textit{c.f.} \eqref{definition courbe gamma t eps}, and on the associated contour $\Sigma_t = \gamma_t(I_{\eps})$, defined by \eqref{contourdef}. In this section we prove the following proposition.
\begin{proposition}[Lower bound on the complex partition function]\label{lowerbound}
There exists $t$-independent constants $C > 0$ and $N_0 \in \mathbb{N}$   such that
\begin{align}
\left| \frac{ \mathcal{Z}_{N ,\Sigma_t}[V_t] }{ \op{Z}_{N, \ga_t^{\eps}}[V_t] } \right| \geq C & & \forall N \geq N_0
\end{align}
uniformly in $t \in \intff{0}{1}$.
In particular $\mathcal{Z}_{N ,\Sigma_t}[V_t]  \neq 0$ so the complex model (\ref{complexmodel}) is well-defined.
\end{proposition}
Note that the upper bound $\op{Z}_{N, \ga_t^{\eps}}[V_t] \geq \big| \mathcal{Z}_{N ,\Sigma_t}[V_t] \big| $ is trivial.

\vspace{2mm}

We recall that $\epsilon^\prime > \epsilon > 0$ is chosen such that the original $S$-curve $\gamma$ in the neighbourhood of the support of the equilibrium measure,
and hence its $t$ interpolation $\ga_t$, is analytic on $\intff{-\epsilon^\prime}{1+\epsilon^\prime}$. Then we introduce $\chi : \mathbb{R} \longrightarrow \intff{0}{1}$,
a $\mc{C}^\infty(\mathbb{R})$ function such that $\chi(x) = 1$ for $x \in  I_{\eps}$ and $\chi(x) = 0$ for $x \in \mathbb{R} \setminus \intff{-\epsilon^\prime}{1+\epsilon^\prime}$.
Then let
\begin{align*}
a_t(x,y) &= \chi(x) \chi(y) \,  \mathrm{arg} \, \Big\{  \frac{\gamma_t(x) - \gamma_t(y)}{x-y}  \Big\} & & x,y \in \mathbb{R} \\
p_t(x) &= \chi(x) \, \mathrm{arg} \, \big\{ \gamma^\prime_t(x) \big\}  & & x \in \mathbb{R}
\end{align*}
so that $a_t$ and $p_t$ are smooth and compactly supported functions on $\mathbb{R}$.

\begin{remark} As stated $a_t$ and $p_t$ are not well-defined because $\arg$ is a multi-valued function. $\gamma^\prime_t(x)$ and $\frac{\gamma_t(x) - \gamma_t(y)}{x-y}$ are always nonzero and so we need only show that we can make a consistent choice of branch. We first fix $\arg \{ \gamma_t^\prime(0)\}$ by making some arbitrary choice of branch and then define its values elsewhere by requiring that $$\frac{\dd}{\dd x} \arg \{ \gamma_t(x)\} = \Im \frac{\gamma_t^{\prime\prime}(x)}{\gamma_t^{\prime}(x)}.$$ 
Similarly $\arg \{ \frac{\gamma_t^\prime(x) - \gamma_t(y)}{x-y} \}$ is defined for $x = y$ as $\arg \{ \gamma_t^\prime(x)\}$ and then for $x < y$ by requiring
$$ \frac{\partial}{\partial x} \arg \left\{\frac{\gamma_t(x) - \gamma_t(y)}{x-y}\right\}  = \Im \left( \frac{\gamma_t^\prime(x)}{\gamma_t(x) - \gamma_t(y)} - \frac{1}{x-y}\right).$$
and then extending to $y < x$ by symmetry.
\end{remark}
\begin{proposition}\label{formula}
It holds that
\begin{align}
\left| \frac{\mathcal{Z}_{N,\Sigma_t}[V_t]}{\op{Z}_{N, \ga_t^{\eps}}[V_t] } \right|
\, =  \, \Big| \mathbb{E}_{N , \ga_t^{\eps}} \big[  \ex{ \mc{G}_{\e{in}} } \big]\Big| \,,
\end{align}
where
\beq
\mc{G}_{\e{in}} \, = \, \frac{\i\beta N^2}{2} \Int{\mathbb{R}^2}{} a_t(x,y) \, \dd \mc{L}_N^{(\mathbf{x})} (x)  \dd \mc{L}_N^{(\mathbf{x})} (y)
+ \i N \Big(1 - \frac{\beta}{2} \Big) \Int{\mathbb{R}}{} p_t(x) \dd \mc{L}_N^{(\mathbf{x})} (x)
\label{ecriture expression D in}
\enq
and where we have introduced the empirical measure centred on $\nu_{\e{sc}}$:
\beq
\mc{L}_N^{(\mathbf{x})} =  L_N^{(\mathbf{x})} - \nu_{\e{sc}} \, .
\label{definition mesure empirique centree}
\enq
\end{proposition}

\begin{proof}

To prove this let us separate the integrand of \eqref{Cpartitionfunction} into a modulus and a phase. We observe that because $\beta$ is even $(\gamma_t(x_i) - \gamma_t(x_j))^\beta$
has the same phase as $\big(\frac{\gamma_t(x_i) - \gamma_t(x_j)}{x_i-x_j} \big)^\beta$. Then if we let $\mathrm{ph}(z) = \frac{z}{|z|}$ denote the phase of a complex number we have
\begin{align*}
\frac{\mathcal{Z}_{N,\Sigma_t}[V_t]}{\op{Z}_{N, \ga_t^{\eps}}[V_t] }  &=
\mathbb{E}_{N, \ga^{\eps}_t} \Bigg[ \mathrm{ph}\bigg( \pl{\substack{i,j=1 \\ i \neq j}}{N}
\Big( \frac{\gamma_t(x_i)-\gamma_t(x_j)}{x_i-x_j} \Big)^{\frac{\beta}{2}} \cdot \pl{i=1}{N} \Big\{ \gamma^\prime_t(x_i) \ex{ -\beta N V_t(\gamma_t(x_i)) } \Big\} \bigg) \Bigg] \\
&=  \mathbb{E}_{N, \ga^{\eps}_t} \Bigg[  \mathrm{ph}\bigg( \prod_{i,j=1 }^N \left( \frac{\gamma_t(x_i)-\gamma_t(x_j)}{x_i-x_j} \right)^{\frac{\beta}{2}}
\pl{i=1}{N} \Big\{ \gamma^\prime_t(x_i)^{1-\frac{\beta}{2}} \ex{-\beta N V_t(\gamma_t(x_i))} \Big\} \bigg) \Bigg]\\
&=  \mathbb{E}_{N, \ga^{\eps}_t} \Bigg[  \e{e}^{\frac{\i\beta N^2 }{2}\int_{\mathbb{R}^2} a_t(x,y) \, \dd L_N^{(\mathbf{x})}(x) \dd L_N^{(\mathbf{x})}(y)
+ \i(1-\frac{\beta}{2})N \int_{\mathbb{R}}  p_t(x) \, \dd L_N^{(\mathbf{x})}(x) - \i \beta N \int_{\mathbb{R}}  \Im\big( V_t(\gamma_t(x)) \big) \, \dd L_N^{(\mathbf{x})}(x)  } \Bigg]
\end{align*}

Let us now re-centre with respect to the equilibrium measure. By using \eqref{definition mesure empirique centree}, one gets
\begin{align}
\begin{split}\label{exponent}
&\frac{\mathcal{Z}_{N,\Sigma_t}[V_t]}{\op{Z}_{N, \ga_t^{\eps}}[V_t] } \, = \,  \ex{\mc{G}_{\e{sc}}}
\cdot \mathbb{E}_{N, \ga^{\eps}_t} \bigg[ \ex{ \mc{G}_{\e{out}} + \mc{G}_{\e{in}} }   \bigg].
\end{split}
\end{align}
There, $\mc{G}_{\e{in}}$ is as in \eqref{ecriture expression D in}, while
\beq
\mc{G}_{\e{sc}} \, = \, \frac{\i\beta N^2 }{2}\Int{\mathbb{R}^2}{} a_t(x,y) \, \dd^2 \nu_{\e{sc}}(x,y)
+ \i\Big(1-\frac{\beta}{2}\Big)N \Int{\mathbb{R}}{}  p_t(x) \, \dd \nu_{\e{sc}}(x)   - \i \beta N^2 \Int{\mathbb{R}}{}  \Im\big( V_t(\gamma_t(x)) \big) \, \dd \nu_{\e{sc}}(x)
\nonumber
\enq
and
\beq
\mc{G}_{\e{out}} \, = \, - \i \beta N^2  \Int{\mathbb{R}}{} \, \dd \mc{L}_N^{(\mathbf{x})}(x)  \Big\{  \Im\big( V(\gamma_t(x)) \big)   \, - \, \Int{\R}{}a_t(x,y) \, \dd \nu_{\e{sc}}(y) \Big\} \;.
\enq
By taking the $x$-derivative of $\mc{G}_{\e{out}}$'s integrand and applying the Sokhotski-Plemelj theorem, one observes that it coincides throughout $I_{\eps}$
with the derivative of $\frac{1}{2} \Im \big( \big[\Phi_{\Sg_t;\mathrm{eff}}\big]_+ + \big[\Phi_{\Sg_t;\mathrm{eff}}\big]_- \big)(\gamma_t(x))$
by virtue of \eqref{Vtdef}.
Hence, there exists $C_t \in \mathbb{R}$, possibly dependent on $t \in \intff{0}{1}$ but independent of $x \in I_{\eps}$,  such that
\begin{align*}
\Im\big( V_t(\gamma_t(x)) \big) - \Int{0}{1} \mathrm{arg} \Big\{ \frac{\gamma_t(x)-\gamma_t(y) }{ x - y } \Big\}  \, \dd \nu_{\mathrm{sc}}(y)
\, = \,  \frac{1}{2} \Im \big( \big[\Phi_{\Sg_t;\mathrm{eff}}\big]_+ + \big[\Phi_{\Sg_t;\mathrm{eff}}\big]_- \big)(\gamma_t(x))  + C_t,
\end{align*}
for any $x \in I_{\eps}$. Yet, by its very construction and analogously to \eqref{Imconstant} of Proposition \ref{steepestdescent},
the first term on the left hand side vanishes on $I_{\eps}$, hence ensuring that the right hand side is constant on that
interval. This entails that $\mc{G}_{\e{out}}=0$ and, by the fact that $\mc{G}_{\e{sc}} \in \i\R$, yields the claim.

\end{proof}

Moving forward,  because $a_t$ and $p_t$ are bounded and compactly supported we can define their Fourier transforms:

\beq
A_t(x,y) \, = \,  \Int{ \mathbb{R}^2 }{} a_t(x^\prime,y^\prime) \ex{ 2\pi \i y y^\prime - 2\pi \i x x^\prime}\, \dd x^\prime \, \dd y^\prime  \qquad \e{and} \qquad
P_t(x) \, =  \, \Int{ \mathbb{R} }{ }    p_t(x^\prime) \ex{ - 2 \pi \i x x^\prime}\, \dd x^\prime  \, .
\enq
By construction, these are Schwartz functions.
It will also be convenient to introduce the Fourier transform of the centred empirical measure \eqref{definition mesure empirique centree}
\beq
\phi_N(x)  \, =  \, N \Int{ \mathbb{R} }{ }   \ex{- 2 \pi \i x x^\prime}\, \dd \mc{L}_N^{(\mathbf{x})}(x^\prime) \;.
\nonumber
\enq
Then, by Plancherel's theorem, $\mc{G}_{\e{in}}$ introduced in \eqref{ecriture expression D in} may be recast as
\beq
\mc{G}_{\e{in}} = \frac{ \i \beta}{2} \Int{\mathbb{R}^2}{} A_t(x,y) \overline{\phi_N(x)} \phi_N(y) \, \dd x \, \dd y  \,+ \,
\i  \big(1 - \frac{\beta}{2} \big) \Int{\mathbb{R}}{} \overline{P_t(x)} \phi_N(x)  \,  \dd x \;,
\label{ecriture expression alternative pour Gin en Fourier}
\enq
in which $\ov{z}$ stands for the complex conjugate of $z \in \Cx$.

Let us define the \textit{real master operator} $\Xi_{\ga_t} : \mc{W}^{\infty}_{k}(I_{\eps}) \longrightarrow \mc{W}^{\infty}_{k-1}( I_{\eps} )$
which acts as
\beq
\Xi_{\ga_t}[f] (x) = \Re \Big( V^\prime_t(\gamma_t(x)) \gamma^\prime_t(x) \Big) f(x)  \, - \,
 \Int{0}{1}  \Re \bigg\{  \frac{\gamma^\prime_t(x) f(x) - \gamma^\prime_t(y) f(y)}{\gamma_t(x) - \gamma_t(y)}\bigg\}   \dd \nu_{\mathrm{sc}}(y).
\label{definition real master operator}
\enq
It is established in Appendix \ref{singularsection} that $\Xi_{\ga_t}$ is invertible on the closed subspace of co-dimension $1$
$\e{ker}[\mc{K}_{\ga_t}]$, in which $\mc{K}_{\ga_t}: \mc{W}_1^{\infty}(\Sg_t) \rightarrow \R$ is the continuous linear form introduced in \eqref{definition form lineaire K3}.
Note that $\mc{K}_{\ga_t}$ is normalised so that $\mc{K}_{\ga_t}[1]=1$ which entails that $f-\mc{K}_{\ga_t}[f] \in \e{ker}[\mc{K}_{\ga_t}]$ for any $f \in \mc{W}^{\infty}_{1}(I_{\eps})$.
In all the manipulations that follow, we will be dealing with $\mc{W}^{\infty}_{k}(I_{\eps})$ functions which are not necessarily
in $\e{ker}[\mc{K}_{\ga_t}]$. However, the functions of our interest will always appear in integrals against zero mass measures.
Hence, one may always translate them by $\mc{K}_{\ga_t}[f]$ without altering the result while ensuring membership of $\e{ker}[\mc{K}_{\ga_t}]$.
In this sense, we may think of the operator
\beq
\wt{\Xi}_{\ga_t}^{-1}\big[ f \big] \, = \,  \Xi_{\ga_t}^{-1}\big[ f - \mc{K}_{\ga_t}[f] \big]
\enq
as providing an extension of $\Xi_{\ga_t}^{-1}$ from $\e{ker}[\mc{K}_{\ga_t}]$ to $\mc{W}_{p}^{\infty}(I_{\eps})$. Still, one should observe that the explicit formula
for $\Xi_{\ga_t}^{-1}$ as an integral operator only involves differences of its argument at different points: it is thus invariant under constant shift.
Thus, in terms of formulae, $\Xi_{\ga_t}$ does coincide with its extension $\wt{\Xi}_{\ga_t}^{-1}$. We shall thus from now on identify this extension
with $\Xi_{\ga_t}^{-1}$, which means that we understand that operator as being given by the integral operator appearing in the \textit{rhs} of
\eqref{ecriture explicite inverse real master operator}.

Thus, within these conventions, for any $f \in \mc{W}_{1}^{\infty}(I_{\eps})$, one may introduce the functionals
\beqa
\mathfrak{m}_t[f] &= &\Big( \frac{1}{\beta}-\frac{1}{2} \Big) \cdot  \nu_{\mathrm{sc}}\Big( \mc{R}_{\ga_t} \circ\Xi_{\ga_t}^{-1} [f]  \Big) \, ,
\label{esperance Gaussienne} \vspace{3mm} \\
\mathfrak{V}_t[f] &= &  \frac{1}{\beta} \cdot \nu_{\mathrm{sc}}\Big(  f^\prime \Xi_{\ga_t}^{-1} [f] \Big) \, .
\label{variance Gaussienne}
\eeqa
There, $\nu_{\e{sc}}(f)=\Int{}{} f \dd \nu_{\e{sc}}$ stands for the linear statistic evaluated against $\nu_{\e{sc}}$ and
\beq
\mc{R}_{\ga_t}[f](s) \, =  \, \Re \Big( \f{ \ga_{t}^{\prime\prime}(s) }{  \ga_{t}^{\prime}(s) } \, + \, \Dp{s}  \Big) f(s) \;.
\enq

$\mathfrak{V}_t[f]$ will turn out to be the variance  of  a random variable, however it is not entirely obvious from the explicit formula \eqref{variance Gaussienne} that it is non-negative. We now establish this property.
\begin{lemma}
\label{Lemma positivite variance}
 For any $f \in \e{ker}[\mc{K}_{\ga_t}]$,  one has that  $\mathfrak{V}_t(f) \, \geq \,   0$.

\end{lemma}

\begin{proof}

Let $I_{\eps}$ be as in \eqref{definition intervale eps voisinage 0 et 1} and recall the functional $\mf{I}_{I_{\eps} }$ \eqref{Jfunctional} adapted to $\ga_t$ and let $\nu_{\e{sc}}$ be its unique minimiser.
Then for any function  $f \in \mc{C}^\infty(\mathbb{R})$ such that $| f^\prime | <1$,  define the transported measure $\nu^{f}_{\e{sc}}$ by
$$\Int{\mathbb{R}}{} g(x) \, \dd \nu^{f}_{\e{sc}}(x) \, =\,  \Int{\mathbb{R}}{} g(x+f(x)) \, \dd \nu_{\e{sc}}(x) \; . $$
We must have $\mf{I}_{\mc{J} }[\nu^{f}_{\e{sc}}] \geq \mf{I}_{\mc{J} }[ \nu_{\e{sc}} ]$. Thus, if we expand this relation with respect to $f$,
the linear term must vanish and the quadratic term must be non-negative.  The linear term gives
\begin{align}
-\frac{1}{2} \Int{\mathbb{R}^2}{} \Re  \Big( \frac{\gamma^\prime_t(x) f(x) - \gamma^\prime_t(y) f(y) }{ \gamma_t(x) - \gamma_t(y)}  \Big) \, \dd^2 \nu_{\e{sc}}(x,y)
\, + \,  \Int{\mathbb{R}}{}  \wt{\vp}_{t}^{\, \prime}(x)   f(x) \,\dd\nu_{\e{sc}}(x)   \,  = \,  0
\label{equality}
\end{align}
where $\wt{\vp}_{t}(x) = \Re\big( V_t( \ga_t(x) ) \big)$ and the quadratic term gives
\bem
\de \mf{I}_{\mc{J}; \nu_{\e{sc}} }^{(2)}[ f ] \, = \,  \frac{1}{2}  \Int{\mathbb{R}^2}{} \Re \bigg\{  \bigg( \frac{\gamma^\prime_t(x) f(x) - \gamma^\prime_t(y) f(y) }{ \gamma_t(x) - \gamma_t(y)}  \bigg)^2
\, -  \,  \frac{\gamma^{\prime\prime}_t(x) f^2(x) - \gamma^{\prime\prime}_t(y) f^2(y) }{ \gamma_t(x) - \gamma_t(y)}  \bigg\} \, \dd^2 \nu_{\e{sc}}(x,y) \\
\,  + \,   \Int{\mathbb{R}}{} \wt{\vp}_{t}^{\, \prime\prime }(x) f^2(x) \, \dd\nu_{\e{sc}}(x) \, \geq  \, 0 \, .
\end{multline}
Since it is a homogeneous functional of degree $2$, the relation extends to all $f$, \textit{i.e.} not necessarily satisfying $| f^\prime | <1$.

The left hand side of $\de \mf{I}_{\mc{J}; \nu_{\e{sc}} }^{(2)}[ f ]$ can be identified with $\nu_{\mathrm{sc}}\big(  f (\Xi_{\ga_t} [f] )^\prime \big)$.
Indeed, starting from the definition \eqref{definition real master operator} of $\Xi_{\ga_t} [f]$ one gets that
$ (\Xi_{\ga_t} [f])^\prime(x) = (A) + (B)$  where
\begin{align*}
(A) &=  \ \wt{\vp}_{t}^{\, \prime}(x)  f^{\prime}(x)  \, - \,  \Fint{0}{1} \Re  \bigg\{ \frac{\gamma^\prime_t(x) f^\prime(x)}{\gamma_t(x) - \gamma_t(y)} \bigg\} \dd \nu_{\mathrm{sc}}(y)  \;, \\
(B) &=  \wt{\vp}_{t}^{\, \prime\prime }(x)  f(x)  \, +  \, \Fint{0}{1} \Re  \bigg\{ \gamma^\prime_t(x) \frac{\gamma^\prime_t(x) f(x) - \gamma^\prime_t(y) f(y)}{ \big(\gamma_t(x) - \gamma_t(y) \big)^2} \,  - \,
\frac{\gamma^{\prime \prime}_t(x) f(x) }{\gamma_t(x) - \gamma_t(y)}  \bigg\}\, \dd\nu_{\mathrm{sc}}(y) \; .
\end{align*}
Upon symmetrising in  $x \leftrightarrow y$ and using \eqref{equality} with the replacement of $f$ with $f f^\prime$,  we observe that
$$\Int{0}{1} f(x) (A) \, \dd \nu_{\mathrm{sc}}(x) = 0 \, . $$
Then the  $x \leftrightarrow y$ symmetrisation of the second contribution yields
$$ \Int{0}{1} f(x) (B) \, \dd\nu_{\mathrm{sc}}(x)  =\de \mf{I}_{\mc{J}; \nu_{\e{sc}} }^{(2)}[ f ]  \, . $$

We thus conclude that for all $f \in \mc{C}^\infty(\mathbb{R})$,  $\nu_{\mathrm{sc}}\big(  f (\Xi_{\ga_t} [f] )^\prime \big) \geq 0$.
Then, upon taking $f= \Xi_{\ga_t}^{-1} [g]$ with $ g \in  \e{ker}[\mc{K}_{\ga_t}]$, it holds $\nu_{\mathrm{sc}} \big( g^\prime \,  \Xi_{\ga_t}^{-1} [g] \big) \geq 0$.

\end{proof}

We have now introduced enough notation and results so as to state the central limit theorem for the real model. The latter will be established with the help of the
machinery of Dyson-Schwinger equations (see \cite{Gu19} for an introduction to this method).

\begin{theorem}[CLT for the real model]\label{CLT}
Let $f \in \mc{C}^\infty(I_{\eps})$ with $I_{\eps}$ as in \eqref{definition intervale eps voisinage 0 et 1} . Then
 the following convergence in distribution holds for the constrained real model $\mathbb{P}_{N,\ga_t^{\eps}}$
\eqref{definition modele reel avec courbe eta t et pot Wt} associated with the curve $\ga_t^{\eps}$ \eqref{contourdef} of support $I_{\eps}$
and in external potential $V_t$ \eqref{interpolatingpotential}-\eqref{interpolatingpotential2}:
\beq
N\mc{L}_N^{(\mathbf{x})}(f) \overset{\mathrm{d}}{\longrightarrow} \mathcal{N}\big( \mathfrak{m}_t[f] , \mathfrak{V}_t[f] \big) \quad as \quad  N \to +\infty \;.
\enq
$\mc{L}_N^{(\mathbf{x})}$ appearing above is as introduced in \eqref{definition mesure empirique centree} where
\beq
\mc{L}_N^{(\mathbf{x})}(f)  \, = \, \Int{\mathbb{R}}{} f(s) \, \dd \mc{L}_N^{(\mathbf{x})}(s)  \;.
\enq
\end{theorem}
The variance $\mathfrak{V}_t[f]$ of the Gaussian is positive owing to Lemma \ref{Lemma positivite variance}, as it should be.

In order to prove Theorem \ref{CLT}, we shall now introduce the Dyson-Schwinger equations which relate moments of linear statistics,
and more generally averages of functions in many variables integrated  against the re-centred empirical measure
$\mc{L}_N^{(\mathbf{x})}$ as introduced in \eqref{definition mesure empirique centree}. To state these, we need several ingredients:

$\chi^{(p)}: \mc{W}_{\ell}^{\infty}(\R^n) \rightarrow  \mc{W}_{\ell}^{\infty}(\R^{n-1})$ stands for the operator
\beq
\chi^{(p)}[\phi]\big( \xi_1,\dots, \xi_{n-1} \big) \, = \, \phi\big( \xi_1,\dots, \xi_{p-1}, \xi_1,\xi_{p},\dots, \xi_{n-1} \big) \,.
\enq
In what follows the  space on which $\chi^{(p)}$ acts  will be indicated implicitly by the function it acts on. We shall further need the $\ga_t$ deformed non-commutative derivative
\beq
\mc{D}_{\ga_t}[f] (x,y) \, = \, \Re  \bigg( \frac{\gamma^\prime_t(x) f(x) - \gamma^\prime_t(y) f(y)}{\gamma_t(x) - \gamma_t(y)} \bigg) \;.
\enq
Finally, we extend the operators $\mc{D}_{\ga_t}$,  $\Xi_{\ga_t}$ and $\mc{R}_{\ga_t}$ to functions $\mc{F}\in \mc{W}^{\infty}_1(I_{\eps}^k)$ with $I_{\eps}=\intff{-\eps}{1+\eps}$
as
\beq
\mc{D}_{\ga_t}^{(1)}[\mc{F}](\bs{s}_{k+1}) \, = \, \mc{D}_{\ga_t}[\mc{F}(*,s_3,\dots, s_{k+1})](s_1,s_2)
\enq
and
\beq
\Xi_{\ga_t}^{(1)}[\mc{F}]( \bs{s}_{k}) \, = \, \Xi_{\ga_t}[\mc{F}(*,s_2,\dots, s_{k})](s_1) \;,
\quad
\mc{R}_{\ga_t}^{(1)}[\mc{F}]( \bs{s}_{k}) \, = \, \mc{R}_{\ga_t}[\mc{F}(*,s_2,\dots, s_{k})](s_1)
\enq
in which $*$ denotes the dummy variables on which the operator acts while $\bs{s}_p=(s_1,\dots, s_p)$.
Finally, we introduce the notation
\beq
\big< \mc{F} \big>_{\bigotimes\limits^{k}_{a=1} \nu_a  } \;  = \; \mathbb{E}_{N, \ga_t^{\eps} } \Big[ \Int{}{} \mc{F}( \bs{s}_k)  \pl{a=1}{k} \dd \nu_a (s_a) \Big] \;.
\enq
Then, the loop equations take the form
\bem
\big< \mc{F} \big>_{\bigotimes\limits^{k+1} \mc{L}_N^{(\mathbf{x})}  } \, = \, \f{1}{2} \Big< \mc{D}_{\ga_t}^{(1)} \circ \big(\Xi_{\ga_t}^{(1)} \big)^{-1}[\mc{F}] \Big>_{\bigotimes\limits^{k+2}\mc{L}_N^{(\mathbf{x})}  }
\, + \, \f{1}{N} \Big( \f{1}{\be}- \f{1}{2} \Big) \Big< \mc{R}_{\ga_t}^{(1)}\circ \big(\Xi_{\ga_t}^{(1)} \big)^{-1}[\mc{F}] \Big>_{L_N^{(\mathbf{x})} \bigotimes\limits^{k}  \mc{L}_N^{(\mathbf{x})}  } \\
\, + \, \f{1}{N^2 \be} \sul{a=1}{k}  \Big< \chi^{(a+1)}\big[\Dp{a+1}  \big(\Xi_{\ga_t}^{(1)} \big)^{-1}[\mc{F}] \big] \Big>_{L_N^{(\mathbf{x})} \bigotimes\limits^{k-1}  \mc{L}_N^{(\mathbf{x})}  }
\, + \, \mf{B}_{k+1}\Big[  \big(\Xi_{\ga_t}^{(1)} \big)^{-1}[\mc{F}] \Big]
\label{ecriture equation des boucles}
\end{multline}
in which, denoting $\ov{\mathbb{E}}_{N, \ga_t^{\eps} }= \op{Z}_{N,\ga_t}[V_t] \cdot  \mathbb{E}_{N, \ga_t^{\eps} }$, one has
\beq
 \mf{B}_{k+1}[ \mc{F} ]  \, = \, \f{- 1}{N^2 \be} \sul{ \a \in \Dp{} I_{\eps} }{} \f{ 1 }{ \op{Z}_{N,\ga_t}[V_t]  }
\Dp{\a}  \ov{\mathbb{E}}_{N, \ga_t^{\eps} } \Big[ \Int{}{} \mc{F}(u, \bs{s}_k)  \pl{a=1}{k} \dd \mc{L}_N^{(\mathbf{x})}  (s_a) \Big]_{\mid u = \a} \;.
\label{definition operateur de bord eqn SD ordre k}
\enq
%
%
%
%
%
%
%
%
%
%
%
%
%
%

The Dyson-Schwinger equations can be derived by integration by parts or invariance of the partition function under changes of variables. We refer to
\cite{Borot:2013aa,KozBorotGuionnetLargeNBehMulIntOfToyModelSoVType,Gu19} for the details of the various derivations.

First of all, we observe that the boundary term $ \mf{B}_{k+1}[ \mc{F} ]$ defined in \eqref{definition operateur de bord eqn SD ordre k} only contributes
to an exponentially small order in $N$.

\begin{lemma}
\label{Lemma estimation partie bord dans SD}

There exists $C, C^{\prime}>0$ such that the following upper bound holds
\beq
\big| \mf{B}_{k+1}[ \mc{F} ] \big| \, \leq \, C^{\prime}  \norm{ \mc{F} }_{ \mc{W}^{\infty}_0(I_{\eps}^{k+1}) } \cdot \ex{- C N} \;.
\enq
\end{lemma}
\begin{proof}

Given  $\a \in \{-\eps, 1+\eps \}$, it holds
\bem
\msc{S}_{\a} \, = \, \Dp{\a}  \ov{\mathbb{E}}_{N, \ga_t^{\eps} } \Big[ \Int{}{} \mc{F}(u, \bs{s}_k)  \pl{a=1}{k} \dd \mc{L}_N^{(\mathbf{x})}  (s_a) \Big]_{\mid u = \a} \, = \, N
  \frac{ \ex{-N \beta \widetilde{\varphi}_t(\a) } |\ga_t^{\prime}(\a)| }{  \e{sgn}(\a)  \cdot  }
\Int{-\eps }{1+\epsilon} \pl{ i < j }{ N-1} |\gamma_t(x_i) - \gamma_t(x_j)|^\beta \\
\times   \pl{k=1}{N-1} \Big\{ |\gamma_t(\a) - \gamma_t(x_k)|^\beta \cdot  \ex{  - N \beta \widetilde{\varphi}_t(x_k)} \cdot |\ga_t^{\prime}(x_k)| \Big\}
\cdot \Big\{ \Int{}{} \mc{F}(\a, \bs{s}_k)  \pl{a=1}{k} \dd \mc{L}_N^{(\wt{\mathbf{x}})}  (s_a) \Big\} \cdot  \dd \mathbf{x}
\nonumber
\end{multline}
where $\wt{\mathbf{x}}=(\a,\bs{x})$. Then, direct bounds yield
$$\big| \msc{S}_{\a} \big| \, \leq \, N 2^k \, \op{Z}_{N,\ga_t}[V_t]\, \varrho_{N}^{t}(\a)  \,  \norm{ \mc{F} }_{ \mc{W}^{\infty}_0(I_{\eps}^{k+1}) }  $$
with the $1$-point density $\varrho_{N}^{t}$ defined in \eqref{definition 1pt Density}. It then only remains to invoke Proposition \ref{edgebound}.

\end{proof}

\begin{definition}[Multilinear norm]

Given $\ell, k \in \mathbb{N}$, the multilinear norm of order $\ell$ in $k$ variables takes the form
\beq
S_{\ell,k} = \sup_{\substack{ \mc{F} \in \mc{W}_{\ell}^{\infty}( I_{\eps}^k ) \\
\| \mc{F} \|_{ \mc{W}_{\ell}^{\infty}( I_{\eps}^k ) } \leq 1}} \Big| \big< \mc{F} \big>_{\bigotimes\limits^{k} \mc{L}_N^{(\mathbf{x})}  } \Big| \;.
\enq
\end{definition}
Note that $S_{\ell+1,k} \leq S_{\ell,k}$.  In what follows, we will show that for any $k \in \mathbb{N}$ there is an $\ell$ sufficiently large (depending on $k$) such that
$S_{\ell,k} = \e{O}(N^{-k})$.  This is an improvement of the bound given in Corollary \ref{moments} which implies that $S_{\ell,k} = \e{O}\big(\left( \frac{\ln N}{N} \right)^\frac{k}{2}\big)$.

\begin{lemma}\label{quadraticbound}
Then there exists a constant $C_\ell > 0$ such that for all $k, \ell \geq 0$ and  $\mc{F} \in  \mc{W}_{\ell+1}^{\infty}( I_{\eps}^{k} )$,
\beq
\big|\big|  \mc{D}_{\ga_t}^{(1)} [\mc{F}]  \big|\big|_{ \mc{W}_{\ell}^{\infty}( I_{\eps}^{k+1} ) } \, \leq \, C_{\ell}
\norm{ \mc{F}  }_{ \mc{W}_{\ell+1}^{\infty}( I_{\eps}^{k} ) }\, ,  \qquad
\big|\big|   \big(\Xi_{\ga_t}^{(1)} \big)^{-1}[\mc{F}]  \big|\big|  _{ \mc{W}_{\ell}^{\infty}( I_{\eps}^{k} ) } \, \leq \, C_{\ell}
\norm{ \mc{F}  }_{ \mc{W}_{\ell+1}^{\infty}( I_{\eps}^{k} ) }
\label{borne sur derivee NC et inverse Xi multivars}
\enq
as well as
\beq
\big|\big|  \mc{R}_{\ga_t}^{(1)} [\mc{F}]   \big|\big|  _{ \mc{W}_{\ell}^{\infty}( I_{\eps}^{k} ) } \, \leq \, C_{\ell}
\norm{ \mc{F}  }_{ \mc{W}_{\ell+1}^{\infty}( I_{\eps}^{k} ) }
\label{borne sur operateur R Gat multivars}
\enq
with all bounds being uniform in $t\in \intff{0}{1}$.
\end{lemma}

\begin{proof}

To start with, observe that one has the integral representation
\bem
 \mc{D}_{\ga_t}^{(1)} [\mc{F}] (x,y,\bs{s}_{k-1})   \, = \,
\Re\bigg\{ \frac{ x-y }{ \gamma_t(x)-\gamma_t(y) }  \Big[ \mc{F}(y,\bs{s}_{k-1}) \f{ \gamma_t^{\prime}(x)-\gamma_t^{\prime}(y) }{ x-y }  \\
\, + \,  \gamma^\prime_t(x) \Int{0}{1} \hspace{-1mm} \dd u  \, \Dp{1}\mc{F}(y + u(x-y) ,\bs{s}_{k-1})    \Big] \bigg\} \;.
\nonumber
\end{multline}
By Lemma \ref{curvelowerbound}, it holds $\frac{x-y}{\gamma_t(x)-\gamma_t(y)} \leq C_\gamma^{-1}$ for some $C_\gamma > 0$ and $\norm{\ga_t}_{ \mc{W}^{\infty}_p(I_{\eps}) }\leq  C_p$,
for some $C_p$, both uniformly in  $t \in \intff{0}{1}$. This entails the first bound in \eqref{borne sur derivee NC et inverse Xi multivars}.
The bound \eqref{borne sur operateur R Gat multivars} follows from direct estimation and the mentioned results of Lemma \ref{curvelowerbound}. Finally, the
second bound in \eqref{borne sur derivee NC et inverse Xi multivars} is a direct consequence of Proposition \ref{realmasterinverse}.

\end{proof}

Note that the above results immediately yield the three upper bounds
\beqa
\Big|  \Big< \mc{D}_{\ga_t}^{(1)} \circ \big(\Xi_{\ga_t}^{(1)} \big)^{-1}[\mc{F}] \Big>_{\bigotimes\limits^{k+2}\mc{L}_N^{(\mathbf{x})}  }  \Big|  & \; \leq \;  &
C_{\ell} \cdot \| \mc{F} \|_{ \mc{W}_{\ell+2}^{\infty}( I_{\eps}^{k+1} ) } \cdot S_{\ell, k+2} \, ,  \label{ecriture borne sup sur contri derivee NC} \vspace{2mm}\\
\Big|  \Big< \mc{R}_{\ga_t}^{(1)} \circ \big(\Xi_{\ga_t}^{(1)} \big)^{-1}[\mc{F}] \Big>_{\bigotimes\limits^{k+1}\mc{L}_N^{(\mathbf{x})}  }  \Big|   &  \; \leq \; &
C_{\ell} \cdot \| \mc{F} \|_{ \mc{W}_{\ell+2}^{\infty}( I_{\eps}^{k+1} ) } \cdot S_{\ell, k+1} \, , \label{ecriture borne sup sur contri partie entropique} \vspace{2mm} \\
\Big|  \Big< \chi^{(a+1)}\big[\Dp{a+1}  \big(\Xi_{\ga_t}^{(1)} \big)^{-1}[\mc{F}] \big] \Big>_{\bigotimes\limits^{k}  \mc{L}_N^{(\mathbf{x})}  } \Big|  & \; \leq \; &
  C_{\ell} \cdot \| \mc{F} \|_{ \mc{W}_{\ell+2}^{\infty}( I_{\eps}^{k+1} ) } \cdot S_{\ell, k} \;.
\label{ecriture borne sup sur contri derivee fct multipts}
\eeqa

We now pass on to establishing an \textit{a priori} bound on the $k^{\e{th}}$ linear statistic of $\mc{F}\in \mc{W}^{\infty}_{\ell}(I_{\eps}^k)$, with $\ell$ large enough.

\begin{proposition}
\label{Proposition a priori croissance statistique k lineaire}

 For any $k \in \mathbb{N}$, there exists $\wt{C}_k>0$ such that given $\mc{F}\in \mc{W}^{\infty}_{3k}(I_{\eps}^k)$ one has the following upper bound
\beq
\Big|  \big<  \mc{F}   \big>_{ \bigotimes\limits^{k}\mc{L}_N^{(\mathbf{x})}  }  \Big|  \,  \leq \,
\wt{C}_{k} \cdot \Big(  \f{\ln N }{ N } \Big)^{\f{k}{2} } \cdot \| \mc{F} \|_{ \mc{W}_{3k}^{\infty}( I_{\eps}^{k} ) } \;.
\enq

\end{proposition}

\begin{proof} Let $\mc{F}\in \mc{C}^{p}(I_{\eps}^k)$.  Then by the Whitney extension theorem there exists a compactly supported extension $\mc{F}_{\mf{e}}$ to $\mc{C}^{p}(\R^k)$ such that $\norm{ \mc{F}_{\mf{e}} }_{ \mc{W}^{\infty}_p(\R^k) } \, \leq \, C_{p,k} \cdot \norm{ \mc{F} }_{ \mc{W}^{\infty}_p(I_{\eps}^k) }$ for some constant $C_{p,k}>0$.

Then by Plancherel's theorem
\begin{align*}
\Big|\big<  \mc{F}   \big>_{ \bigotimes\limits^{k}\mc{L}_N^{(\mathbf{x})}  } \Big|  \; &= \;    \Big| \mathbb{E}_{N, \ga_t^{\eps} }\Big[ \Int{\mathbb{R}^k}{} \wh{ \mc{F}_{\mf{e}} }(\bs{u}_k) \, \overline{\xi_N(u_1)} \dots \overline{\xi_N(u_k)} \, \dd^k u \Big]\Big|\\
&\leq    \Int{\mathbb{R}^k}{} |\wh{ \mc{F}_{\mf{e}} }(\bs{u}_k)| \,   \prod_{m=1}^k \big(\mathbb{E}_{N, \ga_t^{\eps} } \big[ |\xi_N(u_m)|^k \big]\big)^\frac{1}{k} \, \dd^k u
\end{align*}
where $\xi_N(u) = \int_\mathbb{R} e^{-2\pi i  s u } \, \dd \mc{L}_N^{(\mathbf{x})}(s)$. Then by Corollary \ref{moments}
$$
\mathbb{E}_{N, \ga_t^{\eps} } \big[ |\xi_N(u)|^k \big] \leq C_k (1+|u|)^{k} \Big(\frac{\ln N}{N}\Big)^{\frac{k}{2}} \; .
$$
where we have replaced $\mc{L}_N^{(\mathbf{x})}$ by $\overline{\mc{L}}_N^{(\mathbf{x})}$ at the price of terms of order $N^{-3} \ll 
\sqrt{\frac{\ln N}{N}}$. The result then immediately follows upon taking $p=3k$ and integrating the Fourier transform by parts, and observing that $\prod_{m=1}^k (1+|u_m|)^{-2}$ is integrable on $\mathbb{R}^k$.

\end{proof}

It follows from Lemma \ref{Lemma estimation partie bord dans SD} and equations \eqref{ecriture borne sup sur contri derivee NC}, \eqref{ecriture borne sup sur contri partie entropique}
and \eqref{ecriture borne sup sur contri derivee fct multipts} applied to the Dyson-Schwinger equations \eqref{ecriture equation des boucles} that the multilinear norms satisfy, for any $k \geq 0$,  the bounds
\begin{align}\label{DSbound}
S_{\ell+2,k+1} \,  \leq \,   C_{\ell,k} \Big( S_{\ell,k+2}  \, + \,  \frac{1}{ N } S_{\ell,k}\mathbbm{1}_{k \geq 1}
 +  \frac{1}{N} S_{\ell,k+1} \, + \,  \frac{1}{ N^2} S_{\ell,k-1}\mathbbm{1}_{k \geq 2} + \ex{-CN}\Big)
\end{align}
for some $\ell, k$ dependent constants that are uniform in $N$ and $t \in \intff{0}{1}$.

\begin{lemma}
\label{normbound} There exist constants $C_{k, \ell} > 0$ and sequences $b_k$ such that
$$ S_{ \ell , k} \leq C_{k,\ell} \frac{1}{N^k}   $$
for all $k \in \mathbb{N}$ and $\ell \geq b_k$, this holding uniformly in $t \in \intff{0}{1}$.

\end{lemma}

\begin{proof}

One first shows that there exist sequences $a_k^{(j)}$, $0\leq j \leq k$, and constants $C_{j,k,\ell}$ such that
 \begin{align*}
S_{\ell, k } \leq C_{j,k,\ell} \left( \frac{\ln N}{N} \right)^\frac{k+j}{2} \qquad  \forall j, k \in \mathbb{N}, \, \,  j \leq  k \quad \e{and} \quad
\ell \geq a_k^{(j)} \;.
\end{align*}
For $j=0$, this is a consequence of Proposition \ref{Proposition a priori croissance statistique k lineaire} and the result is then established by induction on $j$
with the help of \eqref{DSbound}. Once the result is shown to hold at $j=k$, it is enough to invoke \eqref{DSbound} once more so as to conclude.

\end{proof}

The moral of this method is that we are applying dominant balance analysis to \eqref{ecriture equation des boucles}.
The bound in Lemma \ref{normbound} could also be achieved by the transport method (or "change of variables")
method in Section 4.3.3 of \cite{Gu19}. This method of repeatedly inverting the master operator to improve bounds was used by Borot and Guionnet (see Lemma 4.2 of \cite{Gu19}).

\begin{definition}
Let $\Psi_{G}$ be a multilinear functional defined on finite sets of $\mc{C}^\infty(\intff{-\epsilon}{1+\epsilon})$ functions as follows.  $\Psi_{G}(\emptyset) = 1$ and for $k \geq 0$
\begin{align*}
\Psi_{G}(\{ f_1, \dots, f_{k+1} \}) &=  \mf{m}_t[f_{k+1}] \cdot \Psi_{G}(\{ f_1, \dots, f_k\})  \, +  \, \frac{1}{\beta} \sul{q=1}{k} \mf{C}_t[f_q,f_{k+1}] \cdot \Psi_{G}(\{ f_1, \dots, f_k\} \setminus \{ f_q\}) \;.
\end{align*}
Above, $\mf{m}_t[f]$ is as introduced in \eqref{esperance Gaussienne} and we have set
\beq
\mf{C}_t[f,g]\, = \, \nu_{\mathrm{sc}}\big(  f^\prime \Xi_{\ga_t}^{-1} [g] \big)\;.
\label{definition covariance}
\enq
\end{definition}
Note that $\mf{C}_t[f,f]=\mf{V}_t[f]$  with $\mf{V}_t$ introduced in \eqref{variance Gaussienne}. Moreover, analogously to Lemma \ref{Lemma positivite variance} one shows that
$\mf{C}_t[f,g]$ is symmetric in $f,g$.  Note that $\Psi_{G}(\{ f_1, \dots, f_k\})$
so defined exactly coincides with the expectation $\mathbb{E}_{G}[X_1\cdots X_k]$ of $k$ Gaussian random variables
such that $\mathbb{E}_{G}[X_a] = \mf{m}_t[f_a]$ and $\e{Cov}[X_a, X_b] \; = \; \mf{C}_t[f_a,f_b]$. Indeed, elementary calculations show that the latter satisfies the induction
\beq
\mathbb{E}_{G}[X_1\cdots X_k]\, = \, \mf{m}_t[f_k] \cdot \mathbb{E}_{G}[X_1\cdots X_{k-1}] \, + \, \sul{q=1}{k-1} \e{Cov}[X_q, X_k] \cdot \mathbb{E}_{G}[X_1\cdots \wh{X}_q \cdots  X_{k-1}]
\enq
with the $\wh{}\,$ indicating that the variable $X_q$ is removed from the product.

\begin{proposition}\label{momentconv}
 Let $f_a \in \mc{C}^{\infty}( I_{\eps} )$, with $a=1,\dots, k$ and set
 $\mc{F}(\bs{s}_k)=\pl{a=1}{k}f_a(s_a)$. Then, the below asymptotic expansion holds
\begin{align*}
  \big<  \mc{F}   \big>_{ \bigotimes\limits^{k}\mc{L}_N^{(\mathbf{x})}  }   \, =  \, \f{1}{N^k} \, \Psi_G(\{ f_1, \dots, f_k \})
\, + \,  \e{O}\Big(  \f{1}{N }  \| \mc{F} \|_{ \mc{W}^{\infty}_{\ell_k}( I_{\eps}^k ) } \Big) \;,
\end{align*}
for some sequence $\ell_{k}$. Above, the control on the remainder is uniform in $t \in \intff{0}{1}$.
\end{proposition}
\begin{proof}
The proof goes by induction.  The claim is trivial when $k=0$ since the products are empty, which establishes the base case.
Next, we assume that the expansion holds up to $k$ functions and set
$$\Phi_N(\{ f_1, \dots, f_k \}) \overset{ \mathrm{def}}{=} N^k   \big< \pl{a=1}{k}f_a(s_a)    \big>_{ \bigotimes\limits^{k}\mc{L}_N^{(\mathbf{x})}  } \,.  $$
Then, denoting $\mc{F}_{k+1}(\mf{s}_{k+1})=f_{k+1}(s_1) \pl{a=1}{k}f_a(s_{a+1})$, we get that
the Dyson-Schwinger equations \eqref{ecriture equation des boucles} may be recast as
\bem
\Phi_N(\{ f_1, \dots, f_{k+1} \})  = \f{ N^{k+1} }{ 2 }   \Big< \mc{D}_{\ga_t}^{(1)} \circ \big(\Xi_{\ga_t}^{(1)} \big)^{-1}[\mc{F}_{k+1}] \Big>_{\bigotimes\limits^{k+2}\mc{L}_N^{(\mathbf{x})}  }
\, + \,\mf{m}_t[f_{k+1}] \,  \Phi_N\big( \big\{ f_1, \dots, f_k \big\} \big) \\
  + \frac{1}{N}\left( \frac{1}{\beta}- \frac{1}{2} \right) \Phi_N\big( \big\{ f_1, \dots, f_k,  \mc{R}_{\ga_t}\circ\Xi_{\ga_t}^{-1}[f_{k+1}]  \big\} \big)
  + \frac{1}{\beta} \sul{q=1}{k} \mf{C}_{t}[f_q,f_{k+1}] \, \Phi_N \big( \big\{ f_1, \dots, f_k \big\} \setminus \{ f_q  \} \big)\\
  + \frac{1}{\beta N} \sum_{q=1}^k   \Phi_N\big( \big\{ f_1, \dots, f_k, f_q^\prime \Xi_{\ga_{t} }^{-1}[f_{k+1}] \big\} \setminus \{ f_q \} \big)
\, + \, \mf{B}_{k+1}\Big[  \big(\Xi_{\ga_t}^{(1)} \big)^{-1}[\mc{F}_{k+1}] \Big] \; .
\end{multline}
Thus, upon centring with respect to $\Psi_G$ for $p\in \{k-1,k, k+1\}$
\beq
\Phi_N(\{ f_1, \dots, f_{p} \})  \, = \, \Psi_G(\{ f_1, \dots, f_{p} \})  \, + \, \de \Phi_N(\{ f_1, \dots, f_{p} \}) \;,
\enq
we get
\bem
\de \Phi_N(\{ f_1, \dots, f_{k+1} \})  = \f{ N^{k+1} }{ 2 }   \Big< \mc{D}_{\ga_t}^{(1)} \circ \big(\Xi_{\ga_t}^{(1)} \big)^{-1}[\mc{F}_{k+1}] \Big>_{\bigotimes\limits^{k+2}\mc{L}_N^{(\mathbf{x})}  }
\, + \,\mf{m}_t[f_{k+1}] \, \de \Phi_N\big( \big\{ f_1, \dots, f_k \big\} \big) \\
  + \frac{1}{N}\left( \frac{1}{\beta}- \frac{1}{2} \right) \Phi_N\big( \big\{ f_1, \dots, f_k,  \mc{R}_{\ga_t}\circ\Xi_{\ga_t}^{-1}[f_{k+1}]  \big\} \big)
  + \frac{1}{\beta} \sul{q=1}{k} \mf{C}_{t}[f_q,f_{k+1}] \, \de \Phi_N \big( \big\{ f_1, \dots, f_k \big\} \setminus \{ f_q  \} \big)\\
  + \frac{1}{\beta N} \sum_{q=1}^k   \Phi_N\big( \big\{ f_1, \dots, f_k, f_q^\prime \Xi_{\ga_{t} }^{-1}[f_{k+1}] \big\} \setminus \{ f_q \} \big)
\, + \, \mf{B}_{k+1}\Big[  \big(\Xi_{\ga_t}^{(1)} \big)^{-1}[\mc{F}_{k+1}] \Big] \; .
\end{multline}
Upon using Lemma  \ref{Lemma estimation partie bord dans SD}, the estimate \eqref{ecriture borne sup sur contri derivee NC} along with Lemma \ref{normbound}
one readily infers that for some $c_k$,
\beq
\de \Phi_N(\{ f_1, \dots, f_{k+1} \})  \; = \; \e{O}\Big( \f{1}{N} \norm{ \mc{F}_{k+1} }_{\mc{W}^{\infty}_{c_k}(I_{\eps}^{k+1}) } \Big)
\enq
with a control that is uniform in $t \in \intff{0}{1}$.
\end{proof}

From this it follows that the collection $\left\{ N \mc{L}_N^{(\mathbf{x})}(f_p) \right\}_{p=1}^k$ converges in distribution towards a set of $k$ jointly Gaussian random variables with
means $\mf{m}_t[f_p]$ as introduced in \eqref{esperance Gaussienne}  and covariances $\mathfrak{C}_t[f_p,f_q]$ as introduced in \eqref{definition covariance}.
In particular, Proposition \ref{momentconv}  entails Theorem \ref{CLT}.

\vspace{2mm}

Recall that $\phi_N(x) \overset{\mathrm{def}}{=} N \mc{L}_N^{(\mathbf{x})} (\ex{-2\pi \i x *})$.
Proposition \ref{momentconv} has established that  $\{ \phi_N(x) \}_{x \in \mathbb{R}}$ converges in distribution to a Gaussian family
$\{ \phi(x) \}_{x \in \mathbb{R}}$ where
\beq
\mathbb{E}_{G;t}\big[ \phi(x) \big] \, = \,  \mathfrak{m}_t[\ex{-2\pi \i x *}] \qquad \e{and} \qquad
\mathrm{Cov}_G\big( \phi(x), \overline{\phi(y)} \big) \, = \,  \mathfrak{C}_t[\ex{ - 2\pi \i x * }, \ex{2\pi i y *}] \, .
\nonumber
\enq
\begin{proposition}\label{proplim}
Recall $\mc{G}_{\e{in}}$ \eqref{ecriture expression D in} along with its Fourier representation \eqref{ecriture expression alternative pour Gin en Fourier}.
One has the below convergence of expectations
\begin{align}
\begin{split}\label{gaussianexp}
\mathbb{E}_{N, \ga_t} \big[ \ex{\mc{G}_{\e{in}}} \big]
\underset{N \to +\infty}{\longrightarrow}
 \mathbb{E}_{G;t}  \big[ \ex{\msc{G}_{\e{in}}} \big]
\end{split}
\end{align}
uniformly in $t \in \intff{0}{1}$ and with
\beq
\msc{G}_{\e{in}} \, = \, \frac{i \beta}{2} \Int{\mathbb{R}^2}{} A_t(x,y) \overline{\phi(x)} \phi(y) \, \dd x \, \dd y \, + \,  \i  \Big(1 - \frac{\beta}{2} \Big) \Int{\mathbb{R}}{} \overline{P_t(x)} \phi(x)  \,  \dd x  \;.
\label{definition msc G in}
\enq
\end{proposition}
\begin{proof}

We first establish the convergence of the moments of the random variable $\mc{G}_{\e{in}}$.
$\mathbb{E}_{N,\ga_t}[ \mc{G}_{\e{in}}^m]$ may be expanded as
\bem
\mathbb{E}_{N,\ga_t}[ \mc{G}_{\e{in}}^m] \, = \, \sul{k=0}{m} C^k_m \Big( \i \tfrac{\be}{2} \Big)^{k}   \Big( \i \big[1- \tfrac{\be}{2}  \big]\Big)^{m-k}
\Int{ \mathbb{R}^{2k}  }{}  \dd \mathbf{x}\, \dd \mathbf{y} \hspace{-2mm} \Int{\R^{m-k} }{} \hspace{-2mm} \dd \mathbf{s} \; \pl{ i=1 }{ m } A_t(x_i,y_i) \pl{j=1}{k-m} \overline{P(s_j)}  \\
 \times  \mathbb{E}_{N,\ga_t}\bigg[  \pl{i=1}{m} \Big\{ \ov{\phi}_N(x_i) \phi_N(y_i) \Big\} \pl{j=1}{k-m} \phi_N(s_j) \bigg] \;.
\nonumber
\end{multline}
Above, we could invoke Fubini's theorem to exchange the orders of integration  at finite $N$ since  $|\phi_N(x)| \leq 2N$, and $A_t$ and $P_t$ are in the Schwartz class.
Then by Proposition \ref{momentconv} we have
\bem
\mathbb{E}_{N,\ga_t}\bigg[  \pl{i=1}{m} \Big\{ \ov{\phi}_N(x_i) \phi_N(y_i) \Big\} \pl{j=1}{k-m} \phi_N(s_j) \bigg]  \\
= \mathbb{E}_{G;t}\bigg[  \pl{i=1}{m} \Big\{ \ov{\phi}(x_i) \phi(y_i) \Big\} \pl{j=1}{k-m} \phi(s_j) \bigg]
+ \e{O}\bigg(  \f{1}{N}  \norm{  \mc{F}_{m,k} }_{ \mc{W}^{\infty}_{\ell_{m+k}}(I_{\eps}^{m+k}) }  \bigg)\,,
\nonumber
\end{multline}
in which
\beq
\mc{F}_{m,k}\big(\mathbf{x}_k, \mathbf{y}_k, \mathbf{s}_{m-k} \mid  \bs{\la}_k, \bs{\mu}_k, \bs{\nu}_{m-k} \big)  \, = \,
\pl{a=1}{k} \ex{ 2\i\pi (x_a \la_a - y_a \mu_a)} \pl{a=1}{m-k} \ex{-2\i\pi s_a \nu_a} \;.
\enq
Hence,
\beq
  \norm{  \mc{F}_{m,k} }_{ \mc{W}^{\infty}_{\ell}(I_{\eps}^{m+k}) }  \, \leq \,
\pl{a=1}{k}\Big\{   \big( 1+|x_a| \big)^{\ell} \big( 1+|y_a| \big)^{\ell} \Big\} \cdot
                    \pl{a=1}{m-k}  \big( 1+|s_a| \big)^{\ell} \, .
\enq

Now, observe that Proposition \ref{realmasterinverse} along with Lemma \ref{curvelowerbound} leads to the estimates
\beq\label{growth}
|\mathfrak{m}_t[f]| \, \leq \,  C \| f \|_{  \mc{W}^{\infty}_{2}(I_{\eps})} \qquad \e{and} \qquad
|\mathfrak{C}_t[f,g]| \, \leq  \, C  \| f \|_{  \mc{W}^{\infty}_{1}(I_{\eps})} \| g \|_{  \mc{W}^{\infty}_{1}(I_{\eps})} \, ,
\enq
with constants that are uniform in $t \in \intff{0}{1}$. These bounds along with $A_t, P_t$ being in the Schwartz class uniformly in $t\in \intff{0}{1}$ allow one
to apply Fubini backwards so as to infer that
\beq
\mathbb{E}_{N,\ga_t}[ \mc{G}_{\e{in}}^m] \, = \, \mathbb{E}_{G;t}[ \msc{G}_{\e{in}}^m]  \;+ \; \e{O}\big( \f{1}{N} \big) \,.
\enq
This yields convergence of the moments uniformly in $t \in \intff{0}{1}$.
Now, given that the random variable $ \mc{G}_{\e{in}}$ converges in moments and that
moments uniquely characterise $\mathbb{E}_{G;t}$, it follows that $\mc{G}_{\e{in}}$ converges in distribution
to $\msc{G}_{\e{in}}$. Since
$\mathbb{E}_{G;t}[ \ex{ \msc{G}_{\e{in}} } ] $ is continuous in $t \in \intff{0}{1}$
and the convergence of moments is uniform in  $t \in \intff{0}{1}$, the
convergence of $\mathbb{E}_{N, \ga_t} \big[ \ex{\mc{G}_{\e{in}}} \big]$ towards $\mathbb{E}_{G;t}[ \ex{ \msc{G}_{\e{in}} } ]$ is also uniform in
$t \in \intff{0}{1}$.

\end{proof}
We next give a more explicit formula for the limit in Proposition \ref{proplim}.
\begin{definition} Let us introduce the spaces
\begin{align}
H &= \left\{ f \text{ measurable} \, : \, \int_\mathbb{R} |f(x)|^2 (1+x^6) \, dx < +\infty \right\} \\
H^\prime &= \left\{ f \text{ measurable} \, : \, \int_\mathbb{R} |f(x)|^2 (1+x^6)^{-1} \, dx < +\infty \right\} \;.
\end{align}
It is clear that these are both Hilbert spaces since they are isometrically isomorphic to $L^2(\mathbb{R})$, and $H^\prime$ may be identified with the dual of $H$ and vice versa.
Furthermore under a Fourier transform they may be identified with Sobolev spaces.  Define the natural pairing
\begin{align*}
&( \cdot, \cdot) : H \times H^\prime \rightarrow \Cx \\
&( f, g) = \Int{\mathbb{R}}{} \overline{f(x)}g(x) \, \dd x \;.
\end{align*}
\end{definition}
Let $B_t : \mathbb{R}^2 \longrightarrow \mathbb{C}$ be defined by $B(x,y) = \mathfrak{C}_t[\ex{-2\pi i x *}, \ex{2\pi \i y *}]$ and $\mathsf{m}_t : \mathbb{R} \longrightarrow \mathbb{C}$
be defined by $\mathsf{m}_t(x) = \mathfrak{m}_t[\ex{-2\pi \i x *}]$. The uniform in $t \in \intff{0}{1}$ upper bounds \eqref{growth} ensure that
$\mathsf{m}_t \in H^\prime$ and $B_t \in H^\prime \otimes H^\prime$ so that $B_t$ may be regarded as the integral kernel of a bounded operator
\begin{align}
\mathsf{B}_t : H \longrightarrow H^\prime
\end{align}
Next, since $A_t$ is rapidly decreasing $A_t \in H \otimes H$ and hence may be regarded as the integral kernel of a bounded operator
\begin{align}
\mathsf{A}_t: H^\prime \longrightarrow H
\end{align}
Thus the product $\mathsf{B}_t \, \mathsf{A}_t : H^\prime \longrightarrow H^\prime$ is bounded.
\begin{proposition}
$\mathsf{B}_t \,  \mathsf{A}_t$ is trace class and hence the Fredholm determinant $\det_{H^\prime}(\mathsf{1}- \i \beta \mathsf{B}_t\, \mathsf{A}_t)$ exists.
Furthermore $|\det_{H^\prime}(\mathsf{1}- \i \beta \mathsf{B}_t\, \mathsf{A}_t)|\geq 1$.
\end{proposition}
\begin{proof}
Let us identify $H$ and $H^\prime$ with $L^2(\mathbb{R})$ by the following (unitary) operators
\beq
\left\{ \ba{c }
\mathsf{S} :  H   \longrightarrow   L^2(\mathbb{R}) \vspace{2mm} \\
    \mathsf{S}[f](x) = \sqrt{1+x^6} f(x) \ea  \right.
\qquad \e{and} \qquad
\left\{ \ba{c } \mathsf{S}^\prime : H^\prime    \longrightarrow   L^2(\mathbb{R}) \vspace{2mm} \\
    \mathsf{S}^\prime [f](x) = \frac{ f(x)  }{\sqrt{1+x^6}} \ea \right.  \;.
\enq
Then let
\beq
\ba{cccc} \breve{\mathsf{A}}_t := & \mathsf{S} \mathsf{A}_t (\mathsf{S}^\prime)^{-1} &: & L^2(\mathbb{R}) \longrightarrow  L^2(\mathbb{R}) \vspace{2mm}  \\
\breve{\mathsf{B}}_t :=  & \mathsf{S}^\prime \mathsf{B}_t \mathsf{S}^{-1} &:& L^2(\mathbb{R}) \longrightarrow  L^2(\mathbb{R}) \ea \; .
\enq
$\breve{\mathsf{A}}_t$ has integral kernel $\sqrt{1+x^6}\sqrt{1+y^6}A_t(x,y)$ and $\breve{\mathsf{B}}_t$ has integral kernel $\frac{1}{\sqrt{1+x^6}\sqrt{1+y^6}}B_t(x,y)$.
Both of these are Hilbert-Schmidt operators uniformly in $t$ and hence the product $\breve{\mathsf{B}}_t \, \breve{\mathsf{A}}_t$ is trace class uniformly in $t$.
However $\breve{\mathsf{B}}_t \, \breve{\mathsf{A}}_t = \mathsf{S}^\prime \mathsf{B}_t \, \mathsf{A}_t (\mathsf{S}^\prime)^{-1}$, hence $\mathsf{B}_t \, \mathsf{A}_t$ is also trace class.

Next, observe that both $\breve{\mathsf{A}}_t$ and $\breve{\mathsf{B}}_t$ are self-adjoint operators on $L^2(\mathbb{R})$. Furthermore $\breve{\mathsf{B}}_t$ is positive-semidefinite since
\beq
\left\langle f, \breve{\mathsf{B}}_t f \right\rangle_{L^2(\mathbb{R})} =
\mathfrak{C}_t\bigg[ \widehat{   \tfrac{ f(*) }{ \sqrt{1+(*)^6} }    } ,  \overline{ \widehat{   \tfrac{ f(*) }{ \sqrt{1+(*)^6} }    } } \bigg]  \, \geq \,  0 \;.
\nonumber
\enq
Here, we have extended $\mathfrak{C}_t[f,g]$ to complex valued functions by bi-linearity, which indeed entails from the positivity of the
variance $\mf{V}_t[f]\geq 0$ that $\mathfrak{C}_t[f,\ov{f}]\geq 0$.

Hence we may define a square root $\sqrt{\breve{\mathsf{B}}}_t$, which is also a self-adjoint operator.  Then by Sylvester's identity
$$\det_{H^\prime}(\mathsf{1}- \i \beta \mathsf{B}_t\mathsf{A}_t)  \, = \, \det_{L^2(\mathbb{R})}\left(\mathsf{1}- \i \beta \sqrt{\breve{\mathsf{B}}}_t \breve{\mathsf{A}}_t \sqrt{\breve{\mathsf{B}}}_t \right) \; .$$
$\sqrt{\breve{\mathsf{B}}}_t\breve{\mathsf{A}}_t\sqrt{\breve{\mathsf{B}}}_t$ is self-adjoint and so has real eigenvalues $\lambda_1, \lambda_2, \dots \in \mathbb{R}$. Then,
$$\left| \det_{H^\prime}(\mathsf{1} - \i \beta \mathsf{B}_t\mathsf{A}_t) \right|^2 = \prod_{j\geq 1} \left|1- \i \beta \lambda_j \right|^2 \geq 1 $$
\end{proof}

The right hand side of \eqref{gaussianexp} has a functional-analytic expression in terms of the operators $\mathsf{A}_t$ and $\mathsf{B}_t$.
Recalling the definition  \eqref{definition msc G in} of   $\msc{G}_{\e{in}}$, we have
\begin{align}\label{expformula}
\mathbb{E}_{G;t}\big[ \ex{ \msc{G}_{\e{in}} } \big]   \, =  \, \frac{1}{\sqrt{\det_{H^\prime}(\mathsf{1}- \i \beta \mathsf{B}_t\mathsf{A}_t)}}
\ex{ \msc{B}_{t}[\mathsf{m}_t, P_t]  } \;,
\end{align}
where
\bem
 \msc{B}_{t}[\mathsf{m}_t, P_t] \, = \, \frac{1}{2} \i \beta \big(\mathsf{m}_t, \mathsf{A}_t \big( \mathsf{1} - \i \beta \mathsf{B}_t\mathsf{A}_t \big)^{-1} \mathsf{m}_t \big)_{ L^{2}(\R) }
 \, + \,  \i \beta \big(P_t, (\mathsf{1}- \i \beta \mathsf{B}_t \mathsf{A}_t )^{-1}\mathsf{m}_t\big)_{L^2(\R)} \\
\, + \,  \frac{1}{2}\beta^2  \big( P_t, (\mathsf{1} - \i \beta \mathsf{B}_t\mathsf{A}_t)^{-1} \mathsf{B}_tP_t\big)_{L^2(\R)} \;.
\end{multline}
This formula is proven in Appendix \ref{fredholmsection}.
\begin{proof}[Proof of Proposition \ref{lowerbound}]

We need to establish a uniform lower bound for \eqref{expformula}. First of all, observe that \eqref{growth}
ensures that  $|B_t(x,y)| \leq C (1+|x|)(1+|y|)$ for some $C > 0$ not depending on $t \in \intff{0}{1}$. Hence,
\beq
\| \breve{\mathsf{B}}_t \|_{\mathrm{HS}}^2 \, = \,  \Int{\mathbb{R}^2}{} \frac{ |B_t(x,y)|^2 }{ (1+x^6) (1+y^6) } \, \dd x \, \dd y \, < \,  C^{\prime}
\nonumber
\enq
uniformly in $t \in \intff{0}{1}$. Further, one has
\beq
\| \breve{\mathsf{A}}_t \|_{\mathrm{HS}}^2 \, = \,  \Int{\mathbb{R}^2}{}  |A_t(x,y)|^2  \underbrace{(1+x^6)}_{ \leq (1+|x|^3)^2 } \underbrace{(1+y^6)}_{ \leq (1+|y|^3)^2 }  \, \dd x \, \dd y
\, \leq \, C \norm{ a_t }^2_{ \mc{W}^{\infty}_{6} ( \intff{-\eps^{\prime} }{ 1+\eps^{\prime} }^2) }
\enq
and this quantity is uniformly bounded in $t$ since $(t,x) \mapsto a_t(x)$ is jointly $\mathcal{C}^\infty$.

Hence, it follows that
\begin{align*}
\Big| \det_{H^\prime}\big( \mathsf{1}- \i \beta \mathsf{B}_t\mathsf{A}_t \big)\Big| \, \leq  \, \ex{ \beta \| \breve{\mathsf{A}}_t \|_{\mathrm{HS}}  \| \breve{\mathsf{B}}_t \|_{\mathrm{HS}} }
\end{align*}
is also bounded uniformly in $t \in \intff{0}{1}$.

It remains to bound $ \msc{B}_{t}[\mathsf{m}_t, P_t]$. One infers from \eqref{growth} and the smoothness and uniform boundedness of $a_t, p_t$ that
\beq
\left\{ \ba{ccc} \big| \mathsf{m}_t(x) \big| & \leq & C  ( 1+x^2) \vspace{2mm}  \\ \big| P_t(x) \big| & \leq & \f{ C }{ 1+x^4  } \ea \right.
\quad \e{and}   \quad
\left\{ \ba{ccc} \big| B_t(x,y) \big|& \leq & C  ( 1+|x|)( 1+|y|)  \vspace{2mm} \\ \big| A_t(x,y) \big| & \leq & \f{ C }{ (1 + x^4) (1 + y^4)  } \ea \right.  \; .
\enq
We shall bound $ \msc{B}_{t}[\mathsf{m}_t, P_t]$ by relying on the Fredholm series representation for $(\mathsf{1}- \i \beta \mathsf{B}_t\mathsf{A}_t)^{-1} = \mathsf{1}- \mathsf{R}_t$ where $\mathsf{R}_t$ has kernel
\begin{align*}
R_t(x,y) =\frac{1}{\det\left(\mathsf{1}- \i \beta \mathsf{B}_t\mathsf{A}_t\right)_{H^\prime}} \sul{n \geq 0}{} \frac{(-\i\beta)^{n+1}}{n!} \Int{\mathbb{R}^n}{}
\det\left( \begin{matrix}
\rho_t(x,y) & \rho_t(x,u_k) \\
\rho_t(u_j,y) & \rho_t(u_j, u_k)
\end{matrix} \right)_{j,k =1, \dots, n} \, \dd\mathbf{u}
\end{align*}
where $\rho_t$ is the kernel of $\mathsf{B}_t\mathsf{A}_t$.
The estimates on $A_t$ and $B_t$ entail that
$$|\rho_t(x,y) | \leq C \f{ (1+|x|) }{ (1+y^4) }\, .$$
Then, Hadamard's inequality allow one to infer that there is  $C^{\prime} > 0$ such that
\begin{align*}
|R_t(x,y)| \, \leq \,  C^\prime \f{ (1+|x|) }{ (1+y^4) } \;.
\end{align*}
Putting together all of these uniform in $t\in \intff{0}{1}$ estimates, one obtains the claim.

\end{proof}

\section{Expansion of the complex partition function}\label{expansionsection}

In this section we prove Theorem \ref{maintheorem}.  Let us recall the contour $\Gamma_{\e{eq}}$ whose existence is proven in Proposition \ref{contourexistence}
and the central piece $\gamma : [-\epsilon,1+\epsilon] \longrightarrow \mathbb{C}$ (defined by (\ref{gammadef})) whose image contains the support of the equilibrium measure.
Let $\Sigma = \gamma([-\epsilon,1+\epsilon])$. Then, it holds that
\beqa
\mc{Z}_{ N,\Ga_{\e{eq}} }[V] & = & \mc{Z}_{N,\Sg}[V] \, + \, \Int{\Ga^N}{}\dd \mathbf{z}   \pl{a<b}{N}(z_a-z_b)^{\be} \pl{a=1}{N} \Big\{  \ex{-N \be V(z_a) } \Big\} \cdot \mathbbm{1}_{( \Sg^N)^{\e{c}} }(\mathbf{z})
\nonumber \\
\mc{Z}_{ N,\Ga_{\e{eq}} }[V] & = & \mc{Z}_{N,\Sg}[V] \, + \, \e{O}\bigg( \op{Z}_{N,\ga}[V] \mathbb{P}_{N,\ga}\Big[ \big\{\exists a \, : \, x_a \in \intff{-\eps}{1+\eps}^{\e{c}} \big\} \Big]  \bigg) \;.
\nonumber
\eeqa
Thus, by Propositions  \ref{MaxEigLDP} and \ref{lowerbound}, there is a $C > 0$ such that
$$ \mc{Z}_{N,\Ga_{\e{eq}} }[V] \, = \, \mc{Z}_{N,\Sg}[V] \Big( 1 +  \e{O}\big( \ex{-CN} \big) \Big) \; . $$
Now, recall  $V_t$ defined by \eqref{interpolatingpotential}-\eqref{interpolatingpotential2}, and the curves $\Sigma_t$ and $\Gamma_0$ that are introduced in Definition \ref{contours}.
Owing to Propositions \ref{lowerbound}, $t \mapsto \mc{Z}_{N,\Sg_t}[V_t]$ does not vanish on $\intff{0}{1}$ provided that $N$ is large enough. It is direct to
see that it is also smooth in $t$. From these observations one may infer that
$$ \ln \mc{Z}_{ N,\Ga_{\e{eq}} }[V] =  \ln \mathcal{Z}_{N,\Gamma_0}[V_0] + \Int{0}{1}  \Dp{t} \ln \mathcal{Z}_{N,\Sigma_t}[V_t] \, \dd t
\, + \,  \e{O} \big( \ex{-CN } \big) \;. $$
Note that here, we have also invoked the large deviation principle for the maximal particle in the real model subordinate to the quadratic potential $V_0$
which yields $\mc{Z}_{N,\Ga_{0} }[V_0] \, = \, \mc{Z}_{N,\Sg_0}[V_0] \Big( 1 +  \e{O}\big( \ex{-CN} \big) \Big)$.

As a matter of fact the $t$-derivative under the integral may be expressed in an alternative way which permits large-$N$ estimation.
Indeed, let
\beq
\big<\!\! \big<  \mc{F} \big>\!\! \big>_{ \bigotimes_{a=1}^{k}\nu_a }^{(t)} \, = \, \mc{E}_{N,\Sg_t}\bigg[ \,  \Int{ \R^N}{} \pl{a=1}{k} \dd \nu_a(s_a)\, \mc{F}(\bs{s}_k) \bigg]
\enq
%
%
in which $ \mc{E}_{N,\Sg_t}$ corresponds to the expectation with respect to the complex model's measure $ \Pi_{N, \Sg_t}$ introduced in \eqref{complexmodel}.
Then, since the contour integral depends on the contour only by its end points, one gets
\beq
\Dp{t} \ln \mathcal{Z}_{N,\Sigma_t}[V_t]   \,  =  \, - \beta N^2 \big<\!\! \big<   \Dp{t} V_t  \big>\!\! \big>_{ L_N^{(\mathbf{z})} }^{(t)} \, +  \hspace{-2mm}
\sul{ \a \in \{-\eps, 1+\eps\} }{}\hspace{-2mm} \e{sgn}(\a) \, \mf{r}_N^t\big( \ga_{t}(\a) \big)
\enq
in which
\beq
 \mf{r}_N^t\big( z \big)   \; = \; \frac{ \ex{-N \beta  V_t(z)} }{ \mathcal{Z}_{N,\Sigma_t}[V_t] } \Int{\Sg_t^{N-1}}{} \pl{ a < b }{ N-1} (z_a - z_b)^\beta  \cdot
 \pl{a=1}{N} \Big\{ \ex{-N \beta  V(z_a)} (z - z_a)^\beta  \Big\} \, \dd\mathbf{z} \;.
\enq
Direct bounds lead to
\beq
\big|  \mf{r}_N^t\big( \ga_t(x) \big)  \big| \, \leq \, \bigg|  \f{ \op{Z}_{N,\ga_t^{\eps}}[V_t]  }{   \mathcal{Z}_{N,\Sigma_t}[V_t]  } \bigg| \cdot
\f{ \varrho_N^t(x)  }{ |\ga_t^{\prime}(x) | } \;.
\enq
where $\rho^t_N$ is the 1-point density introduced in Definition  \ref{1pointdensity}.
It then follows from Lemma \ref{curvelowerbound} combined with Propositions \ref{edgebound} and \ref{lowerbound} that there exists $C>0$ such that, uniformly in $t \in \intff{0}{1}$, one has
\beq
\Dp{t} \ln \mathcal{Z}_{N,\Sigma_t}[V_t]  \,  =  \, - \beta N^2 \big<\!\! \big<   \Dp{t} V_t  \big>\!\! \big>_{ L_N^{(\mathbf{z})} }^{(t)} \, +  \e{O}\big( \ex{-CN} \big) \;.
\enq
The Dyson-Schwinger equations for the complex model allow one to access to the large-$N$ asymptotic expansion of
the linear statistics, in particular of $\big<\!\! \big<   \Dp{t} V_t  \big>\!\! \big>_{ L_N^{(\mathbf{z})} }^{(t)}$.

We thus start by introducing the master operator relevant to the problem. In the following, we shall denote by $\mc{O}(\msc{N})$ the space of holomorphic functions
on some open neighbourhood $\msc{N}$ of $\Sg_t$.
\begin{definition}[ Complex master operator]
Let $ \mu_{\Sg_t} \in \mathcal{P}(\Sigma_t)$ be the equilibrium measure defined by \eqref{tequilmeas} and $V_t$ be the interpolating potential \eqref{interpolatingpotential},\eqref{interpolatingpotential2}.
Then the complex master operator $\De_{\Sg_t}: \mc{O}(\mc{N}) \rightarrow \mc{O}(\mc{N})$ is defined as
\beq
\De_{\Sg_t}[f] (z) \,=\,   V^\prime_t(z) f(z)  \,  - \,   \Int{\Sigma_t}{}  \mc{D}[f](z,w)  \dd \mu_{\Sg_t} (w)  \qquad \e{with} \quad  z \in \Sigma_t
\enq
and where
\beq
\mc{D}[f](z,w) \, = \, \f{ f(z) - f(w) }{ z-w }
\label{definition derivee non commutative}
\enq
corresponds to the non-commutative derivative.
\end{definition}

As shown in Proposition \ref{Proposition inversion complex master operator}, there exists a continuous functional $\mc{K}_{\Sg_t}$ on $\mc{O}(\mc{N})$
with $\mc{K}_{\Sg_t}[1]=1$ such that $\e{Im}\big[ \De_{\Sg_t} \big] = \e{ker}[\mc{K}_{\Sg_t} ]$, and $\De_{\Sg_t}: \mc{O}(\mc{N}) \rightarrow  \e{ker}[\mc{K}_{\Sg_t} ]$
is invertible. In a manner similar to the real master operator, it is convenient to extend the inverse to $\mc{O}(\mc{N}) $ as
\beq
\wt{\De}_{\Sg_t}^{-1}[f] \, = \,  \De_{\Sg_t}^{-1}\big[ f - \mc{K}_{\Sg_t}[f] \big]   \;.
\enq
Since the explicit expression for $ \De_{\Sg_t}^{-1}$ only involves differences of its argument, \textit{c.f.} Proposition \ref{Proposition inversion complex master operator},
it is invariant under translations. We will thus keep denoting the extension by $\De_{\Sg_t}^{-1}$.

Recall that the contours $\gamma_t : \intff{-\epsilon}{1+\epsilon} \rightarrow \mathbb{C}$ can be analytically extended to a neighbourhood $\msc{N}$ of $\intff{-\epsilon}{1+\epsilon}$. Moreover,
this neighbourhood can be taken such that $\mc{N}_t  \, =  \, \ga_t(\msc{N})$ contains some tubular neighbourhood of fixed, $t$-independent, size of $\ga_t( \intff{-\epsilon}{1+\epsilon})$.
Let $\om_0, \om_1, \om_2, \dots$ be a sequence of closed contours, oriented anticlockwise,
enclosing $\intff{-\epsilon}{1+\epsilon}$ in this common neighbourhood such that  $\om_{k}$ lies in the open set between  $\om_{k+1}$ and $\intff{-\epsilon}{1+\epsilon}$.
Next let
\beq
\Omega_k^t \overset{\mathrm{def}}{=} \gamma_t(\om_k) \qquad \e{and} \qquad   \mc{N}_t  \, =  \, \ga_t(\msc{N}) \, .
\label{definition voisinage mathcal N et courbes Omega k t}
\enq

We remark that by compactness and continuity $\mathrm{dist}(\Om_k^t ,\Om_{k+1}^t ) \geq \delta_k > 0$ where $\delta_k$ is independent of $t \in \intff{0}{1}$.  Let
\begin{align*}
\| f \|_{ L^{\infty}(\Omega_k^t) } \overset{\mathrm{def}}{=} \sup_{z \in \Omega_k^t} |f(z)|
\end{align*}
Then by Proposition  \ref{Proposition inversion complex master operator} there exists a constant $C_k > 0$ independent of $t \in \intff{0}{1}$ such that
$$ \| (\De_{\Sg_t})^{-1}[ f ]\|_{ L^{\infty}(\Omega_k^t) } \leq C_k \|  f \|_{ L^{\infty}(\Omega_{k+1}^t) } \, . $$
Of course, by the Cauchy integral formula, we also have $\| f^\prime \|_{ L^{\infty}(\Omega_k^t) } \leq \widetilde{C_k} \|  f \|_{ L^{\infty}(\Omega_{k+1}^t) }$ for some $\widetilde{C_k} > 0$.

Then, the loop equations are expressed in terms of the centred empirical measure
$$ \mc{L}_N^{(\mathbf{z})} \, = \,  L_N^{(\mathbf{z})} - \mu_{\Sg_t}$$
and take the form
\bem
 \big<\!\! \big<    \mc{F}  \big>\!\! \big>_{ \bigotimes\limits^{k+1}  \mc{L}_N^{(\mathbf{z})} }^{(t)}
\, = \, \f{1}{2} \Big<\!\! \Big<  \mc{D}^{(1)} \circ \big(\De_{\Sg_t}^{(1)} \big)^{-1}[\mc{F}] \Big>\!\! \Big>^{(t)}_{\bigotimes\limits^{k+2}\mc{L}_N^{(\mathbf{z})}  }
\, + \, \f{1}{N} \Big( \f{1}{\be}- \f{1}{2} \Big)\Big<\!\! \Big< \Dp{1} \circ \big(\De_{\Sg_t}^{(1)} \big)^{-1}[\mc{F}] \Big>\!\! \Big>^{(t)}_{L_N^{(\mathbf{z})} \bigotimes\limits^{k}  \mc{L}_N^{(\mathbf{z})}  } \\
\, + \, \f{1}{N^2 \be} \sul{a=1}{k}  \Big<\!\! \Big< \chi^{(a+1)}\big[\Dp{a+1}  \big(\De_{\Sg_t}^{(1)} \big)^{-1}[\mc{F}] \big] \Big>\!\! \Big>^{(t)}_{L_N^{(\mathbf{z})} \bigotimes\limits^{k-1}  \mc{L}_N^{(\mathbf{z})}  }
\, + \, \mathbb{B}_{k+1}\Big[  \big(\De_{\Sg_t}^{(1)} \big)^{-1}[\mc{F}] \Big] \;.
\end{multline}
Upon setting $\ov{\mc{E}}_{N, \Sg_t }= \mc{Z}_{N,\Sg_t}[V_t] \cdot  \mc{E}_{N, \Sg_t }$, the boundary term reads
\beq
 \mathbb{B}_{k+1}[ \mc{F} ]  \, = \, \f{-1}{N^2 \be} \sul{ \a \in \Dp{} \Sg_{t} }{} \f{ 1 }{ \mc{Z}_{N,\Sg_t}[V_t]  }
\Dp{\a}  \ov{\mc{E}}_{N, \Sg_t } \Big[ \Int{}{} \mc{F}(u, \bs{s}_k)  \pl{a=1}{k} \dd \mc{L}_N^{(\mathbf{z})}  (s_a) \Big]_{\mid u = \a} \;.
\enq
As in the case of the real model, one may invoke  Proposition \ref{edgebound} to bound  the boundary term as
\beq
 \Big|   \mathbb{B}_{k+1}[ \mc{F} ]   \Big| \, \leq \, C \norm{F}_{ L^{\infty}(\Omega_0^t) } \ex{-C^{\prime} N} \;,
\enq
for some $C, C^{\prime} >0$.

By repeating the reasoning outlined in the case of the Dyson-Schwinger equations for the real model, one shows that there exist continuous functionals
$\Psi_{k;p}^{(t)}$
on $\mc{H}(\msc{N}_t^k)$
with the neighbourhood $\msc{N}_t$ of $\Sg_t$ as introduced in \eqref{definition voisinage mathcal N et courbes Omega k t}, such that
\beq
 \big<\!\! \big<    \mc{F}  \big>\!\! \big>_{ \bigotimes\limits^{k+1}  \mc{L}_N^{(\mathbf{z})} }^{(t)} \, = \,
\sul{p=0}{ \ell } \f{1}{N^{k+p} } \Psi_{k;p}^{(t)}\big[ \mc{F} \big] \; + \; \e{O} \Big(  \f{ 1 }{ N^{\ell +1+k} }  \norm{ \mc{F} }_{ \Om_{a_{k,\ell}}^{t} }  \Big)
\enq
for some sequence $a_{k,\ell}$ independent of $t$ with a remainder that is uniform in $t\in \intff{0}{1}$. Moreover, $t \rightarrow \Psi_{k;p}^{(t)}\big[ \mc{F} \big]$ is
continuous in $t\in \intff{0}{1}$.

In particular, one gets the expansion for the $1$-linear statistic
\bem
 \big<\!\! \big<    f \big>\!\! \big>_{  \mc{L}_N^{(\mathbf{z})} }^{(t)} \, = \,  \f{1}{N} \Big( \frac{1}{\beta}- \frac{1}{2} \Big) \mu_{\Sg_t}\Big(  \Dp{1} \De_{\Sg_t} ^{-1} [f] \Big)
 + \frac{1}{N^2} \Big( \frac{1}{\beta} - \frac{1}{2} \Big)^2  \mu_{\Sg_t}\Big(  \Dp{1} \De_{\Sg_t} ^{-1} \circ \Dp{1} \De_{\Sg_t} ^{-1} [f] \Big) \\
    + \frac{1}{2 N} \Big( \frac{1}{\beta} - \frac{1}{2} \Big)^2
\mu_{\Sg_t}\otimes \mu_{\Sg_t}\Big(  \Dp{1} \big( \De_{\Sg_t}^{(1)}\big)^{-1} \Dp{2} \big( \De_{\Sg_t}^{(2)}\big)^{-1} \Big[ \mc{D}^{(1)} \circ \big(\De_{\Sg_t}^{(1)} \big)^{-1}[f] \Big] \Big) \\
  + \frac{1}{2N^2 \beta} \mu_{\Sg_t}\Big( \chi^{(2)} \circ \Dp{2} \circ\mc{D}^{(1)} \circ \big(\De_{\Sg_t}^{(1)} \big)^{-1}[f] \Big)  \, + \,
  \e{O}\bigg( \f{1}{N^2}  \norm{f }_{ L^{\infty}(\Om_{s}^t)  } \bigg)
\label{ecriture AD statistique 1 lineraires complexes}
\end{multline}
with $s$ large enough.

Surprisingly, if we take $f = \Dp{t}V_t$ both the zeroth and first order terms in the expansion of
$ \big<\!\! \big<    \Dp{t}V_t \big>\!\! \big>_{ L_N^{(\mathbf{z})} }^{(t)}$ may be expressed as exact $t$-derivatives.
\begin{lemma}
Recall the complex energy functional $ \mc{I}_{\Sigma_t}$   defined through equation \eqref{definition fnelle energie complexe} of
Definition \ref{complexenergy}. Then, one has
$$\mu_{\Sg_t}\big( \Dp{t}V_t \big) \, =  \, \frac{1}{2} \frac{ \dd }{\dd t} \mc{I}_{\Sigma_t}\big[\mu_{\Sg_t} \big] \, . $$
\end{lemma}
\begin{proof}
We start by writing the complex energy in terms of the parametrisation $\gamma_t$.
\beq
\mc{I}_{\Sigma_t}\big[\mu_{\Sg_t} \big] \, =  \, \Int{0}{1} \Big\{\frac{1}{2} \big(g_{\Sg_t}[\mu_{\Sg_t} ] \big)_+\big(\gamma_t(x)\big) \, + \, \frac{1}{2}
\big(g_{\Sg_t}[\mu_{\Sg_t} ] \big)_-\big(\gamma_t(x)\big) \, + \, 2 V_t\big(\gamma_t(x)\big) \Big\}  \dd \nu_{\e{sc}}(x) \,.
\nonumber
\enq
Observe that $\dd \nu_{\e{sc}}(x) \, = \, \f{8}{\pi} \sqrt{x(1-x)} \dd x$ has a density holomorphic on $\Cx \setminus \big\{ \R^{-} \cup \intfo{1}{+\infty}\big\}$.
This allows us to consider two contours $\mathcal{C}_+$ and $\mathcal{C}_-$, each going from $0$ to $1$ in the domain where $\gamma_t$ is analytic,
with $\mathcal{C}_+$ lying above the real axis and $\mathcal{C}_-$ lying below the real axis. Hence
\beq
\mc{I}_{\Sigma_t}\big[\mu_{\Sg_t} \big] \, =  \,  \frac{1}{2} \Int{\mathcal{C}_+}{}  g_{\Sg_t}[\mu_{\Sg_t}]\big(\gamma_t(x)\big)\dd \nu_{\e{sc}}(x) \, + \,
\frac{1}{2} \Int{\mathcal{C}_-}{} g_{\Sg_t}[\mu_{\Sg_t}]\big(\gamma_t(x)\big) \dd \nu_{\e{sc}}(x)
 \, + \, 2 \Int{0}{1} V_t\big(\gamma_t(x)\big)  \dd \nu_{\e{sc}}(x) \;.
\nonumber
\enq
 This step allows one to differentiate under the integral sign, which leads to
\bem
\frac{ \dd }{ \dd t } \mc{I}_{\Sigma_t}\big[\mu_{\Sg_t} \big]
\, = \,   \pi \i \Int{\mathcal{C}_+ }{} \mathsf{C}[\mu_{\Sg_t}]\big(\gamma_t(x)\big)  \Dp{t}\gamma_t(x) \dd \nu_{\e{sc}}(x)
  \, +\,   \pi \i \Int{\mathcal{C}_-}{} \mathsf{C}[\mu_{\Sg_t}]\big(\gamma_t(x)\big)  \Dp{t}\gamma_t(x) \dd \nu_{\e{sc}}(x)  \\
 \, + \,  2 \Int{0}{1} V_t^\prime\big(\gamma_t(x)\big)  \Dp{t}\gamma_t(x) \dd \nu_{\e{sc}}(x)  \, + \, 2 \Int{0}{1}  \Dp{t} V_t\big(\gamma_t(x)\big)   \dd \nu_{\e{sc}}(x) \;.
\nonumber
\end{multline}
Upon deforming $\mathcal{C}_{\pm}$ back to $\intff{0}{1}$, one  gets that the sum of the first three terms vanishes by \eqref{EulerLagrange}.
\end{proof}

\begin{lemma}
One has that
$$\frac{ \dd }{ \dd t } \Int{\zeta_1}{\zeta_2} \ln \Big\{ \frac{ \dd \mu_{\Sg_t} (z) }{\dd z }  \Big\} \cdot \dd \mu_{\Sg_t}(z)  \, = \,
\mu_{\Sg_t}\Big( \Dp{1}  \De_{\Sg_t}^{-1} [ \Dp{t}V_t ]  \Big) \, . $$
\end{lemma}

\begin{proof}
Let $\rho_t(z) = \frac{ \dd \mu_{\Sg_t} (z) }{\dd z } $ be the density of $\mu_{\Sg_t}$ with respect to the natural curve measure $\dd z$.
Then taking \eqref{fidentity} and differentiating with respect to $t \in \intff{0}{1}$, we find that for all functions $f  \in \mc{W}^{\infty}_{1}(\Sigma_t)$
it holds
\begin{align*}
\Int{\zeta_1}{\zeta_2} \De_{\Sg_t}[f](z)  \Dp{t}  \rho_t(z)  \, \dd z  \, + \,  \Int{\zeta_1}{\zeta_2}  \Dp{t}V_t^{\prime}(z)  \,  f(z) \,\rho_t(z) \, \dd z = 0\, .
\end{align*}
Replacing $f$ with $ \De_{\Sg_t}^{-1}[f](z)$ and using the symmetry of the covariance we find
\begin{align*}
\Int{\zeta_1}{\zeta_2} f(z)  \Dp{t}  \rho_t(z)\, dz   \, \dd z  \, + \,  \Int{\zeta_1}{\zeta_2}   \, f^\prime(z)   \De_{\Sg_t}^{-1}\big[ \Dp{t}V_t  \big](z)  \, \rho_t(z) \, \dd z = 0\, .
\end{align*}
Integrating by parts
\beq
\Int{\zeta_1}{\zeta_2} f(z)  \Dp{t}  \rho_t(z)\, dz   \, \dd z  \, = \,
\Int{\zeta_1}{\zeta_2} f(z)  \Dp{z}\De_{\Sg_t}^{-1}\big[ \Dp{t}V_t   \big](z) \, \rho_t(z) \, \dd z
      \, + \, \Int{\zeta_1}{\zeta_2} f(z)  \De_{\Sg_t}^{-1}\big[ \Dp{t}V_t \big](z) \, \rho_t^{\prime}(z) \, \dd z\, .
\nonumber
\enq
$\rho_t^\prime(z)$ has a singularity of negative square root type at the endpoints $\zeta_1$ and $\zeta_2$, and so it is integrable.
Since $\ln \rho_t(z)$ has at most $\e{O}\big( \ln |(z-\zeta_1)(z-\zeta_2)| \big)$ singularities at the endpoints and
is smooth otherwise, see Lemma \ref{Lemma representation poiur sqrt de Rt}, we can extend
the above relation to  $f(z) = \ln \rho_t(z)$ by density. Finally, by integrating by parts twice
\beq
\mu_{\Sg_t}\Big(  \Dp{z}\De_{\Sg_t}^{-1}\big[ \Dp{t}V_t   \big]  \Big) \, = \,
\Int{\zeta_1}{\zeta_2} \ln \rho_t(z)   \Dp{z}\De_{\Sg_t}^{-1}\big[ \Dp{t}V_t   \big](z) \rho_t(z) \, \dd z
      \, + \, \Int{\zeta_1}{\zeta_2} \ln \rho_t(z)   \De_{\Sg_t}^{-1}\big[ \Dp{t}V_t   \big](z) \rho_t^\prime(z) \, \dd z \,.
\nonumber
\enq
\end{proof}
Next, let us consider the asymptotic expansion of $\ln \mathcal{Z}_{V_0,\Gamma_0}(N,\beta)$.  Because $V_0$ is quadratic, the Selberg-Mehta formula provides an exact expression for this
partition function, which in turn permits an asymptotic expansion.
\begin{proposition}[Asymptotic expansion of Selberg-Mehta integral]\label{selbergexpansionthm} The partition function for the quadratic potential $V_0$ has the following asymptotic expansion; more precisely, for every $K \in \mathbb{N}$ one has
\begin{align*}
\ln \mathcal{Z}_{N,\Gamma_0}[V_0] = \frac{\beta}{2} N \ln N + \frac{3+\frac{\beta}{2}+\frac{2}{\beta}}{12} \ln N + \sum_{k=-2}^K \frac{F_k(\beta,V_0)}{N^k} + \e{O}(N^{-K-1}), & & N \to +\infty \, .
\end{align*}
In particular
\begin{align*}
F_{-2}(\beta,V_0) &= - \frac{\beta}{2}  \mc{I}_{  \Ga_{0} }[\mu_{ \Ga_{0} }], \\
F_{-1}(\beta,V_0) &= \left(\frac{\beta}{2} - 1\right) \left[ \Int{   \Ga_{0} }{} \ln \Big\{ \frac{\dd \mu_{ \Ga_{0} } }{ \dd z }(z) \Big\}  \, \cdot  \, \dd \mu_{ \Ga_{0} } (z)
    +   \ln \frac{\beta}{2} \right] + \frac{\beta}{2} \ln  \frac{2\pi}{\e{e}}- \ln \Gamma\left(  \frac{\beta}{2}\right).
\end{align*}
\end{proposition}
\begin{proof}
See Appendix \ref{selbergexpansion}.
\end{proof}
Inserting the asymptotic expansion \eqref{ecriture AD statistique 1 lineraires complexes} for $\big<\!\! \big<    \Dp{t}V_t \big>\!\! \big>_{  \mc{L}_N^{(\mathbf{z})} }^{(t)}$ and integrating over $t\in \intff{0}{1}$, noting that all remainders are uniform in that parameter, we arrive at Theorem \ref{maintheorem}.

\appendix
\section{Singular integral operators and their inverses}\label{singularsection}

In this section we show how the various singular integral operators of interest to our analysis may be inverted.  To begin we define various quantities.

\begin{definition}[$\mc{W}^{\infty}_k$ spaces]\label{Ckspaces}
Let $\Sg$ be a smooth curve in $\Cx$, then
$$\mc{W}^{\infty}_k(\Sg) \, = \, \big\{ f : \Sg \rightarrow \Cx \; : \;  f^{(p)} \in L^{\infty}(\Sg), \quad p=0,\dots, k  \big\}$$
in which $L^{\infty}(\Sg)$ is defined with respect to the arc length measure on $\Sg$.  We let
\begin{align}
\| f \|_{  \mc{W}^{\infty}_{k}(\Sg)} \overset{\mathrm{def}}{=} \sul{p=0}{k} \underset{x \in \Sg}{\e{supess}}| f^{(p)}(x)| \;.
\end{align}
Note that we have the Banach algebra-type inequality $\| f g \|_{  \mc{W}^{\infty}_{k}(\Sg)} \leq 2^k \| f \|_{  \mc{W}^{\infty}_{k}(\Sg)} \| g \|_{  \mc{W}^{\infty}_{k}(\Sg)}
$ and monotonicity $\| f \|_{  \mc{W}^{\infty}_{k}(\Sg)} \leq \| f \|_{  \mc{W}^{\infty}_{k+1}(\Sg)}$.
\end{definition}

\subsection{The complex master operator}

Below, we focus on a one-parameter $t\in \intff{0}{1}$ family of compact smooth curves $\Om_t$
which are contained in some open neighbourhood $\wt{\mc{N}}_t$ which contains a tubular neighbourhood of $\Om_t$ having a fixed size, \textit{i.e.} $t$-independent size.
The endpoints of $\Om_t$ are assumed to be fixed with respect to $t\in \intff{0}{1}$
and given by $\zeta_1$ for the origin, respectively to $\Om_t$'s  orientation,  and $\zeta_2$ for the final endpoint, \textit{viz}. $\Dp{}\Om_t=\{ - \zeta_1, \zeta_2\}$
where the signs correspond to the algebraic orientation of the point in the boundary.

In this subsection we focus on the so-called complex master operator. It is defined in terms of a finite  measure
\beq
\nu_t \in \mc{M}(\Om_t) \qquad \e{with} \qquad  \dd \nu_{t}(z) \; = \; r_+(z)\, P_t(z) \cdot \dd z
\quad \e{and} \quad  r(z) \,= \,   \big[ (z-\zeta_1)(z-\zeta_2) \big]_{ \Om_t ; +}^{1/2}  \;.
\enq
Above, the square root is as defined through \eqref{definition racine avec coupure sur une courbe prescrite} and $P_t$ is a one parameter $t$ holomorphic function defined on the open neighbourhood
$\wt{\mc{N}}_t$ of $\Om_t$, vanishing nowhere, and such that
\beq
\norm{ P_t }_{ \wt{\mc{N}}_t } \, = \, \e{sup}\big\{ |P_t(z)| \, : \, z \in \wt{\mc{N}}_t \big\} \, < \, C
\enq
for some $C>0$ uniform in $t$. Finally, we focus on an open neighbourhood $\mc{N}_t$ of $\Om_{t}$ strictly contained, uniformly in $t$, in $\wt{\mc{N}}_t$,
\textit{viz} $\e{d}\big( \ov{\mc{N}}_t , \Dp{} \wt{\mc{N}}_t \big)>c$.
For $z \in  \e{Int}\big( \Om_t\big)$ we set
\beq
W_t(z) \, = \, \Fint{ \Om_t}{} \f{ \dd \nu_t(\om)  }{ z-\om } \, = \,  \Oint{ \Ga(\Om_t) }{}   \f{ r(\om)P_t(\om) }{ -2( z-\om )  } \dd \om \;.
\enq
Here $\Ga(\Om_t)$ is a loop in $\mc{N}_t$ of index $1$ around $\Om_t$. Note that the final equality in the above equation ensures that $W_t\in \mc{O}( \mc{N}_t )$ and that  $\norm{ W_t }_{\mc{N}_t} \, \leq  \, C$ uniformly in $t$ for some $C>0$.

The complex master operator $\De_{\Om_t}: \mc{O}( \mc{N}_t ) \rightarrow \mc{O}( \mc{N}_t )$ corresponds to the following integral operator
\beq
\De_{\Om_t}[f] (z) \,=\,   W^\prime_t(z) f(z)  \,  - \,   \Int{\Om_t}{}  \f{ f(z) - f(w) }{ z-w }  \dd \nu_{t} (w)  \qquad \e{with} \quad  z \in \mc{N}_t \,.
\enq
Before discussing the invertibility of $\De_{\Om_t}$  we need to introduce the following functional on $\mc{O}( \mc{N}_t )$
\beq
\mc{K}_{\Om_t}[g] \, = \, \Int{ \Om_t }{} \f{  g(s) }{ - \i \pi  r_+(s)  } \dd s  \, = \, \Oint{ \Ga(\Om_t) }{}\f{  g(s) }{ 2 \i \pi  r(s)  } \dd s
\enq
with $ \Ga(\Om_t)$ a loop of index $1$ around $\Om_t$ in $\mc{N}_t$.
The functional satisfies $\mc{K}_1[1]=1$.
Moreover, one has that $\norm{ \mc{K}_{\Om_t}[g] }_{ \Om_t } \leq C \norm{ g }_{\Om_t}$, so that $\mc{K}_{\Om_t}$
is continuous on $\mc{O}( \mc{N}_t )$.  In particular, $\e{ker}[\mc{K}_{\Om_t}]$ is a closed subspace of $\mc{O}( \mc{N}_t )$ of codimension $1$.

\begin{proposition}
\label{Proposition inversion complex master operator}
The image of the complex master operator coincides with $\e{ker}[\mc{K}_{\Om_t}]$. This operator is invertible on its image and in particular one has
\beq
\De_{\Om_t}^{-1}[g] (z) \,= \, \f{1}{P_t(z)} \Int{\Om_t}{} \f{ g(s)-g(z) }{r_+(s) (s-z) } \cdot \f{ \dd s }{ \pi^2 } \qquad \mathrm{for} \; \mathrm{all} \quad g \in \e{ker}[\mc{K}_{\Om_t} ] \;.
\label{expression pour complex master inverse}
\enq
Finally, given a sequence of successively enclosed contours $\Ga_{k}(\Om_t) \subset \mc{N} $ with $\e{d}\big(\Ga_{k}(\Om_t), \Ga_{k+1}(\Om_t) \big) > 0$
uniformly in $t$, it holds that
\beq
\norm{ \De_{\Om_t}^{-1}[g]  }_{  \Ga_{k}(\Om_t) } \, \leq \, C_k  \norm{ g  }_{  \Ga_{k+1}(\Om_t) }  \qquad \mathrm{for} \; \mathrm{all} \quad g \in \e{ker}[\mc{K}_1] \;.
\enq
\end{proposition}
The proof follows closely the technique explained in \cite{DeiftOrthPlyAndRandomMatrixRHP}. We stress that even though the above formulae are obtained
through complex analysis techniques, they also allow one to invert $\De_{\Om_t}$ on $\mc{W}_{k}^{\infty}(\Sg_t)$ spaces, without needing to recourse to
analyticity of the measure $\nu_t$'s density.

\begin{proof}

 Assume that $g \in \e{Im}\big[ \De_{\Om_t} \big]$. So, for all $z \in \Om_t$, it holds that
\beq
\Fint{ \Om_t}{} \f{ f(\om)  }{ z - \om } \dd \nu_t(\om) \, = \, g(z)  \;.
\enq
Set
\beq
F(z) \, = \, \f{ 1 }{ r(z) } \Int{ \Om_t }{}  \f{ f(\om) \dd \nu_t(\om) }{ 2(\om-z)  } \;.
\label{definition F}
\enq
 $F \in \mc{O}(\Cx \setminus \Om_t)$ and admits $\pm$ continuous boundary values on $\e{Int}( \Om_t )$ and at most square root singularities at the endpoints of $\Om_t$.
Furthermore, it behaves as $\e{O}(z^{-2})$ as $z \rightarrow \infty$ and its boundary values satisfy $F_{+}(z) - F_-(z)=-g(z)/r_+(z)$ for $z \in \e{Int}(\Om_t)$.
Hence,
\beq
F(z) \, = \, \Int{ \Om_t}{}  \f{  - g(s) }{ r_+(s) (s-z) } \cdot \f{ \dd s }{ 2\i\pi } \;.
\label{ecriture F comme sol RHP}
\enq
On the one hand, explicit calculations based on \eqref{definition F} yield $F_+(z)+F_-(z)=\i\pi f(z) P_t(z)$ for $z \in \e{Int}(\Om_t)$. On the other hand, one infers from
\eqref{ecriture F comme sol RHP} that given $z \in \e{Int}(\Om_t)$,
\beq
F_+(z)+F_-(z)= \Fint{ \Om_t}{} \f{  - g(s) }{ r_+(s) (s-z) } \cdot \f{ \dd s }{  \i\pi } \; = \; \Oint{ \Ga(\Om_t) }{} \f{    g(s) }{ r (s) (s-z) } \cdot \f{ \dd s }{  2 \i\pi }  \;.
\enq
Thus, we get that
\bem
f(z) \, = \, \f{1}{P_t(z)} \Oint{ \Ga(\Om_t) }{} \f{ g(s)  }{ r(s) (s-z) } \cdot \f{ \dd s }{ -2 \pi^2 }
\, = \, \f{1}{P_t(z)} \Oint{ \Ga(\Om_t) }{} \f{ g(s) - g(z)  }{ r(s) (s-z) } \cdot \f{ \dd s }{ -2 \pi^2 } \\
\, = \, \f{1}{P_t(z)} \Int{\Om_t}{} \f{ g(s)-g(z) }{r_+(s) (s-z) } \cdot \f{ \dd s }{ \pi^2 } \;.
\label{ecriture expression preimage de g}
\end{multline}
This entails \eqref{expression pour complex master inverse} for $z \in \e{Int}(\Om_t)$. However, by analytic continuation, this formula extends $f$ holomorphically to $\mc{N}_t$.
Moreover, one may check that this extension does coincide with the formula for $f$ that one can obtain upon trivially inverting $\De_{\Om_t}[f] (z) \,=\, g(z)$
for $z   \in \mc{N}_t\setminus \Om_t $. For that it is enough to observe that for such $z$s, one has 
\beq
  W_t(z) \, - \, \Int{ \Om_t}{} \f{ \dd \nu_t(\om)  }{ z-\om }  \; = \; \i\pi r(z) P_t(z) \;.
\label{ecriture expression terme mutliplicatif sur f op master cplx}
\enq

One should note that for the sake of consistency with the asymptotics of $F(z)$, the leading order term in the $z \rightarrow \infty$ expansion of $F$ as given by \eqref{ecriture F comme sol RHP}
should vanish, which implies that $\mc{K}_{\Om_t}[g]=0$. Thus, $ \e{Im}\big[ \De_{\Om_t} \big]\subset \ker[\mc{K}_{\Om_t} ]  $. Reciprocally, if $g\in \ker[\mc{K}_{\Om_t} ]$,
a direct calculation shows that given $f$ as in \eqref{ecriture expression preimage de g}, it holds $\De_{\Om_t}[f]=g$ on $\mc{N}$ and thus
 $\ker[\mc{K}_{\Om_t} ] \subset \e{Im}\big[ \De_{\Om_t} \big]$. Finally, the continuity bounds for $\De_{\Om_t}^{-1}$ are readily obtained by direct estimation.

\end{proof}

\subsection{The real master operator}

We now turn to considering the \textit{real master operator}  $\Xi_{\eta_t} : \mc{W}^{\infty}_{k}( \mc{J}^{\prime} ) \longrightarrow \mc{W}^{\infty}_{k-1}( \mc{J}^{\prime}  )$.
We need to introduce a few auxiliary objects to define it.

First of all, we consider diffeomorphisms $\eta_t: \mc{J}^{\prime} \rightarrow \eta_t( \mc{J}^{\prime} )$
with $\mc{J}^{\prime}$ being a fixed segment $\intff{a^{\prime}}{b^{\prime}}$ and single out a subinterval $\mc{J}=\intff{a}{b}$
of $\mc{J}^{\prime}$, \textit{viz}. $\mc{J} \subset\mc{J}^{\prime}$. In principle, $\mc{J}$ may coincide with $\mc{J}^{\prime}$
We shall work under the assumption that there exist $t$-independent constants $ C_{\eta} > 0$ and $\{ C_p \}_{ p \geq 0 }$
\beq
\norm{\eta_t}_{  \mc{W}^{\infty}_{k}( \mc{J}^{\prime}   )} \, \leq \, C_p \qquad \e{and} \qquad
C_{\eta} \, \leq \, \Big| \f{ \eta_t(x) - \eta_t(y)}{ x - y}   \Big| \, \leq \, C_{\eta}^{-1} \;\; \forall x,y \in \mc{J}^{\prime} \;.
\enq
If $\mc{J}$ is a strict subinterval of $\mc{J}^{\prime}$, \textit{i.e.} $\mc{J} \subset \e{Int}\big( \mc{J}^{\prime} \big)$,
then we shall assume that $\eta_t$ is analytic on some open, $t$-independent, neighbourhood $\msc{N}$ of $\mc{J}^{\prime}$.

The real master operator is defined in terms of a finite  measure
\beq
\xi_t \in \mc{M}( \mc{J} ) \qquad \e{with} \qquad  \dd \xi_{t}(z) \; = \; \mf{s}_+(x)\, \mf{S}_t(x) \cdot \dd x
\quad \e{and} \quad  \mf{s}(z) \,= \,   (z-a)^{\f{1}{2}}(z-b)^{\f{1}{2}}  \;.
\enq
Above, the square roots are defined in terms of the principal branch of the logarithm so that $\mf{s}$ has cuts across $\intff{a}{b}$.
$\mf{S}_t$ is a holomorphic function in a fixed, $t$-independent open neighbourhood $\mf{N}$ of $\mc{J}^{\prime}$.
It is non vanishing on the latter and $\mf{S}_t(x) \in -\i\R^+$ for $x \in \mc{J}^{\prime}$ and satisfies the bound
$\norm{ \mf{S}_t }_{ L^{\infty}(\mf{N}) }, < \, C$ uniformly in $t\in \intff{0}{1}$.

Finally, we define an auxiliary potential for $x \in \e{Int}(\mc{J})$
\beq
\msc{W}_t\big( \eta_t(x) \big) \,  = \, \Fint{ \mc{J} }{} \f{ \dd \xi_t (y)  }{  \eta_t(x) \, - \, \eta_t(y) }
\, = \,  \Oint{ \Ga(\mc{J}) }{} \f{ \mf{r}(y) \mf{S}_t(y) \dd y   }{ -2\big( \eta_t(x) \, - \, \eta_t(y) \big) }
\enq
The second integral representation provides an analytic continuation of $\msc{W}_t$ to some open neighbourhood
of $\mc{J}$, containing $\mc{J}^{\prime}$ in particular. It makes only sense when $\mc{J}$ is a strict subinterval of $\mc{J}^{\prime}$.

We are finally in position to define the real master operator $\Xi_{\eta_t} \, : \, \mc{W}^{\infty}_k(\mc{J})  \rightarrow   \mc{W}^{\infty}_{k-1}(\mc{J})$:
\beq
\Xi_{\eta_t}[f] (x) = \Re \Big( \msc{W}_t(\eta_t(x)) \eta^\prime_t(x) \Big) f(x)  \, - \,
 \Int{0}{1}  \Re \bigg\{  \frac{\eta^\prime_t(x) f(x) - \eta^\prime_t(y) f(y)}{\eta_t(x) - \eta_t(y)}\bigg\}   \dd \xi_{t}(y) \;.
\label{definition real master operator}
\enq

\begin{proposition}\label{realmasterinverse}

 There exists a continuous linear functional $\mc{K}_{\eta_t}: \mc{W}^{\infty}_k(\mc{J}) \rightarrow \Cx$,
satisfying
\beq
 \mc{K}_{\eta_t} [1] = 1 \qquad and  \qquad
 \big| \mc{K}_{\eta_t} [f] \big| \, \leq \, C\norm{ f }_{ \mc{W}^{\infty}_k(\mc{J}) }
\enq
for some $C>0$ independent of $t \in \intff{0}{1}$,  such that the image of
the real master operator is given by $ \e{Im}\big[ \Xi_{\eta_t} \big] \, = \, \e{ker}\big[\mc{K}_{\eta_t} \big]$.
Moreover, $\Xi_{\eta_t} \, : \, \mc{W}^{\infty}_k(\mc{J})  \rightarrow \e{ker}\big[\mc{K}_{\eta_t} \big]$ is invertible and
there exists $C^{\prime}>0$ and independent of $t \in \intff{0}{1}$ such that
\beq
\norm{ \Xi_{\eta_t}^{-1}[g] }_{ \mc{W}^{\infty}_k(\mc{J})  } \; \leq  \;  C^{\prime} \cdot \norm{ g }_{ \mc{W}^{\infty}_{k-1}(\mc{J})  }
\enq

\end{proposition}

\begin{proof}

It is convenient to recast the operator in a slightly different form that allows one to study its inversion more efficiently.
For that, we observe that for $x \not\in \mc{J}$, $x \in \e{Int}\big(  \Ga(\mc{J}) \big)$, by deforming the integration in the second term from $\mc{J}$
to $\Ga(\mc{J})$
\beq
\msc{W}_t\big( \eta_t(x) \big)  \, - \, \Int{ \mc{J} }{} \f{ \dd \xi_t (y)  }{  \eta_t(x) \, - \, \eta_t(y) } \, = \, \i\pi  \f{ \mf{r}(x) \mf{S}_t(x) }{ \eta_t^{\prime}(x)} \;.
\label{ecriture relation poiur fct E en dehors de J}
\enq
Hence,
\beq
\Xi_{\eta_t}[f] (x) \,  = \, \De_{\mc{J}}[f](x) \, + \, \Int{\mc{J}}{} \tau_t(x,y) f(y) \dd y
\enq
in which
$$\tau_t(x,y)  \, = \, \bigg\{ \Re \Big( \frac{ \eta^\prime_t(y) }{ \eta_t(x) - \eta_t(y) } \, - \,  \frac{ 1 }{ x - y }\Big) \bigg\} \cdot  \mf{r}_{+}(y) \mf{S}_t(y) \, . $$
where $\De_{\mc{J}}$ is the complex master operator studied in the previous sub-section and associated to the $t$-dependent measure $ \xi_t$
on the $t$-independent curve $\mc{J}$.

Observe that if $g \in \e{Im}\big[ \Xi_{\eta_t} \big]$, then
\beq
x \mapsto    g(x) \, - \,   \Int{\mc{J}}{} \tau_t(x,y) f(y) \dd y \in \e{Im}\big[ \De_{\mc{J}} \big]
\enq
and thus this function belongs to $\e{ker}\big[ \mc{K}_{\mc{J}}\big]$. As a consequences, there exists $\wt{g}$ such that
\beq
\De_{\mc{J}}[f](x) \, = \, \wt{g}(x) \, - \, \mc{K}_{\mc{J}}[ \wt{g} ] \, - \,  \Int{\mc{J}}{} \wt{\tau}_t(x,y) f(y) \dd y  \quad \e{with} \quad
\wt{\tau}_t(x,y) \, = \, \tau_t(x,y) \, - \, \mc{K}_{\mc{J}}[ \tau_t(*,y) ] \;.
\enq
This leads us to introduce the integral operator
\beq
\mathsf{T}_t : \mc{W}_k^{\infty}(\mc{J}^{\prime}) \longrightarrow  \ker_{ \mc{W}_k^{\infty}(\mc{J}^{\prime})   }  \big[ \mc{K}_{\mc{J}}  \big]  \,
\quad \e{with} \quad \mathsf{T}_t[f](x) \, = \,  \Int{\mc{J}}{} \wt{\tau}_t(x,y) f(y) \dd y \;.
\enq
Hence, acting with $\De_{\mc{J}}^{-1}$ yields
\beq
\big( \mathsf{1} \,  + \, \mathsf{L})[f]\big) \, = \, \De_{\mc{J}}^{-1}\big[  \wt{g}(x) \, - \, \mc{K}_{\mc{J}}[ \wt{g} ]  \, \big]
\qquad \e{with} \qquad
\mathsf{L}_t \, = \,  \De_{\mc{J}}^{-1}\circ \mathsf{T}_t : \mc{W}_k^{\infty}(\mc{J}^{\prime})  \longrightarrow \mc{C}^{\infty}(\mc{J}^{\prime}) \;.
\enq
Note that  $\mathsf{L}_t[f](x)\, = \, \Int{\mc{J} }{} \la_t(x,y) f(y) \dd y$ has its integral kernel given explicitly by
\beq
\la_t(x,y) \, = \, \f{ \mf{S}_t(y)  \mf{r}_{+}(y) }{ \mf{S}_t(x) } \Int{ \mc{J} }{} \f{ \tau_t(s,y)-\tau_t(x,y)  }{ \mf{r}_+(s) (s-x) } \cdot \f{ \dd s }{ \pi^2 } \, =  \, \check{\la}_t(x,y) r_+(y) \;.
\enq
There, $\check{\la}$ is smooth on $\mc{J}^{\prime} \times \mc{J}^{\prime}$ and its $\mc{W}^{\infty}_k$ norms may all be controlled uniformly in $t\in \intff{0}{1}$
It thus follows by a theorem of Gonzalez-Barrios and Dudley \cite{metricentropy} that the integral operator $\mathsf{L}_t : L^2\big( \mc{J}^{\prime} \big) \rightarrow L^2\big( \mc{J}^{\prime} \big) $
 with integral given by $\la_t$ is trace class.

The last observation allows us to invert $\mathsf{1} + \mathsf{L}_t$ in $L^2\big( \mc{J}^{\prime} \big) $ and then use the resulting formula to define an inverse on $\mc{W}^{\infty}_k\big( \mc{J}^{\prime} \big)$.
Because $\mathsf{L}_t$ is trace class, the Fredholm determinant $\det(\mathsf{1} + \mathsf{L}_t)$ exists
and will vanish only if there is an $f \in L^2\big( \mc{J}^{\prime} \big)  \setminus \{ 0 \}$ such that $(\mathsf{1} + \mathsf{L}_t)[f] = 0$.
Thus, assume such an $f$ exists. Because $\mathsf{L}_t$ sends $L^2\big( \mc{J}^{\prime} \big)$ functions to $\mc{C}^\infty ( \mc{J}^{\prime})$ functions, we have
$f = -\mathsf{L}_tf \in \mc{C}^\infty ( \mc{J}^{\prime})$. At this stage, we may just apply $\De_{\mc{J}}$ to both sides
which leads to  the equation $\Xi_{\eta_t}[f] = C_t$ where $C_t$ is a constant depending on $f$ and possibly on $t\in \intff{0}{1}$.
Integrating the \textit{lhs} of the equation by parts we find that for $x \in \e{Int}(\mc{J})$
\begin{align*}
\Int{\mc{J}}{} \ln|\eta_t(x) - \eta_t(y)|   G_t(y)  \, \dd y = C   \qquad \e{with} \qquad
   G_t(y)  \, =\, \big( f \mf{r}_+\mf{S}_t \big)^{\prime}(y)\, .
\end{align*}
Note that by by continuity, the equation also holds up to the endpoints, \textit{i.e.} for $x \in \mc{J}$.

Integrating both sides with respect to $G_t$ over the interval $\mc{J}$ we find
\begin{align*}
\Int{ (\eta_t(\mc{J}))^2 }{} \hspace{-3mm} \ln|z-w| \, \dd \zeta_t(z) \otimes \dd \zeta_t(w) = 0
\end{align*}
where $\dd \zeta_t$ is the pushforward of the signed measure $G_t(y)\dd y$ on $\mc{J}$ to the curve $\eta_t(\mc{J})$ under the map $\eta_t$.
However, by Proposition \ref{logdist}, this implies $\dd \zeta_t \equiv 0$ and so $ f(y) \mf{r}_+(y) \mf{S}_t(y)$ is constant on $\mc{J}$. Since this function vanishes
on $\Dp{}\mc{J}$, we infer that $f(y) = 0 \, \, \forall y \in \e{Int}(\mc{J})$ and thus on $\mc{J}$ by continuity.
This vanishing then implies that the constant $C_t$ also vanishes.  This, in turn, implies by \eqref{ecriture relation poiur fct E en dehors de J}
that $f(x) \i \pi \mf{r}(x) \mf{S}_t(x)(x) = 0$ for all $x \in \mc{J}^{\prime}\setminus \mc{J}$,  so that $f = 0$ on $\mc{J}$.
Therefore, $\det(\mathsf{1} + \mathsf{L}_t)\not=0$.

The inverse of $\mathsf{1}+\mathsf{L}_t$ may be written as  $\mathsf{1}-\mathsf{R}_t$,
where, according to Fredholm theory, $\mathsf{R}_t$ is an integral operator with kernel $R_t$ given by
\begin{align*}
R_t(x,y) = \frac{1}{ \det\big(\mathsf{1}+\mathsf{L}_t \big) } \sul{n \geq 0}{} \frac{1}{n!} \Int{\mc{J}}{}
\det\left( \begin{matrix}
\lambda(x,y) & \lambda(x,u_k) \\
\lambda(u_j,y) & \lambda(u_j, u_k)
\end{matrix} \right)_{j,k =1, \dots, n} \hspace{-6mm} \dd \mathbf{u} \;.
\end{align*}
The above series converges absolutely and uniformly by Hadamard's inequality.
Moreover, the integral equation satisfied by the resolvent kernel ensures that the latter is of the form $R_{t}(x,y) \, = \, \check{R}_{t}(x,y) \mf{r}_+(y)$
with $\check{R}_{t}$ smooth on $\mc{J}^{\prime}\times \mc{J}^{\prime}$, this uniformly in $t \in \intff{0}{1}$.
The above implies that $\| \mathsf{R}_t f \|_{ \mc{W}^{\infty}_k(\mc{J}^{\prime})} \, \leq \, C_j \| f \|_{\mc{W}^{\infty}_0(\mc{J}^{\prime})}$
for some constant $C_j>0$. Hence, we have just proven that any $g \in \e{Im}\big(\Xi_{\eta_t}\big)$ is of the form
$g  = \wt{g} \, - \, \mc{K}_{\eta_t} [\wt{g} ]$ with
\beq
\mc{K}_{\eta_t}[g] \overset{\mathrm{def}}{=} \mc{K}_{\mc{J}}[g] \,  -\,  \mc{K}_{\mc{J}}\Big[ \Int{a}{b} \tau_t(\cdot ,y)
\big\{ (1- \mathsf{R}_t) \circ \De_{\mc{J}}^{-1}\big[ g - \mc{K}_{\mc{J}}[g] \big]\big\}(y) \, \dd y \Big]
\label{definition form lineaire K3}
\enq
and  the solution $f$ to $\Xi_{\eta_t}[f]=g$ with $g \in \e{Im}\big(\Xi_{\eta_t}\big)$
reads
\beq
f \, = \,  (\mathsf{1}- \mathsf{R}_t) \circ \De_{\mc{J}}^{-1}\big[ g - \mc{K}_{\mc{J}}[g] \big] \, .
\label{ecriture explicite inverse real master operator}
\enq
 On then readily infers the continuity bound.

\end{proof}

\newpage

\section{Construction of the contour $\Gamma_{\e{eq}}$ of  Proposition \ref{contourexistence}}\label{curvecontruction}

In this section we prove Proposition \ref{contourexistence}. We shall start by considering the contour provided given by Theorem \ref{existence} of da Silva and Kuijlaars and
build on the hypothesis of one-cut regularity, \textit{c.f.} Definition \ref{regular}, so as to "bootstrap" it a more regular curve.
We shall proceed in two stages: we first construct a piecewise analytic curve that has all the properties we want, save for being $\mc{C}^\infty$.
Then we will "round the corners" to achieve a $\mc{C}^\infty$ curve.
\begin{lemma}[p. 7 of \cite{lirias1643667}]\label{selfintersection}
The support of the equilibrium measure of a contour solving the max-min energy problem (\ref{maxmin}) has no self-intersection points.
\end{lemma}
\begin{proof}
We proceed by contradiction. Let $\Gamma_{\e{eq}} \in \mathcal{T}(\alpha,\alpha^\prime)$ be a curve solving the max-min energy problem, let $\mu_{\Ga_{\e{eq}}}$ be its equilibrium measure,
and assume $\mathrm{supp} \, \mu_{\Ga_{\e{eq}}}$ has a closed loop $\msc{C}$.  Let $\Gamma_s = \overline{ \Gamma_{\e{eq}} \setminus \msc{C} } \in \mathcal{T}(\alpha,\alpha^\prime)$ be the curve obtained
from $\Gamma_{\e{eq}}$ by removing $\msc{C}$, and let $\mu_{\Ga_s}$ be the equilibrium measure on $\Gamma_s$. Since $\mu_{\Ga_{\e{eq}}}\big[ \msc{C} \big] > 0$,
we have $\mu_{\Ga_s} \neq \mu_{\Ga_{\e{eq}}}$. Since $\mathrm{supp}\, \mu_{\Ga_s} \subset \Gamma_{\e{eq}}$,
by uniqueness of the equilibrium measure ensured by Lemma \ref{goodrate}, $\op{I}_{\Ga_{\e{eq}}}[\mu_{\Ga_s}] > \op{I}_{\Ga_{\e{eq}}}[\mu_{\Ga_{\e{eq}}} ]$.
But this contradicts $\Gamma_{\e{eq}}$ being a solution of the max-min energy problem \eqref{maxmin}.
\end{proof}
We first start by defining the admissible contour $\Ga_{\e{eq}}$ of our interest in the vicinity of the support of $\mu_{\Ga_{\e{eq}}}$
by means of the curve $\gamma$ \eqref{gammadef} analytically extended   to $\intff{-\epsilon}{1+\epsilon}$, \textit{c.f.} Lemma \ref{Lemma extension analytique courve gamma}.
We remind that $\eps$ is chosen so that there also exists an $\epsilon^\prime > \epsilon > 0$ such that $\gamma$ admits an analytic extension to
$\intff{-\epsilon^\prime}{1+\epsilon^\prime}$. By \eqref{effectivepositive}, $\varphi_{\Ga_{\e{eq}};\mathrm{eff}}(z) > 0$
for all $z \in \gamma(\intff{-\epsilon^\prime}{1+\epsilon^\prime} \setminus \intff{0}{1})$.
Hence in particular $\varphi_{\Ga_{\e{eq}};\mathrm{eff}}(\gamma(-\epsilon^\prime)) = \varphi_{\Ga_{\e{eq}};\mathrm{eff}}(\gamma(1+\epsilon^\prime)) > 0$.
Here, we stress that this last equality follows from the very specific way in which $\ga$ was built in the neighbourhood of $\intff{0}{1}$.

\vspace{2mm}

Recall that $\ex{\frac{2\pi \i \alpha}{\kappa}}$ and $\ex{\frac{2\pi \i \alpha^\prime}{\kappa}}$ are respectively the incoming and outgoing directions of the contours
in $\mathcal{T}(\alpha,\alpha^\prime)$ (recall Definition \ref{admissible}).
Consider the (oriented) rays
$$ L_{\alpha} \, = \,  \ex{\frac{2\pi \i \alpha}{\kappa}} \intof{+\infty}{R} \qquad \e{and} \qquad  L_{\alpha^\prime} \, =\,  \ex{ \frac{2\pi \i \alpha^\prime}{\kappa}} \intfo{R}{+\infty}\;, $$
where $R > 0$ is chosen sufficiently large so that $\varphi_{\Ga_{\e{eq}};\mathrm{eff}}(z) > 0$ for all $z \in L_{\alpha} \cup L_{\alpha^\prime}$. This is possible because
$\varphi_{\Ga_{\e{eq}};\mathrm{eff}}(z) \sim \frac{1}{\kappa} |z|^\kappa$ as $|z| \to \infty$ along these rays.
\begin{proposition}\label{samecomponent}
Let
$$\Lambda = \{ z \in \mathbb{C} \, : \, \varphi_{\Ga_{\e{eq}};\mathrm{eff}}(z) > 0 \}\, .$$
If the potential $V$ is one-cut regular then the points
$R \, \ex{\frac{2\pi \i \alpha}{\kappa}}$ and $\gamma(-\epsilon^\prime)$ are connected by a continuous path contained in $\Lambda$,
and similarly $R \, \ex{\frac{2\pi \i \alpha^\prime}{\kappa}}$ and $\gamma(1+\epsilon^\prime)$ are connected by a continuous path contained in $\Lambda$.
\end{proposition}
Before proving this we need to establish some preliminary results.
\begin{lemma}\label{sets}
\begin{align*}
\big\{  z \in \mathbb{C} \, : \, \vp_{\Ga_{\e{eq}};\mathrm{eff}}(z) \geq 0  \big\}  \, =  \, \mathrm{supp} \, \mu_{\Ga_{\e{eq}}} \cup  \overline{\Lambda}
\end{align*}
\end{lemma}
\begin{proof}
It is clear that \textit{lhs} $\supset$ \textit{rhs}, thus we need only show that \textit{lhs} $\subset$ \textit{rhs}. Let $z \in \mathbb{C}$ be a point such that
$\vp_{\Ga_{\e{eq}};\mathrm{eff}}(z) \geq 0$. Clearly if $\vp_{\Ga_{\e{eq}};\mathrm{eff}}(z) >0$ or if $z \in \mathrm{supp} \, \mu_{\Ga_{\e{eq}}}$
then $z \in \text{\textit{rhs}}$ so without loss of generality assume $\vp_{\Ga_{\e{eq}};\mathrm{eff}}(z) = 0$ and $z \not\in \mathrm{supp} \, \mu_{\Ga_{\e{eq}}}$.
However if $z \not\in \mathrm{supp} \, \mu_{\Ga_{\e{eq}}}$ then $\vp_{\Ga_{\e{eq}};\mathrm{eff}}$ is harmonic on a neighbourhood of $z$. Since harmonic functions
are real parts of holomorphic functions, on every such neighbourhood
there are some points for which $\vp_{\Ga_{\e{eq}};\mathrm{eff}}(z) \geq 0$. Hence $z \in \overline{\Lambda}$.
\end{proof}
\begin{lemma}\label{finitecomponents} The set $\Lambda$ has only finitely many connected components. In particular, it has at most $\kappa$ connected components.
\end{lemma}
\begin{proof}
We first show that the connected components are unbounded. Suppose by way of contradiction that a connected component $C$ of $\varphi_{\Ga_{\e{eq}};\mathrm{eff}}^{-1}(]0,+\infty[)$ is bounded.
Then $\varphi_{\Ga_{\e{eq}};\mathrm{eff}}$ is maximised by some point in $\overline{C}$. However $\varphi_{\Ga_{\e{eq}};\mathrm{eff}} = 0$ at the boundary and $\varphi_{\Ga_{\e{eq}};\mathrm{eff}} > 0$
in the interior, and so $\varphi_{\Ga_{\e{eq}};\mathrm{eff}}$ has a local maximum, which is impossible because $\varphi_{\Ga_{\e{eq}};\mathrm{eff}}$ is harmonic away from
$\mathrm{supp}\, \mu_{\Ga_{\e{eq}}}$.  Thus the connected components must be unbounded. However for $R>0$ sufficiently large $\Lambda \setminus \op{D}_{0,R}$ has $\kappa$ connected components, where $\op{D}_{0,R} = \{ z \in \mathbb{C} \, : \, |z| < R \}$ (see Lemma 5.5 of \cite{lirias1643667}).

\end{proof}

\begin{lemma}\label{critical1}
Let $C$ and $C^\prime$ be two distinct connected components of $\Lambda$. If
$z \in ( \overline{C} \cap \overline{C^\prime} ) \setminus \mathrm{supp}\, \mu_{\Ga_{\e{eq}}} $ then $R(z) = 0$, where $R$ is the polynomial introduced in Theorem \ref{existence}.
\end{lemma}
\begin{proof}
Firstly, we must have $\varphi_{\Ga_{\e{eq}}; \mathrm{eff}}(z) = 0$, since if $\varphi_{\Ga_{\e{eq}}; \mathrm{eff}}(z) > 0$ then $z$ would be in one of the connected components of $\Lambda$.
But then because these connected components are open, $z$ is bounded away from all other connected components and so could not be in the closure of more than one component.

Secondly,  on every neighbourhood $U$ of $z$,  $U \cap \Lambda$ is disconnected.  We claim this implies $R(z) = 0$. Suppose by way of contradiction that $R(z) \neq 0$.
Because $z \not\in \mathrm{supp}\, \mu_{\Ga_{\e{eq}}}$ the function
$$ \sqrt{R(z)} \, = \, 2\pi \i \, \mathsf{C}\big[ \mu_{\Ga_{\e{eq}}} \big] (z) \, + \,  V^\prime(z)$$
is analytic at $z$. Thus, so is also the map
$$ w \mapsto \Phi(w) \, = \,  \Int{z}{w}  \sqrt{R(\zeta)} \, \dd \zeta $$
and since  $\Phi^\prime(z) = \sqrt{R(z)} \neq 0$, it defines a biholomorphism from a neighbourhood $U_z$of $z$ to a neighbourhood of $0$.
Let $\delta > 0$ be sufficiently small so that $\op{D}_{0,\delta}$ is a subset of the image of $\Phi(U_z)$. The set $\op{D}_{0,\delta} \cap \{ \widetilde{z} \in \mathbb{C}\, : \, \Re \widetilde{z} > 0\}$
is connected, hence using that $\Re \Phi = \varphi_{\mathrm{eff}}$, one gets that
$$\Phi^{-1}\big( \op{D}_{0,\delta} \cap \{ \widetilde{z} \in \mathbb{C}\, : \, \Re \widetilde{z} > 0\} \big) \, = \, \Phi^{-1}\big( \op{D}_{0,\delta}\big)  \cap \Lambda$$
is also connected.  However this contradicts the fact that $U \cap \Lambda$ is disconnected for every neighbourhood $U$ of $z$.
\end{proof}
\begin{lemma}
Assume that the potential $V$ is one-cut regular and let $\zeta_1, \zeta_2 \in \mathbb{C}$ be the endpoints of the support of the equilibrium measure.
Then there is exactly one connected component $C$ of $\Lambda$ such that $\zeta_1 \in \overline{C}$; similarly there is exactly
one connected component $C^\prime$ of $\Lambda$ such that $\zeta_2 \in \overline{C^\prime}$.
\end{lemma}
\begin{proof}
This follows from the fact that the  polynomial $R(z) = (2\pi \i \, \mathsf{C}\big[ \mu_{\Ga_{\e{eq}}} \big] (z) \, + \,  V^\prime(z))^2$ has simple zeros at  $z = \zeta_1$ and $z = \zeta_2$, which implies a local structure of $\varphi_{\Ga_{\e{eq}}; \mathrm{eff}}$ around $\zeta_1$ and $\zeta_2$.
\end{proof}

\begin{lemma}
 \label{Lemme vanishing of R at intersection CC and sup mu eq}

 Let $C$ be a connected component of $\La$ and let $z \in \ov{C}\cap \mathrm{supp} \, \mu_{\Gamma_{\e{eq}}}$. Then $R(z)=0$.

\end{lemma}

\begin{proof}
We proceed by contradiction. Assume that $z \in \overline{C} \cap \mathrm{supp} \, \mu_{\Gamma_{\e{eq}}}$ but $R(z) \not= 0$.
Then, one may continue holomorphically  $\sqrt{R(w)}_{\pm}$ to the $\mp$ side of $\mu_{\Gamma_{\e{eq}}}$, and thus
$$ \big[\Phi_{\Ga_{\e{eq};\e{eff}}}\big]_{\pm}(z) \, = \,  \Int{\zeta_1}{z}  \sqrt{R(\zeta)}_{\pm} \, \dd \zeta $$
admit a holomorphic extension to an open neighbourhood $U_z$ of $z$. Given that $\big[\Phi_{\Ga_{\e{eq};\e{eff}}}\big]_{\pm}^{\prime}(z)= \sqrt{R(z)}_{\pm} \not=0$,
$\big[\Phi_{\Ga_{\e{eq};\e{eff}}}\big]_{\pm}(z)$ extend, up to shrinking $U_z$, to biholomorphisms on $U_{z}$.
If $z=\ga(t)$ with $t \in \intoo{0}{1}$, then upon taking $\om=\ga(t) \pm \i\eps \ga^{\prime}(t)$ with $\eps >0$ and small enough, one has
$$
\Phi_{\Ga_{\e{eq};\e{eff}}}(\om) \; = \; \big[ \Phi_{\Ga_{\e{eq};\e{eff}}} \big]_{\pm}\big( \ga(t) \big) \,  \pm  \,\i \eps \underbrace{ \ga^{\prime}(t)\sqrt{ R( \ga(t) ) }_{\pm} }_{ \in \i \R^+}
\, + \, \e{O}\big( \eps^2 \big) \;.
$$
Thus, upon shrinking $U_z$ if need be, one has that
\beq
\varphi_{\Gamma_{\e{eq}};\mathrm{eff}}(\om) \, = \, \Re\big(  \Phi_{\Ga_{\e{eq};\e{eff}}}(\om) \big)
\left\{  \ba{cc}
<0 & \;\e{on} \; U_z \cap \{ + \e{side}  \, \e{of}\; \ga(\intff{0}{1})  \}   \vspace{2mm} \\
 <0 & \; \e{on} \; U_z \cap \{ - \e{side}  \, \e{of}\; \ga(\intff{0}{1})  \}     \ea \right.  \;.
\enq
Hence, $z$ cannot be a boundary point of $C_i$ as there exist an open neighbourhood $U_z$ thereof such that
\beq
U_z \subset \ga(\intff{0}{1}) \cup \Big\{ z \in \Cx \, : \, \varphi_{\Gamma_{\e{eq}};\mathrm{eff}}(z) <0 \Big\} \;.
\nonumber
\enq

\end{proof}

\begin{proof}[Proof of Proposition \ref{samecomponent}]
An open connected set in $\mathbb{C}$ is also path connected,
so we need only show that $\gamma(-\epsilon^\prime)$ and $R \ex{ \frac{2\pi \i \alpha}{\kappa} }$ belong to the same connected component of $\Lambda$;
and similarly that $\gamma( 1 + \epsilon^\prime )$ and $R \ex{ \frac{2\pi \i \alpha^\prime}{\kappa}} $ belong to the same connected component.
\vspace{2mm}

By the assumption of one-cut regularity there exists a contour $\widetilde{\Gamma} \in \mathcal{T}(\alpha,\alpha^\prime)$ solving the max-min energy problem and for
which $R$ has simple zeros at  $\zeta_1,\zeta_2 \in \widetilde{\Gamma}$, the endpoints of the support of the equilibrium measure, and no other zeros on $\widetilde{\Gamma}$.  Furthermore, because of the way
$\varphi_{\widetilde{\Gamma};\mathrm{eff}}(z)=\varphi_{\Gamma_{\e{eq}};\mathrm{eff}}(z)$ scales as $|z| \to \infty$ (recall that the effective potential only depends on the support of the equilibrium measure),
we can without loss of generality assume that $L_\alpha , L_{\alpha^\prime} \subset \widetilde{\Gamma}$.  Since it solves the max-min energy problem,
we must have $\widetilde{\Gamma} \subset \{ z \in \mathbb{C} \, : \, \varphi_{\widetilde{\Gamma};\mathrm{eff}}(z) \geq 0 \} = \mathrm{supp} \, \mu_{\Gamma_{\e{eq}}} \cup  \overline{\Lambda}$
(by Lemma \ref{EulerLagrange} and Lemma \ref{sets}).

\vspace{2mm}
Let $C_1, \dots, C_m$ be the connected components of $\Lambda$ such that $\overline{C_i} \cap \widetilde{\Gamma} \neq \emptyset$ for $i=1,\dots, m$. Hence
$\widetilde{\Gamma} \subset \overline{C_1} \cup \dots \cup \overline{C_m} \cup \mathrm{supp} \, \mu_{\Gamma_{\e{eq}}}$.
If $\overline{C_i} \cap \mathrm{supp} \, \mu_{\Gamma_{\e{eq}}} \neq \emptyset$ then,
by Lemma \ref{Lemme vanishing of R at intersection CC and sup mu eq}, for all $z \in \overline{C_i} \cap \mathrm{supp} \, \mu_{\Gamma_{\e{eq}}}$, $R(z) = 0$.
Hence at most two such connected components, when closed, insect nontrivially with $\mathrm{supp} \, \mu_{\Gamma_{\e{eq}}}$. Let us call these components $C_i$ and $C_j$
(we will later see that $C_i \neq C_j$ but so far we haven't shown this).

\vspace{2mm}

We now claim that $m \leq 2$, i.e.  no more than two connected components, when closed, intersect non-trivially with $\widetilde{\Gamma}$.  To see this consider the closed set $$K =\widetilde{\Gamma} \cap \bigcup_{\substack{p=1 \\ p \neq i,j}}^m \overline{C_p}$$ $K \cap \mathrm{supp}\, \mu_{\Gamma_{\e{eq}}} = \emptyset$ by construction.
We also claim $K \cap \overline{C_i} = \emptyset$ and $K \cap \overline{C_j} = \emptyset$. This is because any point $z \in K \cap \overline{C_i}$, by Lemma \ref{critical1}, is a zero of $R$.
But the only zeros of $R$ on $\widetilde{\Gamma}$ are the endpoints
$\zeta_1, \zeta_2 \in \mathrm{supp}\, \mu_{\Gamma_{\e{eq}}}$ by the one-cut regularity assumption.
Hence $K$ is disjoint from $\overline{C_i} \cup  \overline{C_j} \cup \mathrm{supp}\, \mu_{\Gamma_{\e{eq}}}$. However $K$ and
$\widetilde{\Gamma} \cap \left( \overline{C_i} \cup  \overline{C_j} \cup \mathrm{supp}\, \mu_{\Gamma_{\e{eq}}} \right)$ are a pair of disjoint closed sets whose union is $\widetilde{\Gamma}$.
Since $\widetilde{\Gamma}$ is connected, $K = \emptyset$.

\vspace{2mm}

Next we claim that $C_i \neq C_j$. By construction $\zeta_1 \in \overline{C_i}$ and $\zeta_2 \in \overline{C_j}$ (or the other way round).
Thus $\gamma(-\epsilon^\prime) \in C_i$ and $\gamma(1+\epsilon^\prime) \in C_j$, since $\gamma(\intff{\epsilon^\prime}{1+\epsilon^\prime} \setminus \intff{0}{1})) \subset \Lambda$
by (\ref{effectivepositive}). Now suppose by way of contradiction $C_i = C_j$. Then there is a continuous path from $\gamma(-\epsilon^\prime)$ to $\gamma(1+\epsilon^\prime)$ contained entirely in $\Lambda$.
Using $\gamma$ to define a path from $\gamma(-\epsilon^\prime)$ to $\gamma(1+\epsilon^\prime)$, we thus have constructed a continuous path, without self-intersections,
lying in $\{ z \in \mathbb{C} \, : \, \varphi_{\Ga_{\e{eq}};\mathrm{eff}}(z) \geq 0 \}$. This path passes through $\mathrm{supp} \, \mu_{\Ga_{\e{eq}}}$, where $\varphi_{\Ga_{\e{eq}};\mathrm{eff}} < 0$ on both sides.
Hence this path encloses a connected component of $\{ z \in \mathbb{C} \, : \,\varphi_{\Ga_{\e{eq}};\mathrm{eff}}(z) < 0\}$.  However by the same arguments used in the proof of Lemma \ref{finitecomponents},
the connected components of $\{ z \in \mathbb{C} \, : \,\varphi_{\Ga_{\e{eq}};\mathrm{eff}}(z) < 0\}$ are unbounded, and so they cannot be any enclosed by any compact loop.
A similar argument can also be used to show that $\overline{C_i} \cap \overline{C_j} = \emptyset$.

\vspace{2mm}

Finally, we conclude the proof by showing that $R\ex{\frac{2\pi \i \alpha}{\kappa}}$ and $R \ex{\frac{2\pi \i \alpha^\prime}{\kappa}}$ belong to different connected components,
\textit{i.e.} $R \ex{\frac{2\pi \i \alpha}{\kappa}} \in C_i$ and $R\ex{\frac{2\pi \i \alpha^\prime}{\kappa}} \in C_j$ (or the other way round).
Let us think of $\widetilde{\Gamma}$ as parametrised by a piecewise $C^1$ function $\widetilde{\gamma} : \mathbb{R} \longrightarrow \mathbb{C}$.
Let $x_1$ be the smallest real number such that $\widetilde{\gamma}(x_1) = \zeta_1$ and $x_2$ the largest real number such that $\widetilde{\gamma}(x_2) = \zeta_2$,
where the path is oriented such that $x_1 < x_2$.  If $R \ex{\frac{2\pi \i \alpha^\prime}{\kappa}} \in C_i$ then $\gamma$ defines a path from $\zeta_2$ to
$R\ex{\frac{2\pi \i \alpha^\prime}{\kappa}}$, and so for some $x_3 > x_2$, $\gamma(x_3) = \zeta_1$. However by the same arguments used in the proof of Lemma \ref{selfintersection},
the energy of the curve can be increased by removing $\widetilde{\gamma}([x_1, x_3])$ and so $\widetilde{\gamma}$ cannot be a solution of the max-min energy problem.
\end{proof}
\begin{remark} In the above proof we are regularly using the basic topological principle that if $A$ and $B$ are closed sets and $\gamma : \intff{0}{1} \longrightarrow A \cup B$ a continuous path,
where $\gamma(0) \in A$ and $\gamma(0) \in B$, then $\gamma(\intff{0}{1}) \cap A \cap B \neq \emptyset$. This follows because if $\gamma(\intff{0}{1}) \cap A$ and $\gamma(\intff{0}{1}) \cap B$
were disjoint they would define a partition of $\gamma(\intff{0}{1})$ and hence show that $\gamma(\intff{0}{1})$ is disconnected, but $\gamma(\intff{0}{1})$ is connected by the continuity of $\gamma$.
\end{remark}

\begin{lemma}[Continuous paths may be replaced by piecewise linear paths]\label{straightlines}
Let $U \subset \mathbb{C}$ be an open, connected (and hence path connected) set. Then given any two points $a,b \in U$, there is a path from $a$ to $b$ contained in $U$
consisting of a finite sequence of straight line segments.
\end{lemma}
\begin{proof}
Because $U$ is path connected, there is a continuous path $\eta : \intff{0}{1} \longrightarrow U$ such that $\eta(0) = a$ and $\eta(1)=b$. Let $\chi = \eta(\intff{0}{1})$,
which is clearly a compact set.  Next, there exists an $\epsilon > 0$ such that $\mathrm{dist}(\chi, \mathbb{C}\setminus U) > 2 \epsilon$.
Hence $\big\{ \op{D}_{\epsilon,z} \big\}_{z \in \chi}$ is an open cover of $\chi$. By compactness we may pass to a finite subcover labelled by
$z_1, \dots, z_n \in \chi$, where without loss of generality $z_1 = a$ and $z_n = b$.

Next, for any $|z_k - z_j| < 2\epsilon$, $k \neq j$, let us draw a straight line between these two points, and note that this straight line is entirely contained in
$\op{D}_{\epsilon,z_j} \cup \op{D}_{\epsilon,z_k} \subset U$.  Let $\widetilde{\chi}$ be the union of the set of such points $\{ z_1, \dots, z_n\}$ and the line segments connecting them.
We claim that $\widetilde{\chi}$ is a connected set, which if so would complete the proof.  Hence by way of contradiction let $A$ and $\widetilde{A}$ be two nonempty sets such that
$A \cup \widetilde{A}= \{ z_1, \dots, z_n\}$, $A \cap \widetilde{A} = \emptyset$, and such that $|z - \widetilde{z}| \geq 2\epsilon$ for all $z \in A$ and $\widetilde{z} \in \widetilde{A}$.  But then
\beq
U \, =  \, \bigcup_{z \in A} \op{D}_{\epsilon,z} \quad \e{and} \quad
\widetilde{U} \, =  \, \bigcup_{ z \in \widetilde{A} } \op{D}_{\epsilon,z}
\nonumber
\enq
are disjoint open sets which together form an open cover of $\chi$.  If both these sets were nonempty $\chi \cap U$ and $\chi \cap \widetilde{U}$ would be a partition of $\chi$.
However $\chi$ is connected and so either $A$ or $\widetilde{A}$ is empty.
\end{proof}
Hence we can construct a contour $\wt{\Gamma}$ consisting of $L_\alpha$, $L_{\alpha^\prime}$,  $\gamma(\intff{-\epsilon^\prime}{1+\epsilon^\prime})$
and a finite number of straight lines connecting $R \ex{\frac{2\pi \i \alpha}{\kappa} }$ to $\gamma(-\epsilon^\prime)$ and $R \ex{\frac{2\pi \i \alpha^\prime}{\kappa} }$ to $\gamma(1+\epsilon^\prime)$,
where these straight lines lie in $\Lambda$.  Hence $\wt{\Gamma} \in \mathcal{T}(\alpha,\alpha^\prime)$,  $\wt{\Gamma} \setminus \mathrm{supp} \, \mu_{\Ga_{\e{eq}}} \subset \Lambda$
and $\gamma(\intff{-\epsilon}{1+\epsilon}) \subset \wt{\Gamma}$.

\vspace{2mm}

We almost have the desired contour except for the "corners" where each analytic arc meets the next.
To fix this we "round" the corners in such a way to achieve a globally $\mc{C}^\infty$ curve.
For example, if we have two analytic arcs $\eta_1 : \intff{-1}{\delta} \longrightarrow \mathbb{C}$ and $\eta_2 : \intff{-\delta}{1} \longrightarrow \mathbb{C}$
(for some $\delta >0$),  such that $\eta_1(0) = \eta_2(0)$, we may join these together smoothly by $\widetilde{\eta} : \intff{-1}{1}  \longrightarrow \mathbb{C}$ where
$$\widetilde{\eta}(x) \,  = \,  \eta_1(x) +( \eta_2(x) - \eta_1(x) ) \, \th\big( \frac{x}{\delta}\big)$$
where $\th : \mathbb{R}\longrightarrow \intff{0}{1}$ is a $\mc{C}^\infty$ function such that $\th(x) = 0$ for $x < 1$ and $\th(x) =1 $ for $x > 1$.

\vspace{2mm}

By making the region of the "rounding" sufficiently small we can ensure the new $\mc{C}^\infty$ contour $\Gamma_{\e{eq}}$ obtain by this procedure has the properties that
$\Gamma_{\e{eq}} \setminus \mathrm{supp} \, \mu_{ \Gamma_{\e{eq}} } \subset \Lambda$ and $\gamma(\intff{-\epsilon}{1+\epsilon}) \subset \Gamma_{\e{eq}}$.
By Corollary \ref{deformation} our new contour $\Gamma$ has the same equilibrium measure as the contour solving the max-min energy problem,
hence our new contour also solves the max-min energy problem. This completes the proof of Proposition \ref{contourexistence}.
\begin{remark} The fact that we are in the off-critical regime is what allows us to make the contour $\mc{C}^\infty$.
If we were in the critical regime we would only generally expect to be able to make our contour piecewise $\mc{C}^\infty$.
\end{remark}

\section{Fredholm determinant}\label{fredholmsection}

In this section we prove the following important proposition.
\begin{proposition}\label{Fredholm}
Let $\{ \psi(x) \}_{x \in \mathbb{R}}$ be a family of complex-valued random variables, viewed as a random measurable function, with the following properties.
\begin{enumerate}[label=(\arabic*)]
\item $\overline{\psi(x)} = \psi(-x)$ and $\psi(0) = 0$. 
\item Any finite subset of $\{ \Re  \psi(x),  \Im  \psi(x)\}_{x > 0}$ is jointly Gaussian (not necessarily centred).
\item For every $k \geq 0$, $\mathbb{E}\big[ |\psi(x)|^k \big] \, \leq \, C(1+|x|^{b_k})$ for some $b_k \geq 0$.
\item $\mathbb{E}\Big[ \Int{\mathbb{R}}{} |\psi(x)|^2 \, \dd x \Big] \, <  \, +\infty$,
which implies that $\psi \in L^2(\mathbb{R})$ almost surely. 
\item $B(x,y) \overset{\mathrm{def}}{=} \mathrm{cov}\Big( \psi(x), \overline{\psi(y)} \Big)$ is such that $\Int{\mathbb{R}^2}{} |B(x,y)|^2 \, \dd x \, \dd y\,  <\,  +\infty$.
\end{enumerate}
Define $\mu \in L^2(\mathbb{R})$ by $\mu(x) \overset{\mathrm{def}}{=} \mathbb{E} \big[ \psi(x) \big] $ and let  $A: \mathbb{R}^2 \longrightarrow \mathbb{C}$ be a continuous function such that
\beq
\overline{A(x,y)}  \, = \,  A(y,x)  \, = \,  A(-x,-y)  \qquad \mathrm{and}   \qquad \Int{ \mathbb{R}^2 }{} |A(x,y)|^2 \, \dd x \, \dd y  \, < \,  +\infty \, ,
\nonumber
\enq
and let $\lambda \in L^2(\mathbb{R})$ be such that $\overline{\lambda(x)} = \lambda(-x)$

Then,  $\mathsf{A} : L^2(\mathbb{R}) \longrightarrow L^2(\mathbb{R})$, resp. $\mathsf{B} : L^2(\mathbb{R}) \longrightarrow L^2(\mathbb{R})$,
defined by
\beq
\mathsf{A}f(x) = \Int{\mathbb{R}}{} A(x,y) f(y) \, \dd y \;, \quad \mathrm{resp.} \quad \mathsf{B} f(x) = \Int{\mathbb{R}}{} B(x,y) f(y) \, \dd y
\nonumber
\enq
are self-adjoint Hilbert-Schmidt operators and one has

\begin{equation}\label{fredholmformula}
\begin{split}
&\mathbb{E}\Bigg[ \exp\bigg( \frac{\i}{2} \Int{\mathbb{R}^2}{} \overline{\psi(x)}A(x,y) \psi(y) \, \dd y + \i \Int{\mathbb{R}}{} \overline{\lambda(x)} \psi(x) \, \dd x\bigg)  \Bigg] \\
&= \frac{\exp( \frac{1}{2} \i \langle \mu,  \mathsf{A} (\mathsf{1} - \i \mathsf{BA})^{-1} \mu \rangle_{L^2(\mathbb{R})} +  \i \langle \lambda,
(\mathsf{1}-\i\mathsf{BA})^{-1}\mu \rangle_{L^2(\mathbb{R})} - \frac{1}{2}\langle \lambda,   (\mathsf{1}-\i \mathsf{BA})^{-1} \mathsf{B} \lambda\rangle_{L^2(\mathbb{R})})   }
                            {\sqrt{\det_{L^2(\mathbb{R})}(\mathsf{1}- \i \mathsf{B  A})}}
\end{split}
\end{equation}
\end{proposition}
We prove Proposition \ref{Fredholm} in two stages.  We first prove the result in the finite rank case and then extend by continuity.

\begin{lemma}\label{finitefredholm}
Denote $\mc{S}_n \, = \, \intn{-n}{n} \setminus \{0\}$ and let
\begin{itemize}
\item $\bs{\xi} = \{ \xi_k \}_{k \in\mc{S}_n }\in \Cx^{2n}$ satisfy $\overline{\xi_j} =  \xi_{-j}$ with $\left\{ \Re \xi_k, \Im \xi_k \right\}_{k=1}^n$ a Gaussian random variable;
\item $\bs{\lambda} = ( \lambda_k )_{k \in \mc{S}_n}$ be such that $\overline{\lambda_k} = \lambda_{-k}$;
\item $A=  \e{Mat}\big( A_{ij}   \, : \, i,j \in \mc{S}_n  \big) $ be such that  $A_{ij}=A_{-j,-i}$.
\end{itemize}

\noindent Then, one has
\begin{multline}
\mathbb{E}\Bigg[ \exp\bigg( \frac{\i}{2}  \big< \bs{\xi}, A \bs{\xi} \big>  + \i \big< \bs{\lambda}, \bs{\xi} \big> \bigg) \Bigg] \\
= \frac{\exp( \frac{1}{2} \i \langle \bs{\mu},  A (\mathsf{1} - \i BA)^{-1} \bs{\mu} \rangle +  \i \langle \bs{\lambda},  (\mathsf{1} - \i B A )^{-1}\bs{\mu} \rangle
- \frac{1}{2}\langle \bs{\lambda},   (\mathsf{1} - \i BA)^{-1} B \bs{\lambda} \rangle)  }
                                            {  \sqrt{ \det(\mathsf{1}- \i B  A)} }
\end{multline}
where $B  = \e{Mat}\big( \mathrm{cov}(\xi_i, \overline{\xi_j})   \, : \, i,j \in \mc{S}_n  \big)  $, $\bs{\mu}  = (\mathbb{E} [ \xi_i ] )_{i \in \mc{S}_n }$
and $\langle \bs{x}, \bs{y} \rangle = \sum_{ i \in \mc{S}_n } \overline{x_i} y_i$. Finally, the quantity in the denominator corresponds to the determination of the square roots of the map
$z \mapsto \det(\mathsf{1}-  z \i B  A)$ in $\Re(z)>0$ equal to $1$ at $z=0$.
\end{lemma}
\begin{proof}
To show this we transform to real variables. Let
\beq
X_k \, = \,  \frac{\xi_k + \xi_{-k}}{\sqrt{2}} \;\; \e{for} \; \; k > 0  \quad \e{and} \quad
X_k  \, = \,  \frac{\xi_{-k} - \xi_{k}}{\sqrt{2} \i }  \;\; \e{for} \; \; k < 0\; ,
\enq
which can be encapsulated  in terms of the matrix transformation $\bs{X} = J \bs{\xi}$ with $J$ unitary in respect to $<.,.>$.
Upon introducing $\widetilde{B} = J B J^\dagger$ , one gets $\mathrm{cov}(X_i,X_j) = \widetilde{B}_{ij}$.
Similarly let $\widetilde{A} = J A J^\dagger$,  $\widetilde{\bs{\lambda}} = J \bs{\lambda} \in \mathbb{R}^{2n}$ and $\widetilde{\bs{\mu}} = \mathbb{E}[ \bs{X} ] = J \bs{\mu} \in \mathbb{R}^{2n}$.
Then the expectation of interest reduces to the evaluation of the integral
\begin{align*}
\frac{\sqrt{\det(\widetilde{B}^{-1})}}{(2\pi )^n} \Int{ \mathbb{R}^{2n} }{ }
\ex{ - \frac{1}{2}(\bs{x}-\widetilde{\bs{\mu}})^\mathsf{T} \widetilde{B}^{-1}(\bs{x}-\widetilde{\bs{\mu}}) + \frac{1}{2} \i \bs{x}^\mathsf{T} \widetilde{A} \bs{x} +
\i \widetilde{ \bs{\lambda} }^\mathsf{T} \bs{x} } \, \dd \bs{x}
\end{align*}
A short calculation involving completing the square then yields the desired formula.
To conclude, one observes that since $B$ is positive definite, $\sqrt{B} A \sqrt{B}$ is a Hermitian matrix and thus has real eigenvalues $\la_a$, $a=1, \dots 2n$ counted with multiplicities. Hence,
\beq
\det(\mathsf{1}-  z \i B  A) \, = \, \det(\mathsf{1}-  z \i \sqrt{B} A \sqrt{B}) \, = \, \pl{a=1}{2n} \big( 1 - \i z \la_a \big) \, \not = \, 0
\enq
if $\Re(z)>0$. The function thus admits a determination of its square root in that simply connected domain which, moreover, may be chosen to be equal to $1$
at $z=0$. The non-vanishing of this determinant also implies the invertibility, in that domain, of the matrix $\mathsf{1}-  z \i B  A$.

\end{proof}

The last step of the proof consists in demonstrating that we can approximate $\mathsf{A}$ by a finite rank operator.
Let $\varphi_k(x) = H_k(x) e^{-\frac{1}{2}x^2}$ where $H_k$ is the $k$th normalised Hermite polynomial, so that
\beq
\big< \varphi_k , \varphi_j \big>_{L^2(\R)}  \, = \, \Int{\mathbb{R}}{}\overline{\varphi_k(x)}\varphi_j(x)\, \dd x \, = \,  \delta_{jk}\, .
\enq
Note that $\varphi_k$ is an even function for even $k$ and odd for odd $k$, hence $\langle \varphi_k , \psi \rangle_{L^2(\mathbb{R})}$ is purely real for $k$ even and purely imaginary for $k$ odd. Let
\beq
\sigma_k  \, = \,   \frac{1}{\sqrt{2}}\varphi_{2k} +  \frac{1}{\sqrt{2}} \varphi_{2k+1} \quad  \e{and} \quad  \sigma_{-k} \, = \,   \frac{1}{\sqrt{2}}\varphi_{2k} -  \frac{1}{\sqrt{2}} \varphi_{2k+1}
\quad \e{for} \quad  k\in [\![1, n]\!] \, .
\nonumber
\enq
Then given $\mc{S}_n$ as in Lemma \ref{finitefredholm} and $\sigma_k^\ast$ being the one-form canonically associated with $\sg_k$ by means of the $L^{2}(\R)$ scalar product, we
define the orthogonal projection $\mathsf{P}_n = \sum_{k\in \mc{S}_n} \sigma_k \otimes \sigma_k^\ast$.

\begin{lemma}
The real and imaginary parts of $\xi_k = \langle \sigma_k, \psi \rangle_{L^2(\mathbb{R})}$ are jointly Gaussian with
\beq
\mathbb{E} [\xi_k] = \langle \sigma_k, \mu\rangle_{L^2(\mathbb{R})}  \qquad  \mathrm{and}  \qquad \mathrm{cov}\, (\xi_i, \overline{\xi_j}) =\langle \sigma_i ,  \mathsf{B}  \sigma_j \rangle_{L^2(\mathbb{R})} \; .
\nonumber
\enq
\end{lemma}
\begin{proof}
Let $\alpha = (\alpha_k)_{k \in  \mc{S}_n }$ with $\overline{\alpha_k} = \alpha_{-k}$. We will show that the moments of $ \sum_{k \in \mc{S}_n} \overline{\alpha_k}\xi_k$
coincide with those of a Gaussian random variable and hence $\{ \Re \xi_k, \Im \xi_k\}_{k=1}^n$ are jointly Gaussian.
Thus, let $f =  \sum_{k \in \mc{S}_n} \alpha_k \sigma_k$ , so that $\langle f, \psi \rangle_{L^2(\mathbb{R})}=\sum_{k \in \mc{S}_n} \overline{\alpha_k}\xi_k$.
Then, for $m \in \mathbb{N}$,
\begin{align*}
\mathbb{E} \bigg[\Big( \sum_{k \in \mc{S}_n} \overline{\alpha_k}\xi_k \Big)^m \bigg] \, =\,
\mathbb{E} \Big[ \langle f, \psi \rangle_{L^2(\mathbb{R})}^m \Big]
\, =  \, \mathbb{E}\bigg[ \Int{\mathbb{R}^m}{} \overline{f(x_1)} \cdots \overline{f(x_m)} \psi(x_1) \cdots \psi(x_n) \, \dd\mathbf{x} \bigg]
\end{align*}
Because $f$ is a Schwarz function and $\mathbb{E}\big[ |\psi(x)|^m \big]$ grows at most algebraically as $|x| \to \infty$,
we may exchange the expectation with the integral sign. Then applying Wick's theorem and reversing our steps we find that
$ \sul{k\in \mc{S}_n}{} \overline{\alpha_k}\xi_k $ has the moments of a Gaussian random variable with the predicted variance, and hence is Gaussian.
\end{proof}

\vspace{2mm}

\begin{proof}[Proof of \ref{Fredholm}]
Let $\mathsf{A}^{(n)} = \mathsf{P}_n \mathsf{A} \mathsf{P}_n$, $\mathsf{B}^{(n)} = \mathsf{P}_n \mathsf{B} \mathsf{P}_n$, $\lambda^{(n)} = \mathsf{P}_n \lambda$.
By Lemma \ref{finitefredholm}, since the associated finite-size matrices $A^{(n)}, B^{(n)}$ and vectors $\bs{\la}^{(n)}, \bs{\mu}^{(n)}$ enjoy the hypotheses of that lemma, we have
\begin{equation}\label{finiterank}
\begin{split}
&\mathbb{E} \bigg[ \exp( \frac{\i}{2} \langle \psi, \mathsf{A}^{(n)} \psi \rangle_{L^2(\mathbb{R})} + \i\langle \lambda^{(n)} , \psi \rangle_{L^2(\mathbb{R})}) \bigg] \\
&= \frac{ \ex{ \frac{1}{2} \i \langle \mu,  \mathsf{A}^{(n)} (\mathsf{1} - i \mathsf{B}^{(n)} \mathsf{A}^{(n)} )^{-1} \mu \rangle_{L^2(\mathbb{R})}
+  \i \langle \lambda^{(n)} ,  (\mathsf{1}-\i\mathsf{B}^{(n)} \mathsf{A}^{(n)} )^{-1} \mu \rangle_{L^2(\mathbb{R})}
- \frac{1}{2}\langle \lambda^{(n)},   (\mathsf{1}-\i \mathsf{B}^{(n)} \mathsf{A}^{(n)} )^{-1} \mathsf{B}^{(n)} \lambda^{(n)} \rangle_{L^2(\mathbb{R})} } }
            {  \sqrt{\det_{L^2(\mathbb{R})}(\mathsf{1}- \i\mathsf{B}^{(n)} \mathsf{A}^{(n)})}  }
\end{split}
\end{equation}
Observing that, by construction, $\ov{B(x,y)} = B(y,x)=B(-x,-y)$, one infers from the properties of the integral kernels $A(x,y)$ and $B(x,y)$
readily ensure that $\mathsf{A}$ and $\mathsf{B}$ are Hilbert-Schmidt and thus $\mathsf{B  A}$ is trace class.
This ensures that the Fredholm determinant $\det_{L^2(\mathbb{R})}(\mathsf{1}- \i \mathsf{B  A})$  is well-defined.  Furthermore  $\mathsf{B} \geq 0$ and so has a well-defined square root $\sqrt{\mathsf{B} }$.
Then by Sylvester's identity
\beq
\det_{L^2(\mathbb{R})}(\mathsf{1}- z \i \mathsf{B  A}) = \det_{L^2(\mathbb{R})}(\mathsf{1}- z \i \sqrt{\mathsf{B}}  \mathsf{A}\sqrt{\mathsf{B}} ) \, = \,
\pl{a=1}{\infty} \big( 1-\i z \la_a) \neq 0 \, , \quad \e{since} \quad
\sqrt{\mathsf{B}}  \mathsf{A}\sqrt{\mathsf{B}}
\nonumber
\enq
is a compact self-adjoint operator with eigenvalues $\la_a$, $a \in \mathbb{N}$, repeated with their multiplicities.
This ensures that $(\mathsf{1}   -  \i\mathsf{BA})^{-1}$ is well-defined and that the holomorphic map $ z \mapsto \det_{L^2(\mathbb{R})}(\mathsf{1}- z \i \mathsf{B  A})$
on $\Re(z)>0$ admits a determination of its square root which equals $1$ when $z=0$.

Now, one can readily check that $\mathsf{A}^{(n)} \to \mathsf{A}$ and $\mathsf{B}^{(n)} \to \mathsf{B}$ in the Hilbert-Schmidt norm $\norm{\cdot}_{\mathrm{HS}}$,
and thus $\mathsf{B}^{(n)}\mathsf{A}^{(n)} \to \mathsf{B}\mathsf{A}$ in the trace class norm. It is direct to see that all $n$-dependent quantities in the
\textit{rhs} of (\ref{finiterank}) converge to the same expression with $n$ removed, \textit{i.e.} we may replace
$\mathsf{A}^{(n)}, \mathsf{B}^{(n)}, \lambda^{(n)}$ with $\mathsf{A}, \mathsf{B}, \lambda$ respectively throughout. On the left hand side of (\ref{finiterank}) we observe that
\begin{align*}
&\left| \mathbb{E}  \bigg[ \exp( \frac{\i}{2} \langle \psi, \mathsf{A}^{(n)} \psi \rangle_{L^2(\mathbb{R})} + \i \langle \lambda^{(n)} , \psi \rangle_{L^2(\mathbb{R})}) \bigg] \,- \,
\mathbb{E} \bigg[\exp( \frac{\i}{2} \langle \psi, \mathsf{A} \psi \rangle_{L^2(\mathbb{R})} + \i\langle \lambda , \psi \rangle_{L^2(\mathbb{R})}) \bigg] \right| \\
&\leq \frac{1}{2} \| \mathsf{A} - \mathsf{A}^{(n)}\|_{\mathrm{HS}}  \, \mathbb{E} \Big[ \| \psi\|_{L^2(\mathbb{R})}^2 \Big]
+ \| \lambda - \lambda^{(n)}\|_{L^2(\mathbb{R})} \sqrt{\mathbb{E} \Big[ \| \psi\|_{L^2(\mathbb{R})}^2\Big] }  \underset{n \to +\infty}{\longrightarrow} 0
\end{align*}
since $\mathbb{E} \Big[ \| \psi\|_{L^2(\mathbb{R})}^2\Big] < +\infty$ by hypothesis.
\end{proof}

\newpage

\section{Asymptotic expansion of the Selberg integral}\label{selbergexpansion}

In this section we prove Theorem \ref{selbergexpansionthm}. We begin with the exact formula for the partition function $\mathcal{Z}_{N,\Gamma_0}[V_0]$.
\begin{proposition}[Selberg-Mehta integral, Ch. 17 of \cite{mehta}] Let $V_0$ be the quadratic potential given in \eqref{definition potentiel V0}.
 Then, one has the closed formula
\beq
\mathcal{Z}_{N,\Gamma_0}[V_0]  \, = \,
 (2 \pi)^\frac{N}{2} \left( \frac{\zeta_2 - \zeta_1}{\sqrt{8N\beta}}\right)^{N\beta\left( \frac{1}{\beta}- \frac{1}{2} \right)+ \frac{\beta N^2}{2}}
 \ex{-\beta N^2 V\left(\frac{\zeta_1+\zeta_2}{2}\right)} \Big[ \Big(\frac{\beta}{2}\Big)! \Big]^{-N} \pl{j=1}{N} \Big( \frac{\beta j}{2}\Big)! \;.
\enq
\end{proposition}
The terms on the right hand side all have straightforward asymptotic expansions except for $\pl{j=1}{N} \Big( \frac{\beta j}{2}\Big)!$. For this we need the following lemma which relates it to a Barnes $G$-function.
\begin{lemma} Let $\mf{t},n \in \mathbb{N}$ be integers greater or equal to $1$. Then
\beq
\prod_{j=1}^{n-1} (\mf{t} j)! \,  = \, \mf{t}^{- \frac{n(\mf{t}-1)}{2}}  \bigg( \pl{j=1}{\mf{t} n-1 } j!  \bigg)^{\frac{1}{\mf{t}}} \cdot
\pl{p=1}{\mf{t}-1} \frac{\Gamma(\frac{p}{\mf{t}})^{1-\frac{p}{\mf{t}}}}{\Gamma(\frac{p}{\mf{t}}+n)^{1-\frac{p}{\mf{t}}}}  \;.
\nonumber
\enq
\end{lemma}
\begin{proof} Let $H(n) :=  \prod_{j=1}^{n-1} (\mf{t} j)!$ and consider the product
\begin{align*}
\prod_{j=1}^{\mf{t} n-1 } j! =  \prod_{k=0}^{n-1} \prod_{\ell = 0}^{\mf{t} - 1} ( \mf{t} k + \ell)!
\, = \, \big[  H(n) \big]^\mf{t} \cdot \prod_{k=0}^{n-1} \prod_{\ell = 0}^{\mf{t} - 1} \prod_{p = 1}^{\ell} (\mf{t} k + p)  \, = \,\big[  H(n) \big]^\mf{t} \cdot
\prod_{p=1}^{\mf{t}-1} \Big( \prod_{k=0}^{n-1} (\mf{t} k + p) \Big)^{\mf{t}-p} \;.
\end{align*}
Finally, we observe that $ \pl{k=0}{n-1} (\mf{t} k + p) \, = \,  \mf{t}^n \frac{\Gamma(\frac{p}{\mf{t}}+n)}{\Gamma(\frac{p}{\mf{t}})}\, . $
\end{proof}
From this we have

\begin{align*}
\ln \mathcal{Z}_{N,\Gamma_0}[V_0]  \, &= \,
\frac{N}{2} \ln (2 \pi) + \Big[ N\beta\Big( \frac{1}{\beta}- \frac{1}{2} \Big)+ \frac{\beta N^2}{2} \Big] \Big( \ln(\zeta_2 - \zeta_1) - \frac{1}{2}\ln(8\beta) - \frac{1}{2}\ln N \Big)  \\
&\quad -\beta N^2 V\Big(\frac{\zeta_1+\zeta_2}{2}\Big)
- N \ln \Big[ \Big(\frac{\beta}{2}\Big)! \Big]  + \ln\Big[  \Big( \frac{\beta N}{2}\Big)! \Big] \\
&\quad  - \frac{N}{2} \Big[ \frac{\beta}{2}-1 \Big] \ln \frac{\beta}{2} + \frac{2}{\beta}\ln\Big[ \prod_{j=1}^{\frac{\beta  N }{2}-1}j! \Big] + \sum_{p=1}^{\frac{\beta}{2}-1} \Big( 1- \frac{2p}{\beta} \Big) \Big\{ \Gamma(\frac{2p}{\beta}) - \Gamma(\frac{2p}{\beta}+N) \Big\} \;.
\end{align*}
Using Stirling's approximation and the asymptotic expansion of the Barnes $G$-function we obtain the $N \to +\infty$ expansion,
\begin{align*}
\ln \mathcal{Z}_{N,\Gamma_0}[V_0]  \, &= \,
\frac{N}{2} \ln (2 \pi) + \Big[ N\beta\Big( \frac{1}{\beta}- \frac{1}{2} \Big)+ \frac{\beta N^2}{2} \Big] \Big( \ln(\zeta_2 - \zeta_1) - \frac{1}{2}\ln(8\beta) - \frac{1}{2}\ln N \Big)  \\
&\quad -\beta N^2 V\Big(\frac{\zeta_1+\zeta_2}{2}\Big)
- N \ln \Big[ \Big(\frac{\beta}{2}\Big)! \Big]  + \frac{\beta N}{2} \ln N+ \frac{\beta N}{2} \ln \frac{\beta}{2} - \frac{\beta N}{2} +\frac{1}{2}\ln N \\
&\quad  - \sum_{p=1}^{\frac{\beta}{2}-1} \Big( 1- \frac{2p}{\beta} \Big) \Big\{ \Big( N+\frac{2p}{\beta}-1 \Big) \ln N-  N +\frac{1}{2}\ln N \Big\} - \frac{N}{2} \Big[ \frac{\beta}{2}-1 \Big] \ln  \frac{\beta}{2} \\
&\quad + \frac{2}{\beta}\Big\{ \frac{\beta^2 N^2}{8}\ln N + \frac{\beta^2N^2}{8}\ln \frac{\beta}{2} - \frac{3\beta^2 N^2}{16} + \frac{\beta N}{4} \ln(2\pi) - \frac{1}{12}\ln N \Big\} + \e{O}(1) \;.
\end{align*}
We remark that logarithm of the Barnes $G$-function and the logarithm of the Gamma function both admit, from the $\e{O}(1)$ scale and beyond, a full asymptotic expansion in powers of $\frac{1}{N}$. Hence also the above $\e{O}(1)$ error could be expanded in a full asymptotic power series in $\frac{1}{N}$. Finally, collecting terms at each scale, and observing that
\begin{align*}
\mc{I}_{  \Ga_{0} }[\mu_{ \Ga_{0} }] &= -\ln(\zeta_2 - \zeta_1) + \ln (4) + 2 V\Big( \frac{\zeta_1 + \zeta_2}{2} \Big) + \frac{3}{4} \\
\mathrm{Ent}[\mu_{ \Ga_{0} }] &= -\Int{\Gamma_0}{} \ln \Big\{ \frac{ \dd  \mu_{\Gamma_0}(z)}{\dd z } \Big\} \cdot  \dd\mu_{\Gamma_0}(z) \, = \, -\frac{1}{2} + \ln  \frac{\pi}{2} +\ln(\zeta_2 - \zeta_1)  \, ,
\end{align*}
we obtain Theorem \ref{selbergexpansionthm}.

\bigskip
\noindent
{\bf Acknowledgments:} A.G., K.K.K.  and A. L. gratefully acknowledge the support of the ERC Project LDRAM: ERC-2019-ADG Project 884584.

\bibliography{./ComplexPotentialRef.bib}{}
\bibliographystyle{plain}

\end{document}